\documentclass[a4paper,12pt]{article}
\usepackage{amssymb}
\usepackage{amsmath}
\usepackage{amsthm}

\usepackage{graphicx,color}
\usepackage{subfigure}

\def\openone{\leavevmode\hbox{\small1\kern-3.8pt\normalsize1}}
\def\tr{{\rm Tr}}

\def\ca{{\cal A}}

\def\cx{{\cal X}}

\newcommand{\be}{\begin{equation}}
\newcommand{\ee}{\end{equation}}
\newcommand{\bea}{\begin{eqnarray}}
\newcommand{\eea}{\end{eqnarray}}
\newcommand{\beann}{\begin{eqnarray*}}
\newcommand{\eeann}{\end{eqnarray*}}

\def\RR{\mathbb{R}}
\def\ZZ{\mathbb{Z}}

\def\NN{\mathbb{N}}
\def\CC{\mathbb{C}}
\def\HH{\mathbb{H}}
\def\KK{\mathbb{K}}
\def\DD{\mathbb{D}}
\def\LL{\mathbb{L}}

\newtheorem{theorem}{Theorem}
\newtheorem{lemma}{Lemma}
\newtheorem{proposition}{Proposition}
\newtheorem{corollary}{Corollary}

\newtheorem{remark}{Remark}
\newtheorem{definition}{Definition}

\title{Dissipative entanglement of quantum spin fluctuations}

\date{\null} 

\author{F. Benatti$^{1,2}$, F. Carollo$^{1,2}$, R. Floreanini$^2$\\
\small $^1$Dipartimento di Fisica, Universit\`a di Trieste, Trieste, 34151 Italy\\
\small ${}^2$Istituto Nazionale di Fisica Nucleare, Sezione di Trieste, 34151 Trieste, Italy}

\begin{document}

\maketitle

\begin{abstract}
\noindent
We consider two non-interacting infinite quantum spin chains immersed in a common thermal environment and undergoing a local dissipative dynamics of Lindblad type.
We study the time evolution of collective mesoscopic quantum spin fluctuations 
that, unlike macroscopic mean-field observables, retain
a quantum character in the thermodynamical limit.
We show that the microscopic dissipative dynamics is able to entangle these 
mesoscopic degrees of freedom, through a purely mixing mechanism.
Further, the behaviour of the dissipatively generated quantum correlations
between the two chains is studied as a function of temperature
and dissipation strength.
\end{abstract}

\section{Introduction}

The presence of an external environment typically affects quantum systems in weak interaction with it via loss of quantum correlations due to decohering and mixing-enhancing effects \cite{Alicki}-\cite{Benatti1}.
Nevertheless, it has also been established that suitable
environments are capable of creating and enhancing quantum entanglement among quantum 
open sub-systems immersed in them instead of destroying it \cite{Plenio}-\cite{Benatti2}.
It is remarkable that entanglement can be generated solely by the mixing structure
of the irreversible dynamics, without any environment induced, direct interaction between the quantum sub-systems.

This mechanism of environment induced entanglement generation
has been studied for systems made of few qubits 
or oscillator modes \cite{Benatti1},\cite{Benatti2}-\cite{Benatti4} and specific protocols have been proposed
to prepare predefined entangled states via the action of suitably
engineered environments \cite{Kraus}. 
Instead, in this paper, we study the possibility that entanglement be created 
through a purely noisy mechanism in many-body systems (for different approaches to entanglement in many-body systems, see \cite{Lewenstein}-\cite{Modi}
and references therein).

In a quantum system made of a large number $N$ of constituents, 
typical accessible observables are collective ones, {\it i.e.} those
involving the degrees of freedom of all its elementary parts.
For these ``macroscopic'' observables, one usually expects that quantum effects 
fade away as $N$ becomes large, even more so when
the many-body system is in contact with an external environment.
This is surely the case for the so-called ``mean field'' observables,
{\it i.e.} averages of microscopic
operators; these quantities scale as $1/N$
and as such behave as classical observables when the number
of system constituents becomes large.

Nevertheless, other collective observables exist that scale as
$1/\sqrt{N}$ and that might retain some quantum properties
as $N$ increases \cite{Goderis1}-\cite{Matsui}. These observables have been called
``fluctuation operators'' and shown to obey a quantum central limit theorem.
In the large $N$ limit, the microscopic fluctuation operators form a bosonic algebra, irrespective of the nature of the microscopic many-body system. 
Being half-way between microscopic observables 
(as for instance the individual spin operators in a generic spin systems) and truly macroscopic ones 
({\it e.g.} the corresponding mean magnetization),
the fluctuation operators have been named ``mesoscopic''.
They provide a particularly suited scenario to look for truly quantum signals
in the dynamics of ``large'' systems, {\it i.e.} in systems 
in which the number of microscopic constituents grows arbitrarily.

Although the emergent time-evolution over the fluctuation algebra has
been extensively studied in many systems \cite{Verbeure}, very little is known
of its behaviour in open many-body systems, {\it i.e.} in systems 
immersed in an external bath.
This is the most common situation encountered in actual experiments, typically
involving cold atoms, optomechanical or spin-like systems \cite{Bloch,Aspelmeyer,Rogers}, 
that can never be thought of as completely isolated from their thermal surroundings.
Actually, the repeated claim of having detected ``macroscopic'' entanglement
in those experiments \cite{Jost,Krauter} poses a serious challenge in trying to 
interpret theoretically those results \cite{Narnhofer}.

Motivated by these experimental findings, in the following we shall show
that quantum behaviour can indeed be present at the mesoscopic level in open many-body
systems provided suitable fluctuation operators are considered.
More specifically, we focus on a many-body system composed by two spin-1/2 chains, one next to the other, which are endowed with a microscopic thermal state at inverse temperature $\beta$ with a tensor product structure, that excludes long-range correlations. A site in the system is thus composed by the corresponding couple of sites in the two chains and suitable single-site operators are considered giving rise to quantum fluctuations that, in the infinite volume limit, identify collective bosonic degrees of freedom clearly attributable to the two chains independently. 
The two chains are immersed in a common environment such that the observables supported by finite lattice intervals are subjected to a Lindblad type dynamics without direct interactions among the spins either in a same or in different chains. The dynamics is chosen in such a way to leave the microscopic state invariant and to map into itself the linear span of the relevant single-site observables. 
Under this condition, we show that the emergent, mesoscopic dissipative quantum fluctuation dynamics is capable of entangling different collective bosonic degrees of freedom and that the dissipatively created entanglement presents interesting features as a function of the temperature  and of the microscopic coupling strength of the two chains~\cite{note}.

The structure of the paper is as follows: Section 2 provides the necessary preliminary notions concerning quantum spin chains and their description at the mesoscopic level based on a Weyl algebra of quantum fluctuations that satisfy a quantum central limit relation as explained in Theorem~\ref{th1}. 

In Section 3, the general techniques exposed in Section 2 are 
applied to the case of a system consisting of two quantum spin $1/2$ chains in a microscopic factorized thermal state: specific microscopic operators are selected 
that give rise to collective degrees of freedom pertaining to each chain independently of the other or to both chains at the same time. The description of the resulting quantum fluctuations is given in terms of bosonic creation and annihilation operators and their mesoscopic thermal state is obtained in Proposition~\ref{prop-state}. 

In Section 4, a microscopic open quantum dynamics of the two chains is considered with a Lindblad generator that does not contain direct spin interactions and whose dissipative term statistically couples also spins belonging to different chains, while leaving the microscopic thermal state invariant. The main result of the paper is contained in Theorem~\ref{qfth} which shows that, in the large $N$ limit, the microscopic dissipative dynamics gives rise to a mesoscopic dynamics of quantum fluctuations consisting of a semigroup of completely positive Gaussian maps sending Weyl operators into Weyl operators.  The Lindblad generator of this so-called quasi-free semigroup is derived in Corollary~\ref{cor1}.  

Section 5 and 6 focus on mesoscopic Gaussian initial states whose form is left invariant by the dissipative mesoscopic dynamics. Specific Gaussian states are considered involving collective degrees of freedom that belong to the two chains, independently. They are obtained with separable squeezing operations on the mesoscopic thermal state: the resulting squeezed state is then separable with respect to the collective degrees of freedom pertaining to different chains.

In Section 7, two concrete microscopic models of open quantum spin chains are considered: in the first one, the dissipative term of the microscopic Lindblad generator is not diagonal in the site indices and consists of Kraus operators involving spins from both chains at each lattice site. Instead, in the second model the dissipative contribution is diagonal in the site indices and each site contributes with Kraus operators pertaining to only one chain.
Propositions~\ref{propo1} and~\ref{propo3} provide the precise forms of the Lindblad generators of the dissipative quasi-free semigroups.

Section 8 studies the entanglement dynamics of the initially separable squeezed states constructed in Section 6 for the two models explicitly solved in Section 7. Squeezed states are not left invariant by the emerging mesoscopic dynamics, although they remain Gaussian, so that they may develop collective entanglement between the two chains at the mesoscopic level which can be quantified by the logarithmic negativity. The temporal behaviour of such a dissipatively generated entanglement is then studied analytically and numerically for different values of temperature, squeezing parameter and dissipation strength.

\section{Quantum spin chains and their fluctuation algebra}
\label{CONSTRUCTION}

In this section, we briefly review how to construct the algebra of quantum fluctuations of a generic spin chain.

\subsection{Quantum fluctuations}

A quantum spin chain is a one-dimensional bi-infinite lattice, whose sites are indexed by an integer $j\in\ZZ$,
all supporting the same finite-dimensional matrix algebra $\ca^{(j)}=M_d(\CC)$.
Its algebraic description~\cite{Bratteli} is by means of the \textit{quasi-local} $C^*$ algebra $\ca$ obtained as an inductive limit from the strictly local 
sub-algebras $\ca_{[q,p]}=\bigotimes_{j=p}^q\ca^{(j)}$ supported by finite intervals 
$[q,p]$, with $q\leq p$ in $\ZZ$. Namely, one considers the 
algebraic union $\bigcup_{q\leq p}\ca_{[q,p]}$ and its completion with respect to the norm inherited by the local algebras.
Any operator $x\in M_d(\CC)$ at site $j$ can be embedded into $\ca$ as:
\be
x^{(j)}=\bold{1}_{j-1]}\otimes x\otimes\bold{1}_{[j+1}\ ,
\label{embed}
\ee
where $\bold{1}_{j-1]}$ is the tensor product of identity matrices at each site from $-\infty$ to $j-1$, while $\bold{1}_{[j+1}$ is the tensor product of identity matrices from site $j+1$ to $+\infty$.
Quantum spin chains are naturally endowed with the translation automorphism 
$\tau:\ca\mapsto\ca$ such that $\tau(x^{(j)})=x^{(j+1)}$.

Generic states $\omega$ on the quantum spin chain are described by positive, normalised linear functionals $\ca\ni a\mapsto\omega(a)$: they are expectation functionals that assign mean values to  all operators in $\ca$.
In the following, we shall consider translation-invariant states such that
\be
\label{transinv}
\begin{split}
\omega(a)&=\omega\big(\tau(a)\big)\hspace{34pt}\qquad\forall a\in \mathcal{A}\ ,\\
\omega(x^{(j)})=\omega(x^{(j+1)})&=\omega(x)=\tr(\rho\,x)\qquad\forall x\in M_d(\CC)\ ,
\end{split}
\ee
where $\rho$ is any density matrix in $M_d(\CC)$: it represents the evaluation of $\omega$ on single site observables. Furthermore, we shall focus upon
translation-invariant states $\omega$ that are also \textit{clustering}, namely  they do not support correlations between far away localized operators:
\be
\label{clustates}
\lim_{n\to\pm\infty}\omega\Big(a^\dag\tau^n(b)c\Big)=\omega(a^\dag\,c)\,\omega(b)\quad \forall a,b,c\in\ca\ .
\ee 

In an infinite quantum spin chain, the operators belonging to strictly local sub-algebras contribute to the microscopic description of the system. In order to move to a description based on collective observables supported by infinitely many lattice sites, a proper scaling ought to be chosen. Most often, mean-field observables are considered; these are constructed as averages of $N$ copies of a same single site observables $x$, from site $j=0$ to site $N-1$:
\be
X_N=\frac{1}{N}\sum_{k=0}^{N-1}x^{(k)}\ ,\qquad x\in M_d(\CC)\ .
\label{macro}
\ee

Given any state $\omega$ on $\ca$, the Gelfand-Naimark-Segal (GNS) construction~\cite{Bratteli} provides a representation $\pi_\omega:\ca\mapsto \pi_\omega(\ca)$ of $\ca$ on a Hilbert space $\HH_\omega$ with a cyclic vector $\vert\omega\rangle$ such that the linear span of vectors of the form $\vert\Psi_a\rangle=\pi_\omega(a)\vert\omega\rangle$ is dense in $\HH_\omega$ and
$$
\omega(b^\dag\,a\,c)=\langle \Psi_b\vert\pi_\omega(a)\vert\Psi_c\rangle\ ,\qquad a,b,c\in\ca\ .
$$

In case of a clustering state $\omega$, one can then consider the limit for $N\to\infty$ of 
$\omega\left(b^\dagger X_N\, c\right)$ where $b,c\in\ca$, obtaining
\be
\lim_{N\to\infty}\omega\left(b^\dagger X_N\, c\right)=\omega(b^\dag c)\,\omega(x)\ .
\label{MET}
\ee
Indeed, for any integer $N_0<N$ one can write:
$$
\lim_{N\to\infty} \omega\left(b^\dagger X_N\, c\right)=
\lim_{N\to\infty} \omega\Bigg( b^\dagger \bigg( \frac{1}{N} \sum_{k=0}^{N_0} x^{(k)} 
+ \frac{1}{N} \sum_{k=N_0+1}^{N-1} x^{(k)}\bigg)\, c\Bigg)\ .
$$
The first contribution in the r.h.s. clearly vanishes in the large $N$ limit. Concerning the second term, since strictly local operators are norm dense in $\mathcal{A}$, without
loss of generality one can assume $c$ to have support on sites with labels $\leq N_0$, so that
one can exchange it with $\sum_{k=N_0+1}^{N-1} x^{(k)}$. Using the clustering property (\ref{clustates})
one immediately gets the result (\ref{MET}). This means that in the so-called weak operator topology,
{\it i.e.} under the state average, $X_N$ converges to a scalar multiple of the identity operator:
\be
\lim_{N\to\infty} X_N = \omega(x)\, {\bf 1}\ .
\ee
Furthermore, in Appendix A it is proved that, given $x,y\in M_d(\CC)$, the product $X_NY_N$ of the mean-field-observables  weakly converges to $\omega(x)\omega(y)$: 
\be
\label{macro1}
\lim_{N\to\infty}\omega\bigg(a^\dag X_N\,Y_N\,b\bigg)=\omega(a^\dag b)\,\omega(x)\,\omega(y)\ .
\ee
It thus follows that the weak-limits of mean-field observables commute and give rise to
a commutative algebra. 
\medskip

\begin{remark}
\label{rem0}
{\rm Since they commute, mean-field observables pertain to the macroscopic, classical description level with no fingerprints of the microscopic quantum framework from which they emerge.
Instead, as outlined in the Introduction, we are interested in studying which collective observables extending over the whole spin chain may keep some degree of quantum behaviour; clearly, a less rapid scaling than $1/N$ is necessary.}
\qed
\end{remark}

\medskip

Let us then consider combinations of microscopic operators of the form:
\be
F_N(x)=\frac{1}{\sqrt{N}}\sum_{k=0}^{N-1}\left(x^{(k)}-\omega(x)\right)\ ;
\label{FL}
\ee
they are quantum analogues of the fluctuation variables in classical stochastic theory: we shall refer to them as
``local quantum fluctuations''.
Their large $N$ limit with respect to clustering states $\omega$ has been thoroughly investigated 
in~\cite{Goderis1,Verbeure} yielding a non-commutative central limit theorem and  an associated  quantum  fluctuation algebra. 

The scaling $1/\sqrt{N}$ is not sufficient to guarantee convergence in the weak-operator topology.
Nevertheless, consider $x,y\in M_d(\CC)$ such that $\left[x\,,\,y\right]=z$. Since $[x^{(j)}\,,\, y^{(\ell)}]=\delta_{j\ell}\,z^{(j)}$, with respect to a clustering state $\omega$, one has, following the same strategy used in (\ref{MET}),
\be
\lim_{N\to\infty} \omega\left(a^\dag\left[F_N(x),F_N(y)\right]b\right)=\lim_{N\to\infty} 
\frac{1}{N}\sum_{j=0}^{N-1}\omega\left(a^\dag z^{(j)}\, b\right)=\omega(a^\dag b)\,\omega(z) ,
\ee
for all $a,b\in\ca$.

Therefore, commutators $\left[F_N(x),F_N(y)\right]$ of local fluctuations do not vanish when $N\to\infty$. 
They behave as mean-field quantities and tend, in the weak-topology, to scalar quantities $\omega(z)$.
This fact indicates that, at the mesoscopic level, the emerging quantum structure is endowed with 
a non-commutative algebraic structure.
\medskip

\begin{remark}
\label{rem0a}
{\rm Because they emerge from a scaling $1/\sqrt{N}$, quantum fluctuations  provide a description level in 
between the microscopic (strictly local)  and the macroscopic (mean-field) ones. We will refer to it as to a \textit{mesoscopic} description level: though collective, it nevertheless inherits to a certain extent the 
quantum, non-commutativity  of the microscopic system from which it emerges.}
\qed
\end{remark}
\medskip

\subsection{Quantum fluctuation algebra}

In order to construct a quantum fluctuation algebra, one starts by selecting a set of $d$ linearly independent single-site microscopic observables 
$\chi=\{x_j\}_{j=1}^d$, $x_j\in M_p(\CC)$, $x_j=x_j^\dag$, and then considers their local elementary fluctuations $F_N(x_j)$ and the large $N$ limit of the expectations of polynomials in the operators $F_N(x_j)$ with respect to a clustering state $\omega$.
In particular, the observables $x_j$ are chosen such that $1)$ the coefficients
\be
\label{cormat}
C^{(\omega)}_{ij}:=\lim_{N\to\infty}\omega\big(F_N(x_i)F_N(x_j)\big)\ ,
\ee
give a well defined positive $d\times d$ correlation matrix $C^{(\omega)}$,
and $2)$ that the characteristic functions $\omega\big(e^{itF_N(x_j)}\big)$ converge to a Gaussian function in $t$ with zero mean and covariance matrix $\Sigma^{(\omega)}$ with entries  
\be
\label{covmat}
\Sigma^{(\omega)}_{ij}=\frac{1}{2}\,\lim_{N\to\infty}\omega\big(\left\{F_N(x_i)\,,\,F_N(x_j)\right\}\big)\ .
\ee 
We shall then define the following bilinear, positive and symmetric map on the real linear span 
$\mathcal{X}=\Big\{x_r=\sum_{i=1}^d r_i\, x_i,\ x_i\in\chi,\ r_i\in\mathbb{R}\Big\}$,
\be
\label{BiFo}
(x_{r_1},x_{r_2})\to (r_1,\Sigma^{(\omega)}\,r_2)=\sum_{i,j=1}^dr_{1i}\, r_{2j}\,\Sigma^{(\omega)}_{ij}\ .
\ee
\medskip

A multivariate version of the {\sl normal quantum central limit theorem} is based on a restricted class of clustering states.
\medskip
 
\begin{definition}
\label{2}
A finite set of self-adjoint operators $\chi=\{x_j\}_{j=1}^d$ is said to have ``normal multivariate quantum fluctuations'' 
with respect to a clustering state $\omega$ if the latter obeys the condition: 
\be
\label{const2}
\sum_{k=0}^{\infty}\Big|\omega(x^{(0)}_ix_j^{(k)})-\omega(x_i)\omega(x_j)\Big|<+\infty\quad\forall x_i,x_j\in\chi\ ,
\ee
and further satisfies
\bea
\label{Gauss1}
\lim_{N\to\infty}\omega\big(F_N^2(x_j)\big)&=& \Sigma^{(\omega)}_{jj}\\
\label{Gauss2}
\lim_{N\to\infty}\omega(e^{itF_N(x_j)})&=&{\rm e}^{-\frac{t^2}{2}\Sigma^{(\omega)}_{jj}}\qquad 
\forall x_j\in\chi,\ \forall\, t\in\mathbb{R}\ .	
\eea
\end{definition}
\medskip
\noindent
We expect quantum fluctuations to obey the canonical commutation relations in the limit of large $N$; then, exponentials of local fluctuations ${\rm e}^{iF_N(x_j)}$ are expected to satisfy Weyl-like commutation relations in that limit
\cite{Verbeure}. 

In full generality, given a set $\chi$ as in {\sl Definition \ref{2}}, one equips the real vector space $\mathcal{X}$ with the symplectic (bilinear) form 
\be
\label{sympform1}
(r_1,r_2)\to(r_1,\sigma^{(\omega)} r_2)=\sum_{i,j=1}^dr_{1i}\,r_{2j}\,\sigma^{(\omega)}_{ij}\ ,
\ee 
defined by the anti-symmetric matrix
$\sigma^{(\omega)}$ with entries
\be
\label{sympform}
\sigma^{(\omega)}_{ij}:=-i\lim_{N\to\infty}\omega\left(\left[F_N(x_i)\,,\,F_N(x_j)\right]\right)=-\sigma^{(\omega)}_{ji}\ .
\ee
The relation between the correlation, covariance and symplectic matrices is 
\be
\label{corcovsym}
C^{(\omega)}=\Sigma^{(\omega)}\,+\,\frac{i}{2}\sigma^{(\omega)}\ .
\ee

For sake of compactness, using the linearity of the map 
that associates an operator $x$ with its local quantum fluctuation $F_N(x)$,
the following notation will be used:
\bea
\label{qfa1}
(r\,,\,F_N)&:=&\sum_{j=1}^dr_j\,F_N(x_j)=F_N(x_r)\qquad\forall x_r\in \chi\ ,\\
\label{qfa2}
W_N(r)&:=&{\rm e}^{i(r\,,\,F_N)}={\rm e}^{iF_N(x_r)}\ ,
\eea
where $F_N=(F_N(x_1),F_N(x_2),\ldots, F_N(x_d))^{tr}$ is the vector of local fluctuations.

With the aid of the symplectic matrix $\sigma^{(\omega)}$, one can construct the abstract \emph{Weyl} algebra 
$\mathcal{W}$, linearly generated by
the Weyl operators $W(r)$, $r\in\mathbb{R}^d$, obeying the relations:
\be
W^\dag(r)=W(-r)\ ,\quad
W(r_1)W(r_2)=W(r_1+r_2)\,{\rm e}^{-\frac{i}{2}(r_1,\sigma^{(\omega)} r_2)}\ .
\label{Weyl}
\ee
The following theorem specifies in which sense the large $N$ limit of the local exponentials $W_N(r)$ can be identified with Weyl operators $W(r)$ \cite{Verbeure}.
\medskip

\begin{theorem}
\label{th1}
Any set $\chi$ with normal fluctuations with respect to a clustering state $\omega$ admits a regular \emph{quasi-free} state 
$\Omega$ on a Weyl algebra $\mathcal{W}(\chi,\sigma^{(\omega)})$ such that:
\bea
\nonumber
&&\hskip-1cm
\lim_{N\to\infty}\omega\big(W_N(r_1)\,W_N(r_2)\big)=
\exp\Bigg(-\frac{\big((r_1+r_2),\Sigma^{(\omega)}\,(r_1+r_2)\big)}{2}\,-\frac{i}{2}\,\big(r_1,\sigma^{(\omega)}r_2\big)\Bigg)\\
&&\hskip 3.7cm
=\,\Omega\big(W(r_1)W(r_2)\big)\ , 
\label{quasistate}
\eea
for all $x_{r_{1,2}}\in \mathcal{X}$ .
\end{theorem}
\medskip

The regularity and quasi-free character of $\Omega$ follow from \eqref{Gauss2}; indeed, as explicitly shown
by \eqref{Gauss2}, $\Omega$ is a Gaussian state (see Section 5). In particular,
its regularity guarantees that one can write 
\be
\label{reg}
W(r)={\rm e}^{iF(x_r)}={\rm e}^{i(r,F)}\ ,\qquad (r,F)=\sum_{i=1}^dr_i\,F(x_i)\ ,
\ee 
where $F$ is an operator-valued $d$-dimensional vector with components $F(x_i)$ that are collective field operators satisfying canonical commutation relations 
\be
\left[F(x_{r_1})\,,\,F(x_{r_2})\right]=\left[(r_1,F)\,,\,(r_2,F)\right]=i\,\big(r_1,\sigma^{(\omega)} r_2\big)\ .
\label{COMSIGMA}
\ee
We shall refer to the Weyl algebra $\mathcal{W}(\chi,\sigma_\omega)$ generated by the strong-closure (in the GNS representation based on $\Omega$) of the linear span of Weyl operators as the {\sl quantum fluctuation algebra}.

\section{Spin-1/2 chains}
\label{EX}

In this section we consider two quantum spin chains whose spins do not directly interact, but are immersed into a same environment in such a way that they are subjected to a same external quantum noise and behave as open quantum systems undergoing a microscopic dissipative quantum dynamics described by a semi-group with a generator in
Kossakowski-Lindblad form.
Our aim is to study which kind of mesoscopic time-evolution emerges from a given microscopic dynamics and how it affects a suitably constructed quantum fluctuation algebra.  In particular, we shall show that, solely because of its statistical mixing properties, the noisy part of the microscopic generator may induce entanglement between the two spin chains at the mesoscopic level.

\subsection{Quantum fluctuations}

We will first focus upon the microscopic double spin chain for which we shall construct a specific fluctuation algebra without considering any dynamics.

At each site of both chains we attach the algebra $M_2(\CC)$ generated by the $2\times 2$ identity matrix 
and the Pauli matrices $\sigma_{1,2,3}$ satisfying the algebraic rules
$$
[\sigma_i\,,\,\sigma_j]=2i\epsilon_{ijk}\,\sigma_k\ .
$$ 
We shall pair sites from the two chains so that $\ca^{(k)}$ will denote the matrix algebra 
$M_4(\CC)=M_2(\CC)\otimes M_2(\CC)$ supported by the $k$-th sites of the double chain. 
The \emph{quasi-local} algebra $\ca$ describing the double chain will then be the tensor product of the quasi-local algebras of the single chains, with $a\otimes 1$ and $1\otimes a$ denoting operators pertaining to the first, respectively the second chain. 

We shall equip $\ca$ with the microscopic thermal state at inverse temperature $\beta$ 
constructed from the infinite tensor product of a same single site thermal state with Hamiltonian:
\begin{equation}
H=\frac{\eta}{2}\big(\sigma_3\otimes\bold{1}+\bold{1}\otimes \sigma_3\big)\ .
\label{MICHAM}
\end{equation}
Explicitly, one then has 
\be
a\mapsto\omega_\beta(a)=\tr_{[q,p]}
\left(\bigotimes_{k=p}^q\rho_\beta^{(k)}\, a\right)\ ,\quad\rho_\beta^{(k)}:=\frac{{\rm e}^{-\beta H^{(k)}}}{\tr\left({\rm e}^{-\beta H^{(k)}}\right)}\ ,
\label{STATE}
\ee
where $H^{(k)}$ coincides with the hamiltonian in \eqref{MICHAM} 
for all $k$ and $a$ is any operator belonging to the strictly local algebra 
$\ca_{[q,p]}\otimes\ca_{[q,p]}$ (more general translationally invariant, clustering states are discussed in \cite{BCF2}). Further, $\tr_j$, respectively $\tr_{[q,p]}$, will denote the trace with respect to the Hilbert spaces $\CC^4$, respectively $\CC^{4^{q-p+1}}$, relative to the site $j\in[p,q]$, respectively to all sites $j\in[p,q]$. 
Setting $\epsilon=\tanh\left(\beta\eta/2\right)$, the only non-vanishing single site expectations are:
\bea
\label{thermexp1}
\omega_\beta\Big(\sigma^{(j)}_3\otimes 1\Big)&=&\omega_\beta\Big(1\otimes \sigma^{(j)}_3\Big)=\frac{\tr\left({\rm e}^{-(\beta\eta/2)\,\sigma_3}\, \sigma_3\right)}{2\cosh(\beta\eta/2)}=-\,\epsilon\\
\label{thermexp2}
\omega_\beta(\sigma^{(j)}_3\otimes \sigma_3^{(k)})&=&\epsilon^2\ .
\eea 
The state $\omega_\beta$ is thus an equilibrium thermal state with respect to the hamiltonian
time-evolution automorphism $\tau_t$ of $\ca$:
namely, $\omega_\beta$ satisfies the Kubo-Martin-Schwinger (KMS) relations at inverse temperature $\beta$ given by
\be
\label{KMS1}
\omega_\beta\big(a\,\tau_t[b]\big)=\omega_\beta\big(\,\tau_{t-i\beta}[b]\,a\big)\qquad\forall\ a,b\in\ca\ .
\ee
Such a state does not support correlations between the two spin chains and manifestly obeys the clustering condition in \eqref{clustates}. 

In the following, we shall consider the quantum fluctuation algebra based upon the self-adjoint subset $\chi=\{x_j\}_{j=1}^8$ consisting of the following $4\times 4$ hermitean matrices 
\bea
\label{matrix}
&&x_1=\sigma_1\otimes\bold{1}\ ,\ 
x_2=\sigma_2\otimes\bold{1}\ ,\ x_3=\bold{1}\otimes \sigma_1\ , \ x_4=\bold{1}\otimes \sigma_2\\
\label{matrixa}
&&x_5=\sigma_1\otimes\sigma_3\ , \ x_6=\sigma_2\otimes\sigma_3\ ,\ x_7=\sigma_3\otimes \sigma_1\ ,\ 
x_8=\sigma_3\otimes \sigma_2\ .
\eea
One easily sees that $\omega_\beta(x_j)=0$ for all $j=1,\dots,8$. Further, the conditions in 
{\sl Definition \ref{2}} are satisfied; indeed,
\be
\label{micst}
\sum_{k=0}^\infty\Big|\omega_\beta(x^{(0)}_ix^{(k)}_j)-\omega_\beta(x_i)\,\omega_\beta(x_j)\Big|
=\Big|\omega_\beta(x_ix_j)\Big|\ .
\ee

\begin{remark}
\label{remchoice}
{\rm There are $16$ single site observables of the form $\sigma_\mu\otimes\sigma_\nu$, $\mu,\nu=0,1,2,3$, $\sigma_0=\bold{1}$.
It turns out that the set of local fluctuation operators,
\be
\label{fluctexpl}
F_N(x_j)=\frac{1}{\sqrt{N}}\sum_{k=0}^{N-1}\Big(x^{(k)}_j-\omega(x_j)\Big)
=\frac{1}{\sqrt{N}}\sum_{k=0}^{N-1} x^{(k)}_j\ ,
\ee 
corresponding to the chosen subset $\chi$, gives rise to a set of mesoscopic bosonic operators $F(x_j)$, $1\leq j\leq 8$ whose Weyl algebra commutes with the one generated by the remaining eight elements.
Moreover, since the matrices $x_{1,2}$ and $x_{3,4}$ do refer to single sites belonging to different spin chains, they will provide collective operators associated to two different mesoscopic degrees of freedom.}
\qed
\end{remark}
\medskip

The microscopic state $\omega_\beta$ is a tensor product state and translation invariant; therefore, from \eqref{cormat}, one gets the correlation matrix $C^{(\beta)}$ with entries
\be
\label{modcorrmat}
C^{(\beta)}_{ij}=\lim_{N\to\infty}\omega_\beta\Big(F_N(x_i)F_N(x_j)\Big)
=\tr(\rho_\beta\, x_i\,x_j)\ .
\ee
with $\rho_\beta$ as in (\ref{STATE}).
The explicit form of this $8\times 8$ matrix is given in Appendix B; it can be expressed as a three-fold tensor products of $2\times 2$ matrices:
\be
C^{(\beta)}=\left(\bold{1}-\epsilon\,\sigma_1\right)\otimes\bold{1}\otimes\left(\bold{1}+
\epsilon\,\sigma_2\right)\ .
\label{modcorrmat2}
\ee
In computing tensor products, we adopt the convention in which 
the entries of a matrix are multiplied by the matrix to its right.

According to the preceding section, the algebraic relations among the emerging mesoscopic operators $F(x_j)$ are described by the symplectic matrix with entries 
$\sigma^{(\beta)}_{ij}=-i\tr\big(\rho_\beta\,[x_i\,,\,x_j]\big)$,
\be
\label{COMM1}
\sigma^{(\beta)}=-2i\epsilon(\bold{1}-\epsilon\sigma_1)\otimes\bold{1}\otimes\sigma_2
\ee
and by the covariance matrix with entries $
\Sigma^{(\beta)}_{ij}=\frac{1}{2}\,\tr\big(\rho_\beta\left\{x_i\,,\,x_j\big\}\right)$,
\be
\label{covmat2}
\Sigma^{(\beta)}=\frac{1}{2}\left(C^{(\beta)}+(C^{(\beta)})^{tr}\right)
=(\bold{1}-\epsilon\sigma_1)\otimes \bold{1}\otimes \bold{1}\ ,
\ee
where $tr$ means matrix transposition.
Notice that the symplectic matrix $\sigma^{(\beta)}$ is invertible; explicitly one finds:
\be
\label{invCOMM1}
(\sigma^{(\beta)})^{-1}=\frac{1}{2c^2\epsilon}\left(\bold{1}+\epsilon\sigma_1\right)\otimes\bold{1}\otimes\,i\sigma_2\ ,
\qquad c=\sqrt{1-\epsilon^2}\ .
\ee

The fluctuation algebra $\mathcal{W}(\chi,\sigma^{(\beta)})$ is then obtained from the linear span of exponential operators of the form (see the discussion after {\sl Theorem \ref{th1}})
\be
\label{Weyl1}
W(r)={\rm e}^{iF(x_r)}={\rm e}^{i\sum_{j=1}^8 r_j\,F(x_j)}={\rm e}^{i(r\,,\,F)}\ ,
\qquad x_r=\sum_{j=1}^8r_j\,x_j\ ,
\ee
where the vector $r$ is now eight dimensional, $r=(r_1,\ldots,r_8)^{tr}\in\RR^8$, while $F$ is the eight-dimensional operator valued vector with components $F(x_j)$, $1\leq j\leq 8$.
The mesoscopic Weyl operators arise from limits of microscopic exponential operators 
\bea
\label{Weyl8}
W_N(r)&:=&{\rm e}^{iF_N(x_r)}={\rm e}^{i(r\,,\,F_N)}\\
\label{Weyl8a}
(r\,,\,F_N)&:=&\sum_{j=1}^8r_j\,F_N(x_j)=F_N(x_r)\ ,
\eea
where $F_N=\{F_N(x_j)\}_{j=1}^8$ is the vector of local fluctuations. From \eqref{Weyl} and \eqref{COMM1}, one has:
\be
\label{displ}
W(r)\,F(x_i)\,W^\dag(r)=F(x_i)\,+\,i\big[(r,F)\,,\,F(x_i)\big]=F(x_i)\,+\,\sum_{j=1}^8
\sigma^{(\beta)}_{ij}\,r_j\ .
\ee

\subsection{Fluctuation algebra}

The Weyl algebraic structure associated with the chosen set $\chi$ and the thermal state $\omega_\beta$ allows for the mesoscopic description to be formulated in terms of four-mode bosonic annihilation and creation operators 
$a^\#_i\equiv(a_i,\, a_i^\dagger)$, $1\leq i\leq 4$,
satisfying the canonical commutation relations
\be
\label{coomodea}
[a_i\,,\,a^\dag_j]=\delta_{ij}\ ,\quad [a_i\,,\,a_j]=[a_i^\dag\,,\,a^\dag_j]=0\ .
\ee
Indeed, one can write 
\be
\label{Weyl2}
F(x_i)=a(f_i)+a^\dag(f_i)\ ,\quad a^\dag(f_i)=\sum_{j=1}^4\,[f_{i}]_j\,a^\dag_j\ ,\ 1\leq i\leq 8\ ,
\ee
by means of the following four-dimensional vectors $f_i\in\CC^4$, with components
\bea
\label{Weyl3}
&&\hskip -1.2cm
f_1=\sqrt{\epsilon}\begin{pmatrix}
1\cr0\cr0\cr0
\end{pmatrix}\, ,\qquad
f_2=-i\,f_1\, ,\; 
f_3=\sqrt{\epsilon}\begin{pmatrix}
0\cr0\cr1\cr0
\end{pmatrix}\, ,\qquad
f_4=-i\,f_3\\ 
\label{Weyl4}
&&\hskip-1.2cm
f_5=\sqrt{\epsilon}\begin{pmatrix}
-\epsilon\cr\sqrt{1-\epsilon^2}\cr0\cr0
\end{pmatrix}\, ,\qquad
f_6=-i\,f_5\, ,\qquad 
f_7=\sqrt{\epsilon}\begin{pmatrix}
0\cr0\cr-\epsilon\cr \sqrt{1-\epsilon^2}
\end{pmatrix}\, ,\qquad
f_8=-i\,f_7\ .
\eea
It follows that
\be
\label{Weyl5}
\big[F(x_i)\,,\,F(x_j)\big]=2\,i\,\mathcal{I}m\left((f_i,f_j)\right)\ ,\quad
(f_i,f_j)=\epsilon\,\Sigma^{(\beta)}_{ij}\,+\,\frac{i}{2}\sigma^{(\beta)}_{ij}\ .
\ee
Setting 
\be
\label{anncrop}
a=(a_1,a_2,a_3,a_4)^{tr}\ ,\quad  a^\dag=(a^\dag_1,a^\dag_2,a^\dag_3,a^\dag_4)^{tr}\ ,\quad
A=(a,a^\dag)^{tr}\ ,
\ee  
one has 
\be
\label{matrix1}
F=\mathcal{M}\,A\ ,\quad \mathcal{M}=\begin{pmatrix}f_1^\dag&f_1^{tr}\cr
\vdots&\vdots\\
f_8^\dag&f_8^{tr}\end{pmatrix}\ ,
\ee
where $f_i^\dag=(f^*_{i1},f^*_{i2},f^*_{i3},f^*_{i4})$, $f^{tr}_i=(f_{i1},f_{i2},f_{i3},f_{i4})$.
The $8\times 8$ matrix $\mathcal{M}$ can be inverted and used to write $A=\mathcal{M}^{-1}F$.
The explicit expressions of $\mathcal{M}$ and $\mathcal{M}^{-1}$ are reported in Appendix B.

From the structure of $\mathcal{M}^{-1}$, one notices that the creation and annihilation operators $a^\#_1$, respectively $a^\#_3$ come from single site operators $x_{1,2}$, respectively $x_{3,4}$ pertaining to the first, respectively the second chain.
Then, $a^\#_1$ and $a^\#_3$ describe two independent mesoscopic degrees of freedom emerging from different chains. Instead, $a^\#_{2}$ and $a^\#_4$ result from combinations of spin operators involving both chains at the same time.
\medskip

\begin{remark}
\label{rem5}
{\rm If the temperature vanishes, {\it i.e.} $\epsilon=1$, 
the non vanishing purely imaginary entries in $C^{(\beta)}$  are all 
proportional to $\pm 1$ (see (\ref{COMM2})). In such a degenerate case, only two bosonic modes can be accommodated:
\be
\label{inverse0}
a_1^{\dagger}=\frac{ F(x_1)\,+\,i\,F(x_2)}{2}\ ,\quad
a_2^\dagger= \frac{F(x_3)\,+\, i\,F(x_4)}{2}\ .
\ee
This degeneracy is due to a so-called coarse graining effect \cite{Verbeure} which forbids distinguishing the mesoscopic limits of some different fluctuation operators. In other terms, it may happen that
$$
\lim_{N\to\infty}\omega\Big(\big[F_N(x_{r_1})-F_N(x_{r_2})\big]^2\Big)=0\ ,
$$
even when $x_{r_1}\neq x_{r_2}$.}
\qed
\end{remark}
\medskip

In the creation and annihilation operator formalism, the Weyl operators become displacement operators $D(z)$ labeled by complex vectors $z\in\CC^4$.
Let $Z=(z,z^*)^{tr}\in\CC^8$ and $\Sigma_3$ denote the diagonal $8\times 8$ matrix
$\hbox{diag}(1,1,1,1,-1,-1,-1,-1)$; then,
\be
\label{displ1}
D(z):={\rm e}^{-(Z,\Sigma_3\,A)}=\exp\left(\sum_{j=1}^4\Big(z_j\,a^\dag_j-z^*_j\,a_j\Big)\right)\ .
\ee

\begin{lemma}
\label{lemma0}
Given the creation and annihilation operators $a_i^\#$, $1\leq i\leq 4$, Weyl and displacement operators are related by  
\bea
\label{Weyl6}
W(r)&=&{\rm e}^{i(r,F)}=D(z_r)\ ,\qquad
Z_r=\begin{pmatrix}z_r\cr z^*_r\end{pmatrix}=i\Sigma_3\,\mathcal{M}^\dag\,r\\
\label{Weyl6b}
D(z)&=&W(r_z)\ ,
\quad r_z=-i(\mathcal{M}^\dagger)^{-1}\Sigma_3\,Z_r\ .
\eea  
\end{lemma}
\medskip

According to Theorem \ref{th1}, the mesoscopic algebra $\mathcal{W}(\chi,\sigma^{(\beta)})$ inherits a regular quasi-free state from the microscopic state $\omega_\beta$.
\medskip

\begin{proposition}
\label{prop-state}
The quasi-free state $\Omega_\beta$ on the Weyl algebra of quantum fluctuations  $\mathcal{W}(\chi,\sigma^{(\beta)})$ is such that 
\be
\label{qfs1}
\Omega_\beta(W(r))=\exp\Big(-\frac{1}{2}(r\,,\Sigma^{(\beta)}\,r)\Big)\ ,
\ee
with covariance matrix $\Sigma^{(\beta)}$ given by \eqref{covmat2}. In the creation and annihilation operator formalism, it amounts to the expectation functional
$\displaystyle \Omega_\beta(W)=\tr(R_\beta\,W)$, where
\be
\label{qfs2}
R_\beta=\frac{{\rm e}^{-\beta\, K}}{\tr\left({\rm e}^{-\beta\,K}\right)} 
\ ,\quad K=\eta\sum_{j=1}^4a^\dag_j a_j\ ,
\ee
namely to a $KMS$ state at inverse temperature $\beta$  
with respect to the group of automorphisms generated by quadratic hamiltonian $K$.
\label{StateBETA}
\end{proposition}

\noindent
\begin{proof}
The tensor product structure and translation-invariance of $\omega_\beta$ yield 
\beann
\omega_\beta\left(W_N(r)\right)&=&\left(\tr\left(\rho_\beta\,{\rm e}^{i/\sqrt{N}\sum_{j=1}^8r_j\,x_j
}\right)\right)^N\\
&=&\left(1-\frac{1}{2N}\sum_{i,j=1}^8r_i r_j\tr(\rho\,x_i\,x_j)\,+\,o\left(\frac{1}{N}\right)\right)^N\ ,
\eeann
whence, since $r\in\RR^8$,
$$
\lim_{N\to\infty}\omega_\beta\big(W_N(r)\big)=
\lim_{N\to\infty}\omega_\beta\left({\rm e}^{i(r\,,\,F_N)}\right)=
\exp\Big(-\frac{1}{2}(r\,,\Sigma^{(\beta)}\,r)\Big)\ .
$$
On the other hand, writing $W(r)$ as a displacement operator $D(z_r)$, from \eqref{Weyl6}, 
its expectation with respect to the KMS state $\Omega_\beta$ reads
$$
\Omega_\beta(W(r))=\exp\Big(-\frac{\|Z_r\|^2}{4\epsilon}\Big)=\exp\Big(-\frac{\sum_{i,j=1}^8r_ir_j\,(f_i,f_j)}{4\epsilon}\Big)\ .
$$ 
Then, the result follows from \eqref{Weyl5}.
\qed
\end{proof}

\section{Dissipative mesoscopic dynamics}
\label{DDTISC}

Once the algebra of quantum fluctuations is constructed, an important issue is what kind of dynamics emerges at the mesoscopic level from a given microscopic time-evolution.
So far, only unitary microscopic dynamics have been considered and these have given rise to quasi-free mesoscopic unitary time-evolutions \cite{Verbeure}. 

Instead, in the following we shall focus upon the double quantum spin chain introduced before, undergoing an irreversible dissipative microscopic dynamics due to the presence of a common environment to which the chains are weakly coupled.
This setting is typical of open quantum systems, so that the double chain will be affected by decoherence due to noise and dissipation. However, quantum correlations in open systems need not only be destroyed by an environment;
if the latter is suitably engineered, entanglement can be created among two open quantum systems immersed into it by a purely statistical mixing mechanism, namely without the intervention of either direct or environment induced hamiltonian interactions \cite{Braun,Benatti2}.  

The main purpose of the following sections is twofold: on one hand, we show that, from a suitable Lindblad-type microscopic dissipative dynamics, one obtains a mesoscopic quasi-free dissipative semigroup at the fluctuation level. On the other hand, we study under which conditions the capacity of the dissipative microscopic dynamics to entangle spins belonging to different chains can persist at the mesoscopic level.

\subsection{Dissipative microscopic dynamics}

We shall study the fluctuation time-evolution emerging from a microscopic irreversible dynamics generated locally by a generator whose action on $X\in\ca_{[0,N-1]}$ is of Kossakowski-Lindblad form. More specifically,
we shall discuss dynamical equations of the following generic form:
\bea
\label{LINDMICO0a}
&&\hskip-1.2cm
\partial_tX(t)=\LL_N[X(t)]\ ,\quad \LL_N[X]=\HH_N[X]\,+\,\DD_N[X]\\
\label{LINDMICO0b}
&&\hskip -1.2cm
\HH_N[X]=i\Big[H_N\,,\,X\Big]\ ,\quad H_N=\sum_{k=0}^{N-1}\,h^{(k)}\ ,\quad H_N^\dagger=H_N\ ,\\
&&\hskip-1.2cm
\label{LINDMICO0}
\DD_N[X]=\sum_{k,\ell=0}^{N-1}J_{k\ell}\sum_{\mu,\nu=1}^d\,D_{\mu\nu}\,\Big(v_\mu^{(k)}\,X\,(v_{\nu}^\dag)^{(\ell)}
-\frac{1}{2}\left\{v_{\mu}^{(k)}\,(v_\nu^\dag)^{(\ell)})\,,\,X\,\right\}\Big)\\
\label{LINDMICO0c}
&&\hskip-.2cm
=\frac{1}{2}\sum_{k,\ell=0}^{N-1}J_{k\ell}\sum_{\mu,\nu=1}^d\,D_{\mu\nu}\,\Big(v_\mu^{(k)}\,\left[X\,,\,(v_{\nu}^\dag)^{(\ell)}\right]\,+\,\left[v_{\mu}^{(k)}\,,\,X\right]\,(v_\nu^\dag)^{(\ell)}\Big)
\ .
\eea
The single site terms in the Hamiltonian contribution $\HH_N$ are the same for each site with no interactions among spins either belonging to a same chain or to different ones. Instead,
in the purely dissipative contribution $\DD_N$, the mixing action of the Kraus operators $v_\mu$ is weighted by the coefficients $J_{k\ell}\,D_{\mu\nu}$, involving in general different sites. 
Altogether, they form a Kossakowski matrix $J\otimes D$; in order to ensure the complete positivity of the generated dynamical maps $\displaystyle\Phi_t^N={\rm e}^{t\LL_N}$,  both $J$ and $D$ must be positive semi-definite.
We shall leave the operators $h$ and $v_\mu$ completely unspecified; they will be fixed only later, when
discussing specific examples of entanglement generation.

In order to enforce translation invariance, one attaches the same hamiltonian to each sites $h^{(k)}=h$, 
and further consider different site couplings $J_{k\ell}$ of the form
\be
\label{JKL1}
J_{k\ell}=J(|k-\ell|)\ ,\qquad J(0)=:J_0> 0\ .
\ee
Furthermore, we shall assume the strength of the mixing terms to decrease with the site distance in such a way that
\be
\label{JKL2}
\lim_{N\to\infty}\frac{1}{N}\sum_{k,\ell=0}^{N-1}|J_{k\ell}|
=J_0+\lim_{N\to\infty}\frac{1}{N}\sum_{k\neq \ell=0}^{N-1}|J_{k\ell}|\,<\,+\infty\ .
\ee
This request together with \eqref{JKL1} implies that
\be
\label{JKL3}
\lim_{N\to\infty}\frac{1}{N}\sum_{k=0}^{N-1}|J_{k\ell}|=0\qquad\forall\ \ell\in\NN\ .
\ee

\medskip

\begin{remark}
\label{rem4}
{\rm The generator $\LL_N$ does not mediate any direct interaction between different spins since the Hamiltonian in $\HH_N$ does not have interaction terms. On the other hand, the dissipative term $\DD_N$ accounts for environment induced dissipative effects by means of the anti-commutator 
$$
-\frac{1}{2}\left\{\sum_{k,\ell=0}^{N-1}\sum_{\mu,\nu=1}^d\,J_{k\ell}D_{\mu\nu}\,v_\mu^{(k)}\,(v_{\nu}^\dagger)^{(\ell)}\,,\,X
\right\}\ ,
$$
while the remaining term
$$
\sum_{k,\ell=0}^{N-1}\sum_{\mu,\nu=1}^dJ_{k\ell}D_{\mu\nu}\,v_\mu^{(k)}\,X\,(v_{\nu}^\dag)^{(\ell)} \ ,
$$
also known as quantum noise, contributes to statistical mixing.
This latter effect can be better appreciated by diagonalising the non-negative matrix $J\otimes D$ and recasting the corresponding  contribution to $\DD_N$ into the Kraus-Stinespring form $\sum_a L_a\,X\,L_a^\dag$ of completely positive maps. By duality, it gives rise to  a map on local density matrices,
$$
\ca_{[0,N-1]}\ni\rho_N\mapsto\sum_aL^\dag_a\,\rho_N\,L_a \ ,
$$ 
that transforms pure states into mixed ones.
As we shall see, the presence of Kraus operators supported by both chains may allow this mixing term to entangle 
them at the mesoscopic level even in absence of direct spin interactions.}
\qed
\end{remark}
\medskip

An important request needed for the discussion presented in the next sections is the time-invariance of the microscopic state $\omega_\beta$. Were it not so, the state dependent mesoscopic canonical commutation relations would also depend on time, opening the way to mesoscopic non-markovian time-evolutions: such an interesting issue is however outside the scope of the present work and will be addressed elsewhere. 
We shall thus consider local generators $\LL_N$ such that
\be
\label{invst}
\omega_N\circ\Phi_t^N=\omega_N\ ,
\ee
where $\omega_N$ denotes the local state resulting from restricting $\omega_\beta$ to $\ca_{[0,N-1]}$.

\subsection{Emerging mesoscopic dynamics}

We shall now prove that, under certain technical conditions to be specified later, the mesoscopic dynamics that emerges in the limit of large $N$ from the local time-evolution \hbox{$\Phi^N_t={\rm e}^{t\LL_N}$}, ${t\geq0}$, generated by \eqref{LINDMICO0a}-\eqref{LINDMICO0c} is a dissipative semigroup $\Phi_t={\rm e}^{t\LL}$, ${t\geq 0}$, of completely positive, unital quasi-free maps on the algebra of fluctuations.
Namely, that, under the mesoscopic dynamics, displacement operators $W(r)$ of the form \eqref{Weyl6} are mapped into themselves,
\begin{equation}
\Phi_t[W(r)]={\rm e}^{f_r(t)}\,W(r_t)\ ,
\label{DYN}
\end{equation}
where both the function $f_r(t)$ and the time-dependent eight-dimensional real vector $r_t=(r_1^t,\ldots,r_8^t)^{tr}$
are to be determined.
\medskip

\begin{remark}
\label{rem4-1}
{\rm It is worth noting that, due to unitality and complete positivity, the maps $\Phi_t$ obey Schwartz-positivity
\be
\label{Schwpos}
\Phi_t(X^\dag X)\,\geq\,\Phi_t(X^\dag)\,\Phi_t(X)\ .
\ee
Moreover, since the Weyl operators $W(r)$ are unitary,
\be
\label{Schwpos1}
\|\Phi_t(W(r))\|=\left|{\rm e}^{f_r(t)}\right|\leq \|W(r)\|=1\ .
\ee
}
\qed
\end{remark}
\medskip

In order to outline the idea of the proof, we first consider the structure of the time-derivative of the time-evolving local exponentials that give rise to
$\Phi_t[W(r)]$ in \eqref{DYN}.
\medskip

\begin{lemma}
\label{lemma1}
Let $W_N(r)\in\ca_{[0,N-1]}$, $r\in\RR^8$, denote the local exponential operators \eqref{Weyl8} and define
\be
\label{P2}
W^t_N(r)={\rm e}^{f_r(t)}\,W_N(r_t)={\rm e}^{f_r(t)}\,{\rm e}^{i(r_t,F_N)}\ ,
\ee
with $r_t=(r_1^t,\ldots,r_8^t)^{tr}$.
Then, 
\be
\frac{{\rm d}}{{\rm d}t}W_N^t(r)=
\Bigg(\frac{{\rm d}f_r(t)}{{\rm d}t}\,+\,i\left(\dot{r}_t\,,\,F_N\right)\,
-\,\frac{1}{2}\Big[(r_t,F_N)\,,\,(\dot{r_t},F_N)\Big]\Bigg)\,W_N^t(r)
+\,E_{N}\,
\label{lemma1-1}
\ee
with $E_{N}$ vanishing in norm when $N\to\infty$ for all finite $t\geq 0$.
\end{lemma}
\medskip

\noindent
\begin{proof} 
\noindent
Recalling (\ref{Weyl8}) and (\ref{Weyl8a}), one can write:
$$
W_N^t(r)={\rm e}^{f_r(t)}\,{\rm e}^{iF_N(x_{r_t})}\ ,\quad
x_{r_t}=\sum_{j=1}^8r_t^j\,x_j\ .
$$
Note that $\displaystyle\dot{F}_N(x_{r_t}):=\frac{{\rm d}}{{\rm d}t}F_N(x_{r_t})=F_N(\dot{x}_{r_t})=(\dot{r}_t,F_N)$. 
Introduce now the following nested commutators:
\be
\label{comms}
\KK_A^n(B):=\Big[A\,,\,\KK^{n-1}_A(B)\Big]\ ,\quad \KK^0_A(B)=B\ .
\ee
Then, as shown in Appendix C, one has:
\bea
\nonumber
&&\hskip-2.5cm
\frac{{\rm d}}{{\rm d}t}\,W_N(r_t)=\left(\sum_{n=1}^\infty\frac{i^n}{n!}\,\KK^{n-1}_
{F_N(x_{r_t})}\Big(F_N(\dot{x}_{r_t})\Big)\right)\,W_N(r_t)\\
\hskip-.8cm
&=&\Big(i\,\left(\dot{r}_t\,,\,F_N\right)
-\frac{1}{2}\Big[F_n(x_{r_t}),\ F_N(\dot x_{r_t})\Big]\Big)\,W_N(r_t)\,+\,E_{N}\\
\label{expder1}
\hskip-.8cm
E_{N}&=&
\sum_{n=3}^\infty\frac{i^n}{n!}\,\KK^{n-1}_{F_N(x_{r_t})}\Big(F_N(\dot{x}_{r_t})\Big)\ ,
\eea
thus recovering the second and third terms in the r.h.s. of \eqref{lemma1-1}. 
Moreover, since operators at different sites commute, one has:
$$
\KK^{n-1}_{F_N(x_{r_t})}\big(F_N(\dot{x}_{r_t})\big)
=\frac{1}{N^{n/2}}\sum_{k=0}^{N-1}\KK^{n-1}_{x^{(k)}_r}(\dot{x}^{(k)}_r)\ .
$$
Further, using
$$
\Big\| \KK^{n-1}_{x_r^{(k)}}(\dot{x}^{(k)}_{r_t}) \Big\|\leq 2^{n-1} \|x_r\|^{n-1} \|\dot x_r\|\ ,
$$
one estimates
$$
\Big\|\KK^{n-1}_{F_N(x_{r_t})}\big(F_N(\dot{x}_{r_t})\big)\Big\|\leq
\frac{1}{\sqrt{N}}\Bigg(\frac{2\|x_{r_t}\|}{\sqrt{N}}\Bigg)^{n-1}\,\|\dot{x}_{r_t}\|\ .
$$
As a consequence, the norm of $E_N$ in (\ref{expder1}) is bounded as
\be
\Big\|E_{N}\Big\|\leq
\frac{{\rm e}^{2\|x_{r_t}\|}}{\sqrt{N}}\, \|\dot{x}_{r_t}\|\ .
\label{bound1}
\ee
Therefore, from $\displaystyle\|\dot{x}_{r_t}\|\leq\sum_{j=1}^8|\dot{r}^j_t|\,\|x_j\|$, it follows that, in the limit of large $N$, $E_{N}$ vanishes in norm uniformly for $0\leq t\leq {\cal T}$, with $\cal T$ any finite, positive constant.
\qed
\end{proof}
\medskip

Notice that, beside the scalar term, the dominant contributions to the time derivative of
$W_N^t(r)$ scale like fluctuations
and mean-field quantities. 
We want to compare them with similarly scaling terms in $\LL_N[W^t_N(r)]$. The following {\sl Lemma}
is then useful.

\begin{lemma}
\label{lemma3}
Given the local dissipative semigroup on $\ca_{[0,N-1]}$ generated by  
\beann
\partial_tX(t)&=&\LL_N[X(t)]\ ,\quad \LL_N[X]=\HH_N[X]\,+\,\DD_N[X]\\
\HH_N[X]&=&i\Big[H_N\,,\,X\Big]\ ,\quad H_N=\sum_{k=0}^{N-1}\,h^{(k)},\quad h^{(k)}=h=h^\dag\\
\DD_N[X]&=&\sum_{k,\ell=0}^{N-1}J_{k\ell}\sum_{\mu,\nu=1}^d\,D_{\mu\nu}\Bigg(v_\mu^{(k)}\,X\,(v_{\nu}^\dag)^{(\ell)}
-\frac{1}{2}\left\{v_{\mu}^{(k)}\,(v_\nu^\dag)^{(\ell)}\,,\,X\,\right\}\Bigg)\ ,
\eeann
with positive semi-definite matrices $J\otimes D=[J_{k\ell}]\otimes[D_{\mu\nu}]$ and
coefficients \hbox{$J_{k\ell}=J(|k-\ell|)$} satisfying \eqref{JKL1} and \eqref{JKL2}, one can 
recast the action of the Lindblad generator on $W_N(r)$ as follows:
\bea
\label{Ldec0}
&&\hskip -1.5cm
\LL_N\big[W_N(r)\big]=\,i\,\LL_N\big[(r,F_N)\big]\,W_N(r)\,-\,\frac{1}{2}\Big[(r,F_N)\,,\,\LL_N\big[(r,F_N)\big]\Big]\,W_N(r)\\
&&\hskip -1cm
+\,\frac{1}{2}\Big(\LL_N\big[(r,F_N)\big]\,(r,F_N)\,+\,(r,F_N)\,\LL_N\big[(r,F_N)\big]
-\,\LL_N\big[(r,F_N)^2\big]\Big)\,W_N(r)\,+\,L_N
\label{Ldec2}
\eea 
with $L_N=\mathcal{R}_N+D_N$ and $\mathcal{R}_N$, $D_N$ vanishing in norm when 
$N\to\infty$.
\end{lemma}
\medskip

\begin{proof}
\noindent
We shall analyze separately the hamiltonian and dissipative contributions.\hfill\break

\leftline{$\bullet$ {\sl Hamiltonian contribution}}
\smallskip

Since $W_N(r)$ is unitary, the Hamiltonian term can be recast as 
\beann
\nonumber
\hskip-.5cm
\HH_N[W_N(r)]&=&i\sum_{k=0}^{N-1}\Bigg(h^{(k)}-W_N(r)\,h^{(k)}\,W^\dag_N(r)\Bigg)W_N(r)\\
\label{ham1}
\hskip-.5cm
&=&
-i\Bigg(\sum_{k=0}^{N-1}H_N^{(k)}(x_r)\Bigg)\,W_N(r)\\
\hskip-.5cm
\label{ham2}
H_N^{(k)}(x_r)&=&\sum_{n=1}^\infty\frac{i^n}{n!}\,\KK^n_{F_N(x_r)}(h^{(k)})=
\sum_{n=1}^\infty\frac{i^n}{n!N^{n/2}}\KK^n_{x^{(k)}_r}(h^{(k)})\ ,
\eeann
whence $\HH_N[W_N(r)]=\Big(H^{(1)}_{N}(x_r)+H^{(2)}_{N}(x_r)\Big)\,W_N(r)\,+\,{\cal R}_{N}$, where
\bea
\label{ham3}
H^{(1)}_{N}(x_r)&=&-\Big[H_N\,,\,(r,F_N)\Big]\\
\label{ham3a}
H^{(2)}_{N}(x_r)&=&-\frac{i}{2}\Big[(r,F_N)\,,\,\Big[H_N\,,\,(r,F_N)\Big]\Big]
\\
\label{ham4}
{\cal R}_{N}&=&-\,i\sum_{k=0}^{N-1}\sum_{n=3}^\infty\frac{i^n}{n!N^{n/2}}\KK^n_{x^{(k)}_r}(h^{(k)})\,W_N(r)\ .
\eea
Since $\|h^{(k)}\|=\|h\|$ and 
$\|x^{(k)}_r\|=\|x_r\|$ for all $k$, one can write:
\be
\label{bound2}
\|{\cal R}_{N}\|\leq\sum_{k=0}^{N-1}\sum_{n=3}^\infty\frac{1}{n!N^{n/2}}\,\left\|\KK^n_{x^{(k)}_r}(h^{(k)})\right\|\leq\frac{{\rm e}^{2\|x_r\|}}{\sqrt{N}}\,\|h\|\ .
\ee
\medskip

\leftline{$\bullet$ {\sl Dissipative contribution}}

Setting $W_N(r)\,v^{(k)}_\mu\,W^\dag_N(r)=v^{(k)}_\mu\,+\,V^{(k)}_{\mu N}$, where
$$
V^{(k)}_{\mu N}=\sum_{n=1}^\infty\,\frac{i^n}{n!} \KK^n_{F_N(x_r)}(v^{(k)}_\mu)=
\sum_{n=1}^\infty\,\frac{i^n}{n!N^{n/2}} \KK^n_{x_r^{(k)}}(v^{(k)}_\mu)\ ,
$$
one rewrites the purely dissipative contribution as
\be
\DD_N[W_N(r)]
=\frac{1}{2}\sum_{k,\ell=0}^{N-1}J_{k\ell}\sum_{\mu,\nu=1}^{d}\,
D_{\mu\nu}\,\Big(
v^{(k)}_\mu(V^\dag_{\nu N})^{(\ell)}
-V^{(k)}_{\mu N}(v^\dag_\nu)^{(\ell)}
-V^{(k)}_{\mu N}\,(V^\dag_{\nu N})^{(\ell)}\Big)\,W_N(r)\ .
\label{decomp}
\ee
Collecting contributions that scale not faster than $1/N$, one can write:
\bea
\label{proof1c}
&&\hskip-1cm
V^{(k)}_{\mu N}=i\,\Big[(r,F_N)\,,\,v_\mu^{(k)}\Big]\,-\,\frac{1}{2}\Big[(r,F_N)\,,\,\Big[(r,F_N)\,,\,v_\mu^{(k)}\Big]\Big]
+\Delta^{(k)}_{\mu N}\ ,\\
\label{proof1d}
&&\hskip-1cm
\Delta^{(k)}_{\mu N}=\sum_{n=3}^\infty\,\frac{i^n}{n!N^{n/2}}\, \KK^n_{x^{(k)}_r}(v^{(k)}_\mu)\\
\label{proof1b}
&&\hskip-1cm
\label{VV}
V^{(k)}_{\mu N}\,(V^\dag_{\nu N})^{(\ell)}=-\,\Big[(r,F_N)\,,\,v_\mu^{(k)}\Big]\,\Big[(r,F_N)\,,\,(v_\nu^\dag)^{(\ell)}\Big]\,+\,\Delta^{(k\ell)}_{\mu\nu N}\\
&&\hskip-1cm
\Delta^{(k\ell)}_{\mu\nu N}=\sum_{n+m\geq3}\frac{i^n(-i)^m}{n!m!N^{(n+m)/2}}\, 
\KK^n_{x^{(k)}_r}(v^{(k)}_\mu)\,\KK^m_{x^{(\ell)}_r}(v^\dag_\nu)^{(\ell)})\ .
\eea
Using as before
$\displaystyle\|\KK^n_{x^{(k)}_r}(v^{(k)}_\mu)\|\leq 2^n\|x_r\|^n\,\|v_\mu\|$, one gets
\be
\label{est0}
\|\Delta^{(k)}_{\mu N}\|\leq\frac{{\rm e}^{2\|x_r\|}}{N^{3/2}}\,\|v_\mu\|\ ,\quad
\|\Delta^{(k\ell)}_{\mu\nu N}\|\leq\frac{{\rm e}^{4\|x_r\|}}{N^{3/2}}\,\|v_\mu\|\,\|v_\nu\|\ .
\ee

Using these results, one can decompose $\DD_N$ as the sum of three contributions scaling at most
as $1/N$, plus a correction term:
$\DD_N[W_N(r)]=\Big(D^{(1)}_{N}(x_r)\,+\,D^{(2)}_{N}(x_r)+\,D^{(3)}_{N}(x_r)\Big)\,W_N(r)\,+D_N$.
The contribution $D^{(1)}_{N}(x_r)$ comes from the first term in \eqref{proof1c}, it scales as a fluctuation and, using \eqref{LINDMICO0c}, it can be rewritten as:
\be
D^{(1)}_{N}(x_r)=i\DD_N\big[(r,F_N)\big]
\ .
\label{proof1e}
\ee
The second contribution scales as $1/N$ and comes from the second term in \eqref{proof1c}
and the first two terms in the r.h.s. of (\ref{decomp}); using
$$
\big[x\,,\,[x\,,\,v]\big]v^\dag\,-\,v\,\big[x\,,\,[x\,,\,v^\dag]\big]
=-\big[x\,,\,v[x\,,\,v^\dag\big]\,+\,\big[v\,,\,x]\,v^\dag\big]\ ,
$$
it can be recast in the form
\be
\label{proof1f}
D^{(2)}_{N}(x_r)=-\,\frac{1}{2}\big[(r,F_N)\,,\,\DD_N\left[(r,F_N)\right]\big]\ .
\ee
Further, using the relation
\beann
&&
x\,\left(v\,[x\,,\,v^\dag]\,+\,[v\,,\,x]\,v^\dag\right)\,+\,\Big(v\,[x\,,\,v^\dag]\,+\,[v\,,\,x]\,v^\dag\Big)\,x\,-\\
&&\hskip 3cm
-\,v\,[x^2\,,\,v^\dag]\,-\,[v\,,\,x^2]\,v^\dag=2\,[x\,,\,v]\,[x\,,\,v^\dag]\ ,
\eeann
the third contribution, that comes from the first term in the r.h.s of (\ref{VV}) and the last term in
the r.h.s of (\ref{decomp}) and scales as a mean-field quantity, can be rewritten as
\begin{eqnarray}
\nonumber
&&\hskip-1cm
D^{(3)}_{N}(x_r)=\frac{1}{2}\sum_{k,\ell=0}^{N-1}J_{k\ell}\sum_{\mu,\nu=1}^{d}D_{\mu\nu}\Big[(r,F_N)\,,\,v_\mu^{(k)}\Big]\,\Big[(r,F_N)\,,\,(v_\nu^\dag)^{(\ell)}\Big]\\
&&
=\frac{1}{2}\Big(\DD_N\big[(r,F_N)\big]\,(r,F_N)\,+\,(r,F_N)\,\DD_N\big[(r,F_N)\big]\,
-\,\DD_N\big[(r,F_N)^2\big]\Big)\ .
\end{eqnarray}
Notice that the Hamiltonian term is such that
$$
\HH_N\big[(r,F_N)\big]\,(r,F_N)\,+\,(r,F_N)\,\HH_N\big[(r,F_N)\big]\,-\,\HH_N\big[(r,F_N)^2\big]=0\ ,
$$
so that one can add the above contribution to that of $\DD_N$ without modifying it,  thus obtaining
\be
\label{proof1g}
D^{(3)}_{N}(x_r)=\frac{1}{2}\Big(\LL_N\big[(r,F_N)\big]\,(r,F_N)\,+\,(r,F_N)\,\LL_N\big[(r,F_N)\big]\,-\,\LL_N\big[(r,F_N)^2\big]\Big)\ .
\ee
Finally, the correction term $D_{N}$ reads 
$$
D_{N}=\frac{1}{2}\sum_{k,\ell=0}^{N-1}J_{k\ell}\sum_{\mu,\nu=1}^{d}\,
D_{\mu\nu}\,\Big(
v^{(k)}_\mu(\Delta^\dag_{\nu N})^{(\ell)}-(v_\nu^\dag)^{(\ell)})\,\Delta^{(k)}_{\mu N}-\Delta^{(k\ell)}_{\mu\nu N}\Big)\,W_N(r)\ ,
$$
and \eqref{est0} provides the upper bound 
\be
\label{bound3}
\|D_{N}\|\leq \frac{3}{2N^{3/2}}\sum_{k,\ell=0}^{N-1}|J_{k\ell}|\,\sum_{\mu,\nu=1}^d\,|D_{\mu\nu}|\,\|v_\mu\|\,\|v_\nu\|\,{\rm e}^{4\|x_r\|}\ ,
\ee
whence the condition \eqref{JKL2} on the coefficients $J_{k\ell}$ makes it vanish in norm as $1/\sqrt{N}$ when $N\to\infty$. 
Putting together all these results and estimates, the statement of the Lemma immediately follows. 
\qed
\end{proof}

\subsection{Quasi-free dissipative mesoscopic dynamics}

We shall choose single particle Hamiltonian operators $h=h^\dag$ and Kraus operators $v_\mu$ such that, for all $0\leq k\leq N-1$, the linear span $\cx$ of the chosen set $\chi$ of on-site microscopic observables be mapped into itself by the Lindblad generator:
\be
\label{Ldec}
\LL_N[x^{(k)}_j]=\HH_N[x^{(k)}_j]\,+\,\DD_N[x^{(k)}_j]=\sum_{p=1}^8 \left(\mathcal{H}_{jp}\,+\,\mathcal{D}_{jp}\right)\,x^{(k)}_p\ .
\ee
We have denoted by $\mathcal{H}=[\mathcal{H}_{jp}]$ and $\mathcal{D}=[\mathcal{D}_{jp}]$ the $8\times 8$ matrices of coefficients specifying the action of the hamiltonian and dissipative generators and set 
\be
\label{singsiteL}
\mathcal{L}=\mathcal{H}+\mathcal{D}\ ,\qquad \mathcal{H}_{ij}^*=\mathcal{H}_{ij}\ ,\quad
\mathcal{D}_{ij}^*=\mathcal{D}_{ij}\ .
\ee

When comparing the time derivative in \eqref{lemma1-1} with the action of the generator in \eqref{Ldec0}, one has to match contributions with the same scaling. 
Since, for large $N$, mean-field observables behave as scalar multiples of the identity the matching among them can be obtained by a proper choice of the unknown function $f_r(t)$.
On the other hand, the term $i(\dot{r}_t,F_N)$ in the time-derivative that scales as a fluctuation should be matched by the term $i\LL_N[(r,F_N)]$ in the action of the generator.

Then, for generic $r\in\RR^8$, the equality
\be
\label{sqrtn}
(\dot{r}_t,F_N)=\LL_N\big[(r_t,F_N)\big]=\left(\mathcal{L}^{tr}r_t,F_N\right) 
\ee
is equivalent to having
\be
\label{Ldec4}
r_t={\rm e}^{\,t\,\mathcal{L}^{tr}}\,r\ ,\quad 
\Phi^N_t\big[(r,F_N)\big]=\Big(r,\,{\rm e}^{t\mathcal{L}}\,F_N\Big)\ ,
\ee
where, as before, $\mathcal{L}^{tr}$ denotes the transposed $\mathcal{L}$. Notice that such a time-dependence also satisfies
\be
\label{Ldec1}
\Big[(r_t,F_N)\,,\,(\dot{r}_t,F_N)\Big]=\Big[(r_t,F_N)\,,\,\LL_N\big[(r_t,F_N)\big]\Big]\ .
\ee 
Therefore, the difference between the time-derivative of $W_N^t(r)$ and the action of the generator on the same operator becomes
\be
\label{difft}
\frac{{\rm d}}{{\rm d}t}W^t_N(r)-\LL_N\left[W_N^t(r)\right]=E_N-L_N+\left(\frac{{\rm d}f_r(t)}{{\rm d}t}-D^{(3)}_{N}(x_{r_t})\right)\,W^t_N(r_t)\ ,
\ee
where now $D^{(3)}_{N}(x_{r_t})$ in (\ref{proof1g}) can be expressed as:
\be
D^{(3)}_{N}(x_{r_t})=\frac{1}{2}\Big(\big(\mathcal{L}^{tr}r_t,F_N\big)(r_t,F_N)+(r_t,F_N)\big(\mathcal{L}^{tr}r_t,F_N\big)
-\LL_N\big[(r_t,F_N)^2\big]\Big)\ .
\label{difft0}
\ee
Since the microscopic state $\omega_\beta$ is $\Phi^N_t$-invariant, so that $\omega_\beta\circ\LL_N=0$, and of the product form \eqref{STATE} with $\omega_\beta(x_j)=0$, we get
\be
\label{difft1}
\omega_\beta\left(D^{(3)}_{N}(x_{r_t})\right)=\frac{1}{2}\big(\mathcal{L}^{tr}
r_t,\Sigma^{(\beta)}\,r_t\big)\,+\,\frac{1}{2}\big(r_t,\Sigma^{(\beta)}\mathcal{L}^{tr}\,r_t\big)
=\big(r_t,\mathcal{L}\,\Sigma^{(\beta)}\,r_t\big)\ ,
\ee
where the last equality follows from $r$ being a real vector and the covariance matrix 
$\Sigma^{(\beta)}$ \eqref{covmat2} being real symmetric.
This result and \eqref{difft} suggest then to choose
\bea
\label{difft2}
\frac{{\rm d}}{{\rm d}t}f_r(t)&=&\omega_\beta\left(D^{(3)}_{N}(x_{r_t})\right)\quad \hbox{so that}\\
\label{difft2a}
f_r(t)&=&-\frac{1}{2}\,\left(r,\mathcal{Y}_t\,r\right)\ ,\quad \mathcal{Y}_t\,=\,
\Sigma^{(\beta)}\,-\,{\rm e}^{t\mathcal{L}}\,\Sigma^{(\beta)}\,{\rm e}^{t\mathcal{L}^{tr}}
\eea 
with initial condition $f_r(0)=0$.
It turns out that 
\be
\label{difft3}
\mathcal{Y}_t\geq 0\quad\hbox{so that}\quad f_r(t)\leq 0\quad\hbox{and}\quad 
{\rm e}^{f_r(t)}\leq 1\ ,
\ee
for all $t\geq 0$, in agreement with \eqref{Schwpos1}.
This can be seen as follows:
let $\lambda\in\mathbb{C}^8$ be a generic complex vector and set $q_\lambda=\sum_{j=1}^8
\lambda_j\,x_j\in\mathcal{X}$. Then, Schwartz positivity \eqref{Schwpos}, the time-invariance of $\omega_\beta$ and the second relation in \eqref{Ldec4} yield
\begin{eqnarray*}
&&
\frac{1}{2}\sum_{i,j=1}^d\lambda_i^*\lambda_j\,\omega_\beta\Big(\Big\{F_N(x_i)\,,\,F_N(x_j)
\Big\}\Big)=\\
&&\hskip .5cm
=\frac{1}{2}\sum_{i,j=1}^d\lambda_i^*\lambda_j\,\omega_\beta\Big(\Phi^N_t\Big[\Big\{F_N(x_i)\,,\,F_N(x_j)\Big\}\Big]\Big)\\
&&\hskip .5cm
\geq\frac{1}{2}\omega_\beta\Big(\Phi^N_t[F_N(q^\dag_\lambda)]\,\Phi_t^N\left[F_N(q_\lambda)\right]\Big)\,+\,\frac{1}{2}\omega_\beta\Big(\Phi^N_t\left[F_N(q_\lambda)\right]\,\Phi^N_t[F_N(q^\dag_\lambda)]\Big)\\ 
&&\hskip .5cm
=\frac{1}{2}\omega_\beta\Big(\Big(\lambda,{\rm e}^{t\mathcal{L}}\,F_N\Big)\,\Big(\lambda^*,{\rm e}^{t\mathcal{L}}\,F_N\Big)\Big)\,+\,\frac{1}{2}\omega_\beta\Big(\Big(\lambda^*,
{\rm e}^{t\mathcal{L}}\,F_N\Big)\,\Big(\lambda,{\rm e}^{t\mathcal{L}}\,F_N\Big)\Big)\\
&&\hskip .5cm
=\frac{1}{2}\sum_{i,j;r,s=1}^d\lambda_i^*\lambda_r\,\left({\rm e}^{t\mathcal{L}}\right)_{ij}\,
\left({\rm e}^{t\mathcal{L}}\right)_{rs}\,\omega_\beta\Big(\Big\{F_N(x_j)\,,\,F_N(x_s)\Big\}\Big)\ .
\end{eqnarray*}
Recalling (\ref{covmat}), in the large $N$ limit one thus obtains, for all $\lambda\in\mathbb{C}^d$,
$$
\Big(\lambda,\Sigma^{(\beta)}\,\lambda\Big)\geq \sum_{i,j;r,s=1}^d\lambda_i^*\lambda_r\left({\rm e}^{t\mathcal{L}}\right)_{ij}\,\left({\rm e}^{t\mathcal{L}}\right)_{rs}\,\Sigma^{(\beta)}_{js}=
\Big(\lambda,{\rm e}^{t\mathcal{L}}\,\Sigma^{(\beta)}\,{\rm e}^{t\mathcal{L}^{tr}}\,\lambda\Big)\ .
$$
Equipped with these considerations, we prove the following main technical result.
\medskip

\begin{theorem}
\label{qfth}
Consider the quasi-local algebra $\ca$ with  
the translation-invariant KMS state $\omega_\beta$ in \eqref{STATE}, the  self-adjoint set 
$\chi=\{x_j\}_1^8$ in \eqref{matrix}, \eqref{matrixa} and the resulting quantum fluctuation algebra 
$\mathcal{W}(\chi,\sigma^{(\beta)})$.
Let the local algebras $\ca_{[0,N-1]}$ evolve under the local dissipative semigroups 
$\{\Phi^N_t\}_{t\geq0}$ with Lindblad generator as in 
\eqref{LINDMICO0a}-\eqref{LINDMICO0c} where the 
Hamiltonian and Kraus operators satisfy the relations \eqref{Ldec}.
In the limit of large $N$, the emerging dissipative mesoscopic dynamics is described by  a
semi-group $\{\Phi_t\}_{t\geq0}$ of completely positive, unital maps 
on $\mathcal{W}(\chi,\sigma^{(\beta)})$, such that
\be
\label{RD2}
\lim_{N\to\infty}\omega_\beta\Big(W_N(a)\,\Phi^N_t\big[W_N(r)\big]\,
W_N(b)\Big)=\Omega_\beta\Big(W(a)\,\Phi_t\big[W(r)\big]\,W(b)\Big)\ ,
\ee
for all microscopic exponential operators $W_N(a)$, $W_N(b)$, $W_N(r)$, with $W(a)$, 
$W(b)$ and $W(r)$ the corresponding Weyl operators in the algebra $\mathcal{W}(\chi,\sigma^{(\beta)})$ and $\Omega_\beta$ the state on it defined by \eqref{qfs1}, which is then left invariant by $\Phi_t$. 
Moreover, the maps $\Phi_t$ are quasi-free, {\it i.e.} they map Weyl operators into Weyl operators:
$\displaystyle \Phi_t[W(r)]={\rm e}^{f_r(t)}\,W(r_t)$, with
$r_t$ and $f_r(t)$ as in \eqref{Ldec4} and \eqref{difft2}-\eqref{difft2a}, respectively.
\end{theorem}
\medskip

\begin{remark}
\label{rem6}
{\rm The chosen type of convergence conforms to the fact that the action of any map on the quantum fluctuation algebra is totally specified by its action on the Weyl operators 
$W(r)$. Such an action is in turn completely defined by the matrix elements in the GNS representation based on the limit state $\Omega_\beta$. Both the state $\Omega_\beta$ and the Weyl operators arise from the large $N$ limit of the microscopic exponential operators $W_N(r)$ with respect to the microscopic state $\omega_\beta$.}
\qed
\end{remark}
\medskip

\noindent
\begin{proof}
For sake of simplicity, we shall set
\be
\label{notation}
\omega^N_{ab}(\cdot):=\omega_\beta\Big(W_N(a)\,\cdot\,W_N(b)\Big)\ ,\
\Omega_{ab}(\cdot)=
\Omega\Big(W(a)\,\cdot\,W(b)\Big)
\ee
and then show that, for arbitrary $a,b,r\in\RR^8$, the positive quantity
\be
\label{P1}
I_N=\Bigg|\Omega_{ab}\Big(\Phi_t\big[W(r)\big]\Big)\,-\,\omega^N_{ab}\Big(\Phi_t^N\big[W_N(r)\big]\Big)\Bigg|
\ee
vanishes when $N\to\infty$. Writing
$\Phi_t^N[W(r)]=\Phi_t^N[W_N(r)]-W^t_N(r)+W^t_N(r)$ one has $I_N\leq I^{(1)}_{N}+I^{(2)}_{N}$, where
\bea
\label{I1}
I^{(1)}_{N}&:=&\Big|\omega^N_{ab}\Big(W_N^t(r)\,-\,\Phi_t^N[W_N(r)]\Big)\Big|\\
\label{I2}
I^{(2)}_{N}&:=&\Big|\Omega_{ab}\Big(\Phi_t[W(r)]\Big)\,-\,\omega^N_{ab}\Big(W_N^t(r)\Big)\Big|\ .
\eea
Because of \eqref{P2} and \eqref{Schwpos1}, one gets
$$
I^{(2)}_{N}\leq\left|\Omega_{ab}\Big(W(r_t)\Big)\,-\,\omega_{ab}^N\Big(W_N(r_t)\Big)\right|\ .
$$
Then, the properties of the exponential operators (see Remark \ref{rem6}) make
$I^{(2)}_{N}\to 0$ with $N\to\infty$, uniformly for any finite time interval, $0\leq t\leq {\cal T}$.
On the other hand, in order to estimate $I^{(1)}_{N}$, we write
\beann
W^t_N(r)\,-\,\Phi_t^N[W_N(r)]&=&\int_0^t{\rm d}s\,\frac{{\rm d}}{{\rm d}s}\left(\Phi_{t-s}^N\Big[W^s_N(r)\Big]\right)\\
&=&\int_0^{t}{\rm d}s\,\Phi_{t-s}^N\Big[\frac{{\rm d}}{{\rm d}s}W_N^s(r)\,-\,\LL_N[W_N^s(r)]\Big]\ .
\eeann
Then, recalling (\ref{difft}), one obtains: 
\beann
I^{(1)}_{N}&\leq&\int_0^t{\rm d}s\,\Big|\omega^N_{ab}\Big(\Phi^N_{t-s}[\delta_N(r,s)]\Big)\Big|\\
\delta_N(r,s)&:=&E_N-L_N\,+\,\left(\frac{{\rm d}f_r(s)}{{\rm d}t}-D^{(3)}_{N}(x_{r_s})\right)\,W^t_N(r_s)\ ,
\eeann
with $D^{(3)}_{N}(x_{r_s})$ given by \eqref{difft0}.
Since the microscopic state $\omega_\beta$ obeys the KMS conditions \eqref{KMS1}, from the Cauchy-Schwartz inequality it follows that
\beann
&&\hskip-.7cm
\Big|\omega^N_{ab}\Big(\Phi^N_{t-s}[\delta_N(r,s)]\Big)\Big|^2=\Big|\omega_\beta\Big(\Phi_{t-s}^N[\delta_N(r,s)]\,
W_N(b)\tau_{i\beta}[W_N(a)]\Big)\Big|^2\\
&&\hskip-.7cm
\leq\omega_\beta\Big(\Phi^N_{t-s}[\delta_N(r,s)]\Phi^N_{t-s}[\delta^\dag_N(r,s)]\Big)\,
\omega_\beta\Big(\big(\tau_{i\beta}[W_N(a)]\big)^\dag\tau_{i\beta}[W_N(a)]\Big)\ .
\eeann
For finite inverse temperatures $\beta$, the second term in the right side of the inequality is bounded on a dense subset of operators in the Weyl algebra, while the first one can be estimated by means of the invariance of $\omega_\beta$ under $\Phi_t^N$ and of Schwartz-positivity \eqref{Schwpos}:
$$
\omega_\beta\Big(\Phi^N_{t-s}[\delta_N(r,s)]\Phi^N_{t-s}[\delta^\dag_N(r,s)]\Big)\leq
\omega_\beta\Big(\delta_N(r,s)\,\delta^\dag_N(r,s)\Big)\ .
$$
The proof of the theorem can thus be completed by showing that, when $N\to\infty$, the right hand side of the above inequality vanishes uniformly for $0\leq s\leq t\leq {\cal T}$. 
The Cauchy-Schwartz inequality 
$\left|\omega(a^\dag b)\right|^2\leq\omega(a^\dag a)\,\omega(b^\dag b)$ yields
$$
\omega_\beta\left((a+b)^\dag(a+b)\right)\leq\left(\sqrt{\omega_\beta(a^\dag a)}\,+\,\sqrt{\omega_\beta(b^\dag b)}\right)^2\ .
$$
Therefore, setting $\dot{f}_r(s):={\rm d}f_r(s)/{\rm d}s$ and using \eqref{difft} and \eqref{Schwpos1} together with 
$\omega_\beta(a^\dag a)\leq \|a\|^2$ and \eqref{difft1}--\eqref{difft3}, one gets
\begin{eqnarray*}
&&
\sqrt{\omega_\beta\Big(\delta_N(r,s)\,\delta^\dag_N(r,s)\Big)}\leq
\sqrt{\omega\Big((E_N-L_N)^\dag(E_N-L_N)\Big)}\\
&&
+\,{\rm e}^{f_r(t)}\,\sqrt{\omega_\beta\bigg(\left(\dot{f}_r(s)-D^{(3)}_{N}(x_{r_s})\right)\,\left(\dot{f}_r(s)-\big(D^{(3)}_{N}\big)^\dagger(x_{r_s})\right)\bigg)}\\
&&\leq\|E_N-L_N\|\,+\,\sqrt{\omega_\beta\bigg(\left(\dot{f}_r(s)-D^{(3)}_{N}(x_{r_s})\right)\,\left(\dot{f}_r(s)-\big(D^{(3)}_{N}\big)^\dagger(x_{r_s})\right)\bigg)}\\
&&\leq\|E_N-L_N\|\,+\,\sqrt{\omega_\beta\left(D^{(3)}_{N}(x_{r_s})\,\big(D^{(3)}_{N}\big)^\dagger(x_{r_s})\right)-\left|\omega_\beta\left(D^{(3)}_{N}(x_{r_s})\right)\right|^2}\ .
\end{eqnarray*}
According to {\sl Lemma 2} and {\sl Lemma \ref{lemma3}} (with $r_t$ in the place of $r$ in the bound \eqref{bound3}), one obtains $\lim_{N\to\infty}\|E_N-L_N\|=0$, uniformly for $0\leq t\leq {\cal T}$. Furthermore,
\begin{eqnarray*}
D^{(3)}_{N}(x_r)&=&\frac{1}{2}\sum_{k,\ell=0}^{N-1}J_{k\ell}\sum_{\mu,\nu=1}^{d}D_{\mu\nu}\Big[(r,F_N)\,,\,v_\mu^{(k)}\Big]\,\Big[(r,F_N)\,,\,(v_\nu^\dag)^{(\ell)}\Big]\\
&=&\frac{1}{2N}\sum_{k,\ell=0}^{N-1}J_{k\ell}\sum_{\mu,\nu=1}^{d}\sum_{i,j=1}^8
D_{\mu\nu}r_s^ir_s^j\Big[x_i^{(k)}\,,\,v_\mu^{(k)}\Big]\,\Big[x^{(\ell)}_j\,,\,(v_\nu^\dag)^{(\ell)}\Big]
\end{eqnarray*}
can be recast in the form
$$
D^{(3)}_{N}(x_r)=\frac{1}{N}\sum_{k,\ell=0}^{N-1} J_{k\ell}\sum_{\mu,\nu=1}^{d}D_{\mu\nu}\,  
a_\mu^{(k)}\,b_\nu^{(\ell)}\ ,
$$
where $a_\mu^{(k)}$ and $b_\nu^{(\ell)}$ are single site operators. Then,
\beann
&&\hskip-1.5cm
\omega_\beta\left(D^{(3)}_{N}(x_{r_s})\,\big(D^{(3)}_{N}\big)^\dagger(x_{r_s})\right)-\left|\omega_\beta\left(D^{(3)}_{N}(x_{r_s})\right)\right|^2\\
&&=\sum_{k_1,\ell_1=0\atop k_2,\ell_2=0}^{N-1}\sum_{\mu_1,\nu_1=1\atop\mu_2,\nu_2=1}^d\,\frac{J_{k_1\ell_1}\,J_{k_2\ell_2}}{N^2}\ D_{\mu_1\nu_1}\,D_{\mu_2\nu_2}\,
\Bigg(\omega_\beta\Big(a_{\mu_1}^{(k_1)}\,b_{\nu_1}^{(\ell_1)}(b_{\nu_2}^\dag)^{(\ell_2)}\,(a_{\mu_2}^\dag)^{(k_2)}\Big)\\
&&\hskip6cm
-\,\omega_\beta\Big(a_{\mu_1}^{(k_1)}\,b_{\nu_1}^{(\ell_1)}\Big)\,\omega_\beta\Big((b^\dag_{\mu_2})^{(\ell_2)}\,(a_{\mu_2}^\dag)^{(k_2)}\Big)\Bigg)\ .
\eeann
Because of the assumption \eqref{JKL2} and its consequence \eqref{JKL3}, this quantity vanishes when \hbox{$N\to\infty$}. For example, suppose $k_1=k_2$, then the corresponding multiple sums can be bounded by a term proportional to
$$
\frac{1}{N^2}\sum_{k,\ell_1,\ell_2=0}^{N-1}\,|J_{k\ell_1}|\,|J_{k\ell_2}|\ .
$$
Then, the right hand side of the previous expression 
vanishes uniformly for $0\leq s\leq t\leq {\cal T}$ because of the finite number of summands and the bounded norm of all the spin operators involved in any finite interval of time.
\qed
\end{proof}
\medskip

The previous theorem shows that, when the linear space $\mathcal{X}$ of selected single-site operators is stable under the action of the local Lindblad generator, then the emergent mesoscopic irreversible dynamics maps Weyl operators into themselves: 
it turns out that such a dynamics corresponds to a semigroup of unital, completely positive maps on the Weyl algebra $\mathcal{W}(\chi,\sigma^{(\beta)})$, generated by a Lindblad generator which is at most quadratic in the fluctuation operators $F(x_i)$.
\medskip

\begin{corollary}
\label{cor1}
The maps $W(r)\mapsto\Phi_t[W(r)]=W_t(r)={\rm e}^{f_r(t)}\,W(r_t)$ with $r_t\in\mathbb{R}^8$
and $f_r(t)$ given by \eqref{Ldec4}, respectively \eqref{difft2a}, satisfy the time-evolution equation $\partial_tW_t(r)=\LL[W_t(r)]$, where the generator $\LL$ is given by
\bea
\label{fluctLind1}
&&\hskip-1.5cm
\LL[W_t(r)]=\frac{i}{2}\,\sum_{i,j=1}^8 H^{(1)}_{ij}\big[F(x_i)F(x_j)\,,\,W_t(r)\big]\\
\label{fluctLind2}
&&\hskip-.5cm+
\sum_{i,j=1}^8D^{(1)}_{ij}\left(F(x_i)\,W_t(r)\,F(x_j)\,-\,\frac{1}{2}\big\{
F(x_i)F(x_j)\,,\,W_t(r)\big\}\right),
\eea
with $H^{(1)}$ a Hermitian $8\times 8$ matrix and $D^{(1)}$ a positive semi-definite $8\times 8$ hermitian matrix, given by
\bea
\label{Flind1a}
&&H^{(1)}=-i(\sigma^{(\beta)})^{-1}\left(\mathcal{L}\,C^{(\beta)}\,-\,C^{(\beta)}\,\mathcal{L}^{tr}\right)\,(\sigma^{(\beta)})^{-1}\ ,\\
&&D^{(1)}=(\sigma^{(\beta)})^{-1}\left(\mathcal{L}\,C^{(\beta)}\,+\,C^{(\beta)}\mathcal{L}^{tr}\right)(\sigma^{(\beta)})^{-1}\ .
\label{Flind1b}
\eea
In the creation and annihilation operator formalism, using the notation introduced in 
\eqref{anncrop}, 
the generator reads
\bea
\label{fluctLind3}
&&\hskip-1.5cm
\LL[D_t(z)]=\frac{i}{2}\,\sum_{i,j=1}^8 H^{(2)}_{ij}\left[A^\dag_i\,A_j\,,\,D_t(z)\right]\\
\label{fluctLind4}
&&\hskip-.5cm+
\sum_{i,j=1}^8 D^{(2)}_{ij}\left(A_i^\dag\,D_t(z)\,A_j\,-\,\frac{1}{2}\left\{
A_i^\dag A^\dag\,,\,D_t(z)\right\}\right),
\eea
where $D_t(z)$ is the time-evolved displacement operator \eqref{Weyl6b} corresponding to the time-evolved Weyl operator $W_t(r)$ and $H^{(2)}$ and $D^{(2)}$ are $8\times 8$ matrices, given by
\be
\label{Flind2a}
H^{(2)}=\mathcal{M}^\dag\,H^{(1)}\,\mathcal{M}\ ,\qquad
D^{(2)}=\mathcal{M}^\dag\,D\,\mathcal{M}\ ,
\ee
where $\mathcal{M}$ is the matrix in \eqref{matrix1.1} of Appendix B.
\end{corollary}

\begin{proof}
Using \eqref{appb1} in Appendix C, the explicit expressions for $\dot{r}_t$, $f_r(t)$ and the relation \eqref{corcovsym} among the correlation, covariance and symplectic matrices, one computes
\beann
\partial_tW_t(r)&=&\Big(\dot{f}_r(t)\,+\,i(\dot{r}_t,F)\,-\,\frac{1}{2}\big[(r_t,F)\,,\,(\dot{r}_t,F)\big]\Big)\,W_t(r)\\
&=&
\left(i(r_t,\mathcal{L}\,F)\,+\,(r_t,\mathcal{L}\,\Sigma^{(\beta)} r_t)\,+\,\frac{i}{2}(r_t,\mathcal{L}\,\sigma^{(\beta)} r_t)\right)\,W_t(r)\\
&=&\Big(i(r_t,\mathcal{L}\,F)\,+\,(r_t,\mathcal{L}\,C^{(\beta)} r_t)\Big)\,W_t(r)
\ .
\eeann
In order to show how to match this time-derivative with the action on $W_t(r)$ of a linear map as in the statement of the Corollary, it is useful to recall \eqref{displ}, which gives 
$$
W_t(r)\,F(x_i)=\left(F(x_i)\,+\,\sum_{j=1}^8\sigma^{(\beta)}_{ij}\,r^j_t\right)\,W_t(r)\ .
$$
It is then straightforward to derive that
\beann
\LL[W_t(r)]&=&\frac{i}{2}\Big(\left(r_t,\sigma^{(\beta)}\big(H^{(1)}+(H^{(1)})^{tr}\big)F\right)\,+\,
\left(r_t,\sigma^{(\beta)}\, H^{(1)}\,\sigma^{(\beta)} r_t\right)\Big)\,W_t(r)\\
&+&\frac{1}{2}\Big(\left(r_t,\sigma^{(\beta)}\big(D^{(1)}-(D^{(1)})^{tr}\big)F\right)\,+\,
\left(r_t,\sigma^{(\beta)}\, D^{(1)}\,\sigma^{(\beta)} r_t\right)\Big)\,W_t(r)\ .
\eeann
By equating the operatorial, respectively the scalar contributions from the time-derivative and the generator action, one obtains
\beann
\mathcal{L}&=&\frac{1}{2}\sigma^{(\beta)}\left(H^{(1)}\,+\,(H^{(1)})^{tr}\right)\,
-\,\frac{i}{2}\sigma^{(\beta)}\left(D^{(1)}-(D^{(1)})^{tr}\right)\\
\mathcal{L}\,C^{(\beta)}&=&\sigma^{(\beta)}\,\frac{i\,H^{(1)}\,+\,D^{(1)}}{2}\,\sigma^{(\beta)}\ ,
\eeann
whence, by the invertibility of $\sigma^{(\beta)}$ (see \eqref{invCOMM1}), the hermiticity of 
$C^{(\beta)}$ and the the fact that $\mathcal{L}^\dag=\mathcal{L}^{tr}$ (see \eqref{singsiteL}), the result follows from 
$$
\mathcal{L}\,C^{(\beta)}\pm C^{(\beta)}\mathcal{L}^{tr}=\sigma^{(\beta)}\,\left(\frac{i\,H^{(1)}\,+\,D^{(1)}}{2}\,\mp\,\frac{i\,H^{(1)}\,-\,D^{(1)}}{2}\right)\,\sigma^{(\beta)}\ .
$$
The second part of the corollary follows from using \eqref{matrix1} and inserting it into
\eqref{fluctLind1} and \eqref{fluctLind2}
$$
F(x_i)=F^\dag(x_i)=\sum_{k=1}^8\mathcal{M}^*_{ik}\,A^\dag_k\ ,\quad
F(x_j)=\sum_{\ell=1}^8\mathcal{M}_{i\ell}\,A_\ell\ .
$$
\qed
\end{proof}

\section{Gaussian states}

The mesoscopic dissipative dynamics $\Phi_t$ obtained in the previous section is quasi-free as it maps Weyl operators into Weyl operators. The dual maps $\Psi_t$ acts on the states $\rho$ on the Weyl algebra $\mathcal{W}(\chi,\sigma^{(\beta)})$, sending them into
$\rho_t=\Psi_t[\rho]$ according to the duality relation
\be
\label{duality}
\rho_t(W)=\rho\big(\Phi_t[W]\big)\qquad\forall\, W\in\mathcal{W}(\chi,\sigma^{(\beta)})\ .
\ee
Particularly useful states on $\mathcal{W}(\chi,\sigma^{(\beta)})$ are the Gaussian states
(with zero averages) which are identified by their characteristic functions being Gaussian, 
{\it i.e.} by the following expectation of Weyl operators
\bea
\label{charfunct1}
\rho_G\big(W(r)\big)&=&\rho_G\left({\rm e}^{i(r,F)}\right)=\exp\left(-\frac{1}{2}(r,G\,r)\right)\ ,\qquad\forall r\in\RR^8\\
G&=&[G_{ij}]\ ,\quad G_{ij}=\frac{1}{2}\rho_G\left(\Big\{F(x_i)\,,\,F(x_j)\Big\}\right)\ .
\eea
These states are completely identified by their covariance matrix $G$; in particular, 
positivity of $\rho_G$ is equivalent to the following
condition on $G$ \cite{Holevo}:
\be
G+\frac{i}{2}\sigma^{(\beta)}\geq0\ ,
\label{g-positivity}
\ee
where $\sigma^{(\beta)}$ is the symplectic matrix in (\ref{COMM1}).
Clearly, the maps $\Psi_t$ transform Gaussian states into Gaussian states:
\bea
\nonumber
\Psi_t[\rho_G](W(r))&=&\rho_G\Big(\Phi_t[W(r)]\Big)={\rm e}^{f_r(t)}\,\rho_G\Big(W(r_t)\Big)\\
\label{gaussian1}
&=&
\exp\left(f_r(t)\,-\,\frac{1}{2}(r_t,G\,r_t)\right)=\rho_{G_t}\Big(W(r)\Big)\ , 
\eea
with the time-dependent covariance matrix $G_t$ obtained recalling {\sl Corollary 1}, \eqref{Ldec4} and
\eqref{difft2a}:
\be
\label{gaussian2}
G_t=\Sigma^{(\beta)}\,-\,{\rm e}^{t\mathcal{L}}\,\Sigma^{(\beta)}\,{\rm e}^{t\mathcal{L}^{tr}}+{\rm e}^{t\mathcal{L}}\,G\,{\rm e}^{t\mathcal{L}^{tr}}\ .
\ee
It follows that the mesoscopic state $\Omega_\beta$ in \eqref{qfs1} is Gaussian with covariance matrix $G=\Sigma^{(\beta)}$ and thus, as the microscopic state $\omega_\beta$ is invariant under the local dissipative dynamics $\Phi^N_t$, $\Omega_\beta$ is invariant under the mesoscopic dissipative dynamics $\Psi_t$, {\it i.e.} $G_t=\Sigma^{(\beta)}$.

A useful equivalent expression for the covariance matrix can be obtained by organizing the creation and annihilation operators in the new vector $\tilde A=(a_1,a^\dag_1,a_2,a^\dag_2,a_3,a^\dag_3,a_4,a^\dag_4)^{tr}$, and by 
introducing the coefficient vector $\tilde Z=(z_1,\bar{z}_1,z_2,\bar{z}_2,z_3,\bar{z}_3,z_4,\bar{z}_4)^{tr}\in\CC^8$ 
together with the $8\times 8$ matrix
$\tilde \Sigma_3=$diag$(1,-1,1,-1,1,-1,1,-1)$; it will be useful in the next Section while discussing
entanglement criteria for Gaussian states.

\begin{lemma} 
The displacement operator $D(z)=\exp\Big(-(Z,\Sigma_3\, A)\Big)$ in \eqref{displ1} 
can be recast as $D(z)=\exp\Big(-(\tilde Z,\tilde \Sigma_3\, \tilde A)\Big)$ with
\begin{equation}
\tilde A= \mathcal{P}^{tr}\,A,\qquad \tilde Z=\mathcal{P}^{tr}\, Z,\qquad 
\tilde\Sigma_3={\cal P}^{tr}\, \Sigma_3\, {\cal P}\ ,\qquad
\mathcal{P}\mathcal{P}^{tr}={\bf 1}_{8}\ ,
\label{lemm2}
\end{equation}
where $\mathcal{P}$ is explicitly given in \eqref{lastmat} of Appendix B.
\label{lem2}
\end{lemma}

Using this new ordering, the expectation of the displacement operator $D(z)$ with respect to a Gaussian state $\rho_G$ reads
\be
\rho_G\left(D(z)\right)=\exp\left(-\frac{1}{2}(\tilde Z,\tilde G,\tilde Z)\right)\ ,
\label{bigcov1}
\ee
with the new covariance matrix $\tilde G$ explicitly given by
\be
\tilde G=\begin{pmatrix}
\tilde G_{11}&\tilde G_{12}&\tilde G_{13}&\tilde G_{14}\\
\tilde G_{21}&\tilde G_{22}&\tilde G_{23}&\tilde G_{24}\\
\tilde G_{31}&\tilde G_{32}&\tilde G_{33}&\tilde G_{34}\\
\tilde G_{41}&\tilde G_{42}&\tilde G_{43}&\tilde G_{44}
\end{pmatrix}\ ,
\label{bigcov2}
\ee
where
\be 
\label{covariance2}
\tilde G_{ij}=\frac{1}{2}
\begin{pmatrix}
\rho_G\left(\big\{a_i,a^\dagger_j\big\}\right)&-\rho_G\left(\big\{a_i,a_j^{\phantom{\dagger}}\big\}\right)\\
-\rho_G\left(\big\{a_i^\dagger,a^\dagger_j\big\}\right)&\rho_G\left(\big\{a^\dagger_i,a_j\big\}\right)\\
\end{pmatrix}\ .
\ee
The $2\times2$ matrices along the diagonal represent single-mode covariance matrices, while the off-diagonal ones account for correlations among the various modes.

\section{Entanglement in Gaussian states}
\label{EML}

Using the previous results, and in particular
the quasi-free property of the maps $\Phi_t$, we want now to study 1) whether it is possible to generate mesosocopic entanglement between different chains entirely by means of the dissipative microscopic dynamics and further 2) 
investigate the fate of the generated entanglement in the course of time and of its dependence on the strength of the coupling with the environment and on the temperature of the given microscopic invariant state. 

By \textit{mesoscopic entanglement} we mean the existence of mesoscopic states carrying non-local, quantum correlations  among the fluctuation operators pertaining to different chains. 
More precisely, we shall focus on the creation and annihilation operators $a^\#_1$ and $a^\#_3$ that, as already observed before, are collective degrees of freedom attached to the first, second chain, respectively. We shall then study the time-evolution of two-mode Gaussian states $\rho^{(13)}$, obtained by tracing a full
four-mode Gaussian state over $a_2^\#$ and $a_4^\#$; indeed, as discussed below, the trace operation does not spoil the Gaussian character of the
initial four-mode states.

In the case of two-mode Gaussian states, the presence of entanglement can be ascertained using the partial
transposition criterion, {\it i.e.} by looking at their behaviour when $a_1$ and $a^\dag_1$ are exchanged  while keeping $a_1^\dag a_1$ and $a_1a^\dag_1$ unchanged and without touching $a_3$ and $a_3^\dag$. 
If under this substitution, $\rho^{(13)}$ does not remain positive, then it carries quantum correlations between the modes $1$ and $3$ and thus results entangled. 
Vice versa, a Gaussian state with respect to these two modes that remains positive under the above substitution is for sure separable.
This is the content of the so-called Simon entanglement criterion \cite{Simon}.

Notice that the state $\Omega_\beta$ in \eqref{qfs1} is separable with respect to all its four modes; indeed, its density matrix representation $R_\beta$ in \eqref{qfs2} can be written as a product of four independent density matrices one for each of the modes. Indeed, the corresponding covariance matrix $\widetilde\Sigma^{(\beta)}$ results diagonal when expressed in the representation \eqref{bigcov1}, \eqref{bigcov2}, thus showing neither quantum nor classical correlations between the different modes.

As initial states, we shall consider states that are obtained from $R_\beta$ by the action of suitable squeezing operators in the
modes 1 and 3, {\it i.e.} Gaussian states of the form
\begin{equation}
\rho^{(\beta)}_{r_1r_3}=S_1(r_1)S_3(r_3)\,R_\beta\, S^\dagger_3(r_3)S^\dagger_1(r_1)\ ,
\label{t0}
\end{equation}
where $S_j(r_j)$, $r_j\in\RR$, are single-mode squeezing operators such that
$$
S^\dag_j(r_j)\,a^\dagger_j\,S_j(r_j)=\cosh(r_j)\, a^{\dagger}_j\,-\,\sinh(r_j)\,a_j\ ,\quad j=1,3\ .
$$
The squeezing operators map displacement operators $D(z)$ in \eqref{displ1} into  displacement operators
$$
D(z')=S^\dag_3(r_3) 
S^\dag_1(r_1)\,D(z)\,S_1(r_1)S_3(r_3)\ ,
$$
where $z'=(z_1',z_2,z'_3,z_4)$ with $z'_{1,3}=\cosh(r_{1,3})z_{1,3}-\sinh(r_{1,3})\bar{z}_{1,3}$.
Further, the modes are not mixed by the squeezing so that 
$\rho^{(\beta)}_{r_1r_3}$ is also a 
separable Gaussian state relatively to all four modes. In particular, after squeezing, the $8\times 8$ covariance matrix 
$\widetilde\Sigma^{(\beta)}$ of the thermal state $R_\beta$ is mapped into the following one:
\be
\widetilde \Sigma^{(\beta)}_{r_1,r_3}=\frac{1}{2\epsilon}
\begin{pmatrix}
{\cal S}(r_1) & {\bf 0}_4\\
{\bf 0}_4 & {\cal S}(r_3)
\end{pmatrix}\ ,
\qquad
{\cal S}(r)=\begin{pmatrix}
\phantom{-}\cosh(2r)&-\sinh(2r)&0&0\\
-\sinh(2r)&\phantom{-}\cosh(2r)&0&0\\
0&0&1&0\\
0&0&0&1
\end{pmatrix} \ ,
\ee
where ${\bf 0}_4$ is the null matrix in four dimensions; in presenting this result, the ordering introduced
at the end of the previous Section (denoted by a tilde) has again been used, so that $\widetilde \Sigma^{(\beta)}_{r_1,r_3}$
takes a convenient block diagonal form.

Moreover, a state $\rho^{(13)}$ on the Bose algebra generated by $a_{1,3}^\#$ can be obtained from $\rho^{(\beta)}_{r_1r_3}$ by restricting its action on displacement operators of the form $D(z_{13})$ with $z_{13}=(z_1,0,z_3,0)$ and $z_{1,3}\in\CC$. 
Namely, $\rho^{(13)}$ is completely defined by the expectations
\be
\label{state13}
\rho^{(13)}\big(D(z_{13})\big)={\rm Tr}\left(\rho^{(\beta)}_{r_1r_3}\,D(z_{13})\right)=
{\rm Tr}\big[R_\beta\, D(z'_{13})\big]\ ,
\ee
and then inherits the Gaussian character of $R_\beta$ as these expectations are Gaussian functions of $z_{1,3}$. 
Finally, the same argument shows that the mesoscopic, dissipative time-evolution $\Phi_t$ transforms it in a Gaussian state at all times $t\geq 0$:
\be
\label{evolvGauss}
\rho^{(13)}_t\left(D(z_{13})\right)={\rm Tr}\Big[\rho^{(\beta)}_{r_1r_3}\,\Phi_t\big[D(z_{13})\big]\Big]\ .
\ee
Therefore the covariance matrix of interest, that involves only the modes $1,3$, can be retrieved from the total matrix in the form \eqref{bigcov2} by discarding the blocks relative to modes $2,4$. Explicitly,
\begin{equation}
\label{newcov}
\widetilde G_{red}(t)=
\begin{pmatrix}
\rho^{(13)}_t( a_1^{\dagger}a_1)+\frac{1}{2}&-\rho^{(13)}_t(a_1^{2})&\rho^{(13)}_t( a_1a_3^{\dagger})&-\rho^{(13)}_t( a_1a_3)\\
-\rho^{(13)}_t( a_1^{\dagger2})&\rho^{(13)}_t( a_1^{\dagger}a_1)+\frac{1}{2}&-\rho^{(13)}_t( a_1^{\dagger}a_3^{\dagger})&\rho^{(13)}_t( a_1^{\dagger}a_3)\\
\rho^{(13)}_t( a_1^{\dagger}a_3)&-\rho^{(13)}_t( a_1a_3)&\rho^{(13)}_t( a_3^{\dagger}a_3)+\frac{1}{2}&-\rho^{(13)}_t( a_3^{2})\\
-\rho^{(13)}_t( a_1^{\dagger}a^{\dagger}_3)&\rho^{(13)}_t( a_1a_3^{\dagger})&-\rho^{(13)}_t( a_3^{\dagger2})&\rho^{(13)}_t( a_3^{\dagger}a_3)+\frac{1}{2}
\end{pmatrix}
\equiv
\begin{pmatrix}
\Sigma_1&\Sigma_c\\
\Sigma_c^{\dagger}&\Sigma_2
\end{pmatrix}\ .
\end{equation}

For two mode-Gaussian states, the already mentioned Simon's criterion not only provides an exhaustive entanglement witness, but it also offers a means to quantify it \cite{Simon}.
It is nevertheless convenient to formulate the criterion in terms of the previous covariance matrix \cite{Souza}. Consider the block structure of $\widetilde G_{red}(t)$ and define:
\be
I_1=\det(\Sigma_1)\ ,\qquad I_2=\det(\Sigma_2)\qquad I_3=\det(\Sigma_c)\ ,\qquad
I_4=\tr\Big(\Sigma_1\sigma_3\Sigma_c\sigma_3\Sigma_2\sigma_3\Sigma_c^{\dagger}\sigma_3\Big)\ .
\label{crit-1}
\ee
Then, the necessary and sufficient condition for a state to be separable is:
\begin{equation}
S\equiv I_1I_2+\Big(\frac{1}{4}-|I_3|\Big)^2-I_4-\frac{(I_1+I_2)}{4}\ge0\ .
\label{crit}
\end{equation}
Taking real squeezing parameters $r_1,r_2$ for both chains, we have that $\Sigma_c=\Sigma_c^\dagger$; in this case, the
four quantities $I_j$ can be explicitly computed as shown in Appendix F. 

Further, the amount of entanglement in two-mode Gaussian states can be measured through the so-called logarithmic negativity of the state: 
\be
\label{entmeas1}
E=\max\left\{0,-\frac{1}{2}\log_2\left(4\, {\cal I}\right)\right\}\ ,\\
\ee
where
\be
\label{entmeas2}
{\cal I}=\frac{I_1+I_2}{2}-I_3-\sqrt{\left[\frac{I_1+I_2}{2}-I_3\right]^2-(I_1I_2+I_3^2-I_4)}\ .
\ee

\section{Spin chain models}

In the following we shall apply the theoretical tools developed so far to the study of the dissipative generation of mesoscopic entanglement in two different models: in the first one, 
the microscopic Lindblad generator contains contributions involving single-site
operators from both chains, while in the second one all terms contain single-site operators from one chain only.

\subsection{Model 1}

We shall consider a Lindblad generator of the form (\ref{LINDMICO0a})-(\ref{LINDMICO0c}), with Hamiltonian term
\be
\label{mod1H}
\HH_N[X]=-i\big[H_N,\, X\big]\ ,\qquad
H_N=\frac{\eta}{2}\sum_{k=0}^{N-1} h^{(k)}\ ,\quad h^{(k)}=\sigma_3^{(k)}\otimes \bold{1}^{(k)}\,+\,\bold 1^{(k)}\otimes\sigma_3^{(k)}\ ,
\ee
and dissipative contribution of the generic form \eqref{LINDMICO0c},
\be
\DD_N[X]=\frac{1}{2}\sum_{k,\ell=0}^{N-1}J_{k\ell}\sum_{\mu,\nu=1}^4\,D_{\mu\nu}\Big(v_\mu^{(k)}\,\left[X\,,\,(v_{\nu}^\dag)^{(\ell)}\right]\,+\,\left[v_{\mu}^{(k)}\,,\,X\right]\,(v_\nu^\dag)^{(\ell)}\Big)\ ,
\label{mod1-dissip}
\ee
with the following single-site Kraus operators
\be
\label{Krops}
v_1=\sigma_+\otimes \sigma_-\ ,\quad v_2=\sigma_-\otimes \sigma_+\,,\quad
v_3=\frac{1}{2}\big(\sigma_3\otimes \bold{1}\big)\,,\quad v_4=\frac{1}{2}\big(\bold{1}\otimes \sigma_3\big)\ ,
\ee
where $\sigma_{\pm}=(\sigma_1\pm i\,\sigma_2)/2$, while the $4\times 4$ matrix $D$ is given by
\be
D=\begin{pmatrix}
\delta&0&\gamma&\gamma\\
0&\delta&\gamma&\gamma\\
\gamma&\gamma&\delta&0\\
\gamma&\gamma&0&\delta
\end{pmatrix}\ ;
\label{Kosmat}
\ee
by choosing $|\gamma|\le \delta/2$, $D$ results positive semi-definite.
In this case, one can recast $\DD_N$ in a double commutator form:
\be
\DD_N[X]=\frac{1}{2}\sum_{k,\ell=0}^{N-1}J_{k\ell}\sum_{\mu,\nu=1}^4D_{\mu\nu}\Big[\left[v_\mu^{(k)}\,,\,X\right]\,,\,(v^\dag_{\nu})^{(\ell)}\Big]\ .
\label{double}
\ee
In the following we shall study the emergent mesoscopic dynamics corresponding to the microscopic dissipative dynamics locally generated by
$\LL_N[X]=\HH_N[X]+\DD_N[X]$ as given above.

Local states $\rho_N$ evolve according to the master equation involving the dual generator 
$\LL_N^{\phantom{|}\star}$:
\be
\partial_t \rho_N(t)=\LL_N^{\star\phantom{|}}[\rho_N(t)]=-i\big[H_N,\,\rho_N(t)\big]+ \DD_N[\rho_N(t)]\ .
\ee
The microscopic thermal state 
$$
\rho_N^{(\beta)}=\bigotimes_{j=0}^{N-1}\frac{1}{4\cosh^2(\eta\beta/2)}\,{\rm e}^{-\beta \eta h^{(k)}/2}\ ,
$$
in \eqref{STATE} is left invariant by the dissipative dynamics; indeed, 
$\LL_N^\star[\rho_N^{(\beta)}]=0$, as it follows from
$$
\left[\sigma_3\otimes\bold{1}+\bold{1}\otimes\sigma_3\,,\,v_\mu\right]=0\qquad \forall \mu=1,2,3,4\ .
$$ 

Further, since spin operators at different sites commute, given the Lindblad generator $\LL_N$, its action on the self adjoint element $x_i^{(k)}$ from the set $\chi$
at site $k$ is given by:
\beann
\LL_N\left[x_i^{(k)}\right]&=&i\frac{\eta}{2}\left[\sigma_3^{(k)}\otimes\bold{1}+\bold{1}\otimes \sigma_3^{(k)}\,,\,x_i^{(k)}\right]\\
&+&J_0\sum_{\mu,\nu=1}^4\frac{D_{\mu\nu}}{2}\left[\left[v_\mu^{(k)}\,,\,x_i^{(k)}\right]\,,\,(v^\dag_{\nu})^{(k)}\right]\ .
\eeann
This action maps the linear span $\chi$ in itself; indeed, $\LL_N\left[x_i^{(k)}\right]=\sum_{j=1}^8\mathcal{L}_{ij}\,x_j^{(k)}$, with the $8\times 8$ matrix $\mathcal{L}=\mathcal{H}+\mathcal{D}$ explicitly given in Appendix D.

Then, the generator of the mesoscopic dissipative dynamics as given in {\sl Corollary \ref{cor1}} is completely determined by the $8\times 8$ matrices $H^{(1)}$ and $D^{(1)}$ in \eqref{Flind1a}, \eqref{Flind1b} or $H^{(2)}$ and $D^{(2)}$ in \eqref{Flind2a}.
Here, we give the form of the generator with respect to
creation and annihilation operators.
 
\begin{proposition}
\label{propo1}
In terms of annihilation and creation operators $a^\#_i$, $i=1,2,3,4$, the mesoscopic Lindblad generator acts on displacement operators $D(z)$ as $\LL=\HH\,+\,\DD$, with
$\HH$ and $\DD$ given by
\bea
\label{LINDBOS0}
\HH[D(z)]&=&i\eta\Big[\sum_{j=1}^4a^\dag_j a_j\,,\,D(z)\Big]\\
\DD[D(z)]&=&\sum_{i,j=1}^{8}K_{ij}^{(\beta)}\left(V^\dag_i\,D(z)\,V_j\,-\,\frac{1}{2}\left\{V^\dag_i\,V_j\,,\,D(z)\right\}\right)\ ,
\label{LINDBOS}
\eea
where $V=(a_1,a_2,a^\dag_1,a^\dag_2,a_3,a_4,a^\dag_3,a^\dag_4)^{tr}$ and Kossakowski matrix 
\bea
\label{gen0}
\hskip-1cm
&&
K^{(\beta)}=\frac{J_0}{\epsilon}\begin{pmatrix}A_\beta&B_\beta\cr B_\beta&A_\beta\end{pmatrix}\ ,\quad 
A_\beta=\delta\,\begin{pmatrix}
1+\epsilon&0&0&0\cr   
0&1+\epsilon&0&0\cr    
0&0&1-\epsilon&0\cr    
0&0&0&1-\epsilon
\end{pmatrix}\\ 
\hskip-1cm
\label{gen2a}
&&
B_\beta=\gamma\begin{pmatrix}
\epsilon(1+\epsilon)&-(1+\epsilon)c&0&0\cr
-(1+\epsilon)c&-\epsilon(1+\epsilon)&0&0\cr
0&0&\epsilon(1-\epsilon)&-(1-\epsilon)c\cr
0&0&-(1-\epsilon)c&-\epsilon(1-\epsilon)
\end{pmatrix}\ ,
\eea
where $\epsilon=\tanh(\eta\beta/2)$ and $c=\sqrt{1-\epsilon^2}$ as before.
\end{proposition}

\begin{proof}
The Hamiltonian contribution $\HH$ to the generator is defined by the matrix $H^{(2)}$ in equation \eqref{H1A}
of Appendix D: it is diagonal in the operators $A_i^\#$, defined in (\ref{anncrop}).
Moreover, $A^\dag_{5,6,7,8}=A_{1,2,3,4}$; thus, by using the canonical commutation relations $[a_i\,,\,a^\dag_j]=\delta_{ij}$, the mesoscopic Hamiltonian results proportional to the number operator $\sum_{j=1}^4a^\dag_j a_j$. 

The form of the dissipative term $\DD$ in the generator derives from the expression of the Kossakowski matrix given in equations \eqref{D1A1} and \eqref{D1A2} of Appendix D. Using {\sl Corollary~1}, the form (\ref{LINDBOS}) then follows;
note that, for convenience, the sums over the indices $i,j$ in (\ref{LINDBOS}) use the ordering
$(a_1,a_2,a^\dag_1,a^\dag_2,a_3,a_4,a^\dag_3,a^\dag_4)^{tr}$ instead of
$(a_1,a_2,a_3,a_4,a^\dag_1,a^\dag_2,a^\dag_3,a^\dag_4)^{tr}$ introduced before.
\qed
\end{proof}
\medskip

\begin{remark}
{\rm From the above expression of the Lindblad generator there emerge two main features
of the mesoscopic dissipative dynamics: 1) the unitary contribution $\HH$ to the collective dynamics of the Boson degrees of freedom shows no interactions among them.
The mesoscopic Hamiltonian is proportional to the number operator and as such it does commute with the dissipative contribution: $\DD\circ\HH=\HH\circ\DD$.  
In fact, $\DD$ is gauge-invariant, it does not change by sending $a_i$ into ${\rm e}^{i\phi}a_i$ and $a_i^\dag$ into ${\rm e}^{-i\phi}a^\dag_i$, $i=1,2,3,4$. Furthermore, 2) were it not for the off-diagonal blocks $B_\beta$ in the Kossakowski matrix, the dissipative dynamics would correspond to decaying process affecting independently the various bosonic degrees of freedom. For instance, in absence of off-diagonal terms in the Kossakowski matrix, one would have
$$
\LL[a_i]=-\left(i\omega\,+\,J_0\delta\right)\,a_i\ .
$$ 
Instead, the presence of $B_\beta\neq 0$ statistically couples the collective operators, $a^\#_{1,3}$, $a^\#_{2,4}$  referring to different chains.}
\qed
\end{remark}

\subsection{Model 2}

While the Lindblad operators $v$'s of the first model involve contributions from both chains ({\it c.f.} \eqref{Krops})
and different sites are statistically coupled by the coefficients $J_{k\ell}$, in the following we shall consider a Lindblad generator 
with the same Hamiltonian term as in (\ref{mod1H}),
and a diagonal dissipative contribution of the form:
\be
\label{modD2a}
\DD_N[X]=\sum_{k=0}^{N-1}\DD^{(k)}[X]\ ,\quad
\DD^{(k)}_N[X]=\sum_{\mu,\nu=1}^{6}D_{\mu\nu}\left(v^{(k)}_\mu\, X\,
v_\nu^{(k)}-\frac{1}{2}\left\{v^{(k)}_\mu\,v^{(k)}_\nu\,,\,X\right\}\right)\ ,
\ee
with self-adjoint Lindblad operators,
\be
\label{Krausopa}
v_{1,2,3}=\sigma_{1,2,3}\otimes\bold{1}\ ,\quad v_{4,5,6}=\bold{1}\otimes\sigma_{1,2,3}\ ,
\ee
and $6\times 6$ Kossakowski matrix $D$ given by
\be
\label{Kossmat2}
D=\begin{pmatrix}M&M\cr M&M\end{pmatrix}\ ,\quad M=
\begin{pmatrix}
1&-i\epsilon&0\\
i\epsilon&1&0\\
0&0&\xi
\end{pmatrix}\ ,
\end{equation}
where the conditions $\xi\ge0$ and $\epsilon=\tanh(\eta\beta/2)\leq 1$ guarantee $D\geq 0$. 
Because of the symmetry of the Kossakowski matrix, each single site contribution to the Lindblad generator can be recast in the simpler form: 
\bea
\label{Lind2a}
\DD^{(k)}_N[X]&=&\sum_{\mu,\nu=1}^{3}M_{\mu\nu}\left(w^{(k)}_\mu\,X\,w^{(k)}_\nu-\frac{1}{2}\left\{w^{(k)}_\mu\,w^{(k)}_\nu\,,\,X\right\}\right)\\
\nonumber
&&\hskip-2cm
=\frac{1}{2}\Big(\left[w^{(k)}_1\,,\,\left[X\,,\,w^{(k)}_1\right]\right]+
\left[w_2^{(k)}\,,\,\left[X\,,\,w^{(k)}_2\right]\right]+\gamma\left[w_3^{(k)}\,,\left[X\,,\,w^{(k)}_3\right]\right]\Big)\\
\label{Lind2c}
&-&i\frac{\epsilon}{2}\,\left\{w^{(k)}_1\,,\,\left[X\,,\,w^{(k)}_2\right]\right\}+i\frac{\epsilon}{2}\,\left\{w_2^{(k)}\,,\,\left[X\,,\,w^{(k)}_1\right]\right\}
\eea
with operators $w_\mu=\sigma_\mu\otimes\bold{1}+\bold{1}\otimes\sigma_{\mu}$ obeying
\bea
\label{Pauli1}
\left[w_j\,,\,w_k\right]&=&2i\epsilon_{jk\ell}\,w_\ell\\
\label{Pauli2}
\left\{w_j\,,\,w_k\right\}&=&\sigma_j\otimes\sigma_k\,+\,\sigma_k\otimes\sigma_j\,+\,i\epsilon_{jk\ell}\left(\sigma_\ell\otimes\bold{1}-\bold{1}\otimes\sigma_\ell\right)\ .
\eea
In the Schr\"odinger picture, the local spin states $\rho_N$ evolve in time according to the dual generator 
$\LL_N^{\star\phantom{|}}=\left(\HH_N^{\,\star\phantom{|}}+\DD_N^{\,\star\phantom{|}}\right)$ where
\beann
\hskip-1cm
\HH_N^{\,\star\phantom{|}}[\rho_N]&=&-i\eta\sum_{k=0}^{N-1}\left[w_3^{(k)}\,,\,\rho_N\right]\ ,\quad
\DD_N^{\,\star\phantom{|}}[\rho_N]=\sum_{k=0}^{N-1}\left(\DD^{(k)}\right)^{\star\phantom{|}}[\rho_N]\ ,\\
\left(\DD^{(k)}_N\right)^{\star\phantom{|}}[\rho_N]&=&\sum_{\mu,\nu=1}^{3}M_{\mu\nu}\left(w^{(k)}_\nu\,\rho_N\,w^{(k)}_\mu-\frac{1}{2}\left\{w^{(k)}_\mu\,w^{(k)}_\nu\,,\,\rho_N\right\}\right)\\
\hskip-1cm
&=&\frac{1}{2}\sum_{\mu=1}^2\left[w^{(k)}_\mu,\left[\rho_N,w^{(k)}_\mu\right]\right]+\gamma\left[w_3^{(k)},\left[w^{(k)}_3,\rho_N\right]\right]\\
&&\hskip-2cm
+i\frac{\epsilon}{2}\,\left\{w^{(k)}_1,\left[\rho_N,w^{(k)}_2\right]\right\}-i\frac{\epsilon}{2}\,\left\{w_2^{(k)},\left[\rho_N,w^{(k)}_1\right]\right\}-2\epsilon\left\{w_3,\rho_N\right\}\ .
\eeann
In terms of the operators $w_\mu$, the microscopic state $\rho^{(\beta)}_N$ in \eqref{STATE} is the tensor product of $N$ density matrices of the form
$$
\frac{1}{4\cosh^2(\frac{\eta\beta}{2})}\, \exp\left(-\frac{\eta\beta}{2} w_3\right)\ .
$$
Expanding the exponential and using \eqref{Pauli2} with $j=k=3$ one gets:
$$
\rho_N^{(\beta)}=\bigotimes_{k=0}^{N-1}\frac{1}{4}\left(\bold{1}\,-\,\epsilon\,w^{(k)}_3+\epsilon^2\sigma^{(k)}_3\otimes\sigma^{(k)}_3\right)\ ,\qquad \epsilon=\tanh\left(\frac{\beta\eta}{2}\right)\ .
$$
By explicit computation one then checks that $\LL_N^{\star\phantom{|}}\big[\rho^{(\beta)}_N\big]=0$, whence the microscopic local states are left invariant by the microscopic dissipative dynamics.
This fact is one of the two conditions for applying the results of the previous sections; the other condition is that the action of the local generator $\LL_N$ maps into itself the linear span $\mathcal{X}$ of the elements $x_j\in\chi$ in \eqref{matrix},\eqref{matrixa}.
This is verified in Appendix E. Finally, as for the first model, it is sufficient to explicitly write
the generator of the quasi-free mesoscopic semigroup emerging from the above microscopic dissipative dynamics
in the language of creation an annihilation operators:
 
\begin{proposition}
\label{propo3}
In terms of annihilation and creation operators $a^\#_i$, $i=1,2,3,4$, the mesoscopic Lindblad generator reads $\LL=\HH\,+\,\DD$, where the action of $\HH$ and $\DD$ on displacement operators $D(z)$ is as in (\ref{LINDBOS0}) and (\ref{LINDBOS}),
where the Kossakowski matrix now reads 
\bea
\label{gen20}
K_\beta&=&\frac{2}{\epsilon}\begin{pmatrix}(1+\epsilon)M_\beta&0&(1+\epsilon)N_\beta&0
\cr 0&(1-\epsilon)M_\beta&0&(1-\epsilon)N_\beta\cr
(1+\epsilon)N_\beta&0&(1+\epsilon)M_\beta&0\cr
0&(1-\epsilon)N_\beta&0&(1-\epsilon)M_\beta\end{pmatrix}\\ 
\label{gen21} 
M_\beta&=&\,\begin{pmatrix}
1+\xi&0\cr   
0&3+\xi\end{pmatrix}\ ,\quad 
N_\beta=\begin{pmatrix}
\epsilon^2&-\epsilon c\cr
-\epsilon c&1+c^2
\end{pmatrix}\ ,
\eea
again with $\epsilon=\tanh(\eta\beta/2)$, $c=\sqrt{1-\epsilon^2}$.
\end{proposition}

\noindent
The proof is very similar to the one discussed for the previous model and it is based on 
{\sl Corollary 1} and the results of Appendix E.

Though the details are different, the structure of the Kossakowski matrix is similar to the one in Model 1, so that again the Hamiltonian contribution $\HH$ to the mesoscopic Lindblad generator commutes with the dissipative one.
Moreover, also in this case, the off-diagonal elements of the Kossakowski matrix statistically couple the mesoscopic operators $a^\#_{1,3}$, $a^\#_{2,4}$ referring to different chains.

\section{Environment induced mesoscopic entanglement}

Given the results of the previous Section, one can now study whether the mesoscopic dissipative time-evolutions in Model 1 and 2 can give rise to mesoscopic entanglement between the two independent chains, and, if yes, analyze 
the fate of the generated entanglement in the course of time.

\subsection{Entanglement Dynamics: Model 1}
\label{EDPTM}

In this case the entanglement criterion (\ref{crit}) can be studied analytically:
we will show that the two spin chains can indeed become mesoscopically entangled, and relate the behaviour of these bath-induced quantum correlations to the squeezing parameters, the parameter $\gamma$ and the temperature associated to
the initial microscopic state. For sake of simplicity, we shall further
set $\delta=J_0=\eta=1$, since these parameters do not play any role in the discussion that follows.

The behaviour in time of the logarithmic negativity $E$, introduced in (\ref{entmeas1}), is shown in Fig.\ref{GAMMA} for
different values of the dissipative parameter $\gamma$ appearing in the Kossakowski matrix and fixed initial temperature $T$.
Both a ``symmetrically squeezed'', with $r_1=r_3=r$, and ``one-mode squeezed'', with $r_1=r$, $r_3=\,0$, initial state
have been studied; however, since similar results hold for both cases, only the graphs relative to the symmetric squeezed case will be shown.
From the behaviour of $E$, one clearly sees
that the two infinite spin chains get entangled by the dynamics. Since the Hamiltonian does not contain coupling terms,
this entanglement is due solely to the mixing effects of the environment within which the two spin chains are embedded.
Moreover, the amount of created entanglement increase as the dissipative parameter $\gamma$ gets larger,
while a non-zero entanglement appears earlier in time.
\begin{figure}[t]
\center\includegraphics[scale=0.65]{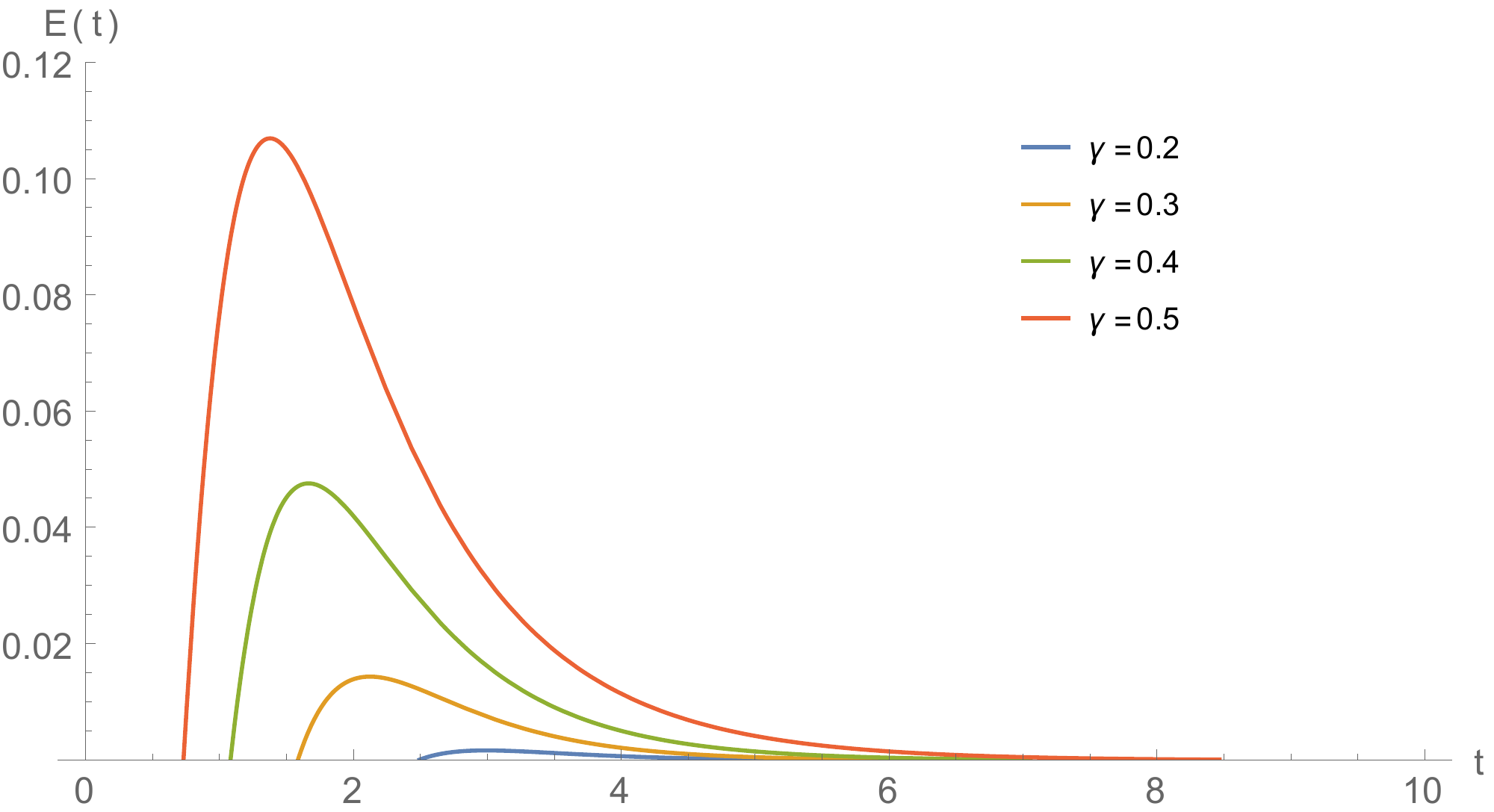}
\caption{\small Model 1: behaviour in time of the logarithmic negativity $E$ for different values of $\gamma$ at fixed temperature
$T=0.1$, for a symmetrically squeezed initial state with $r_1=r_3=r=1$.}
\label{GAMMA}
\end{figure}
\begin{figure}[h!]
\center\includegraphics[scale=0.65]{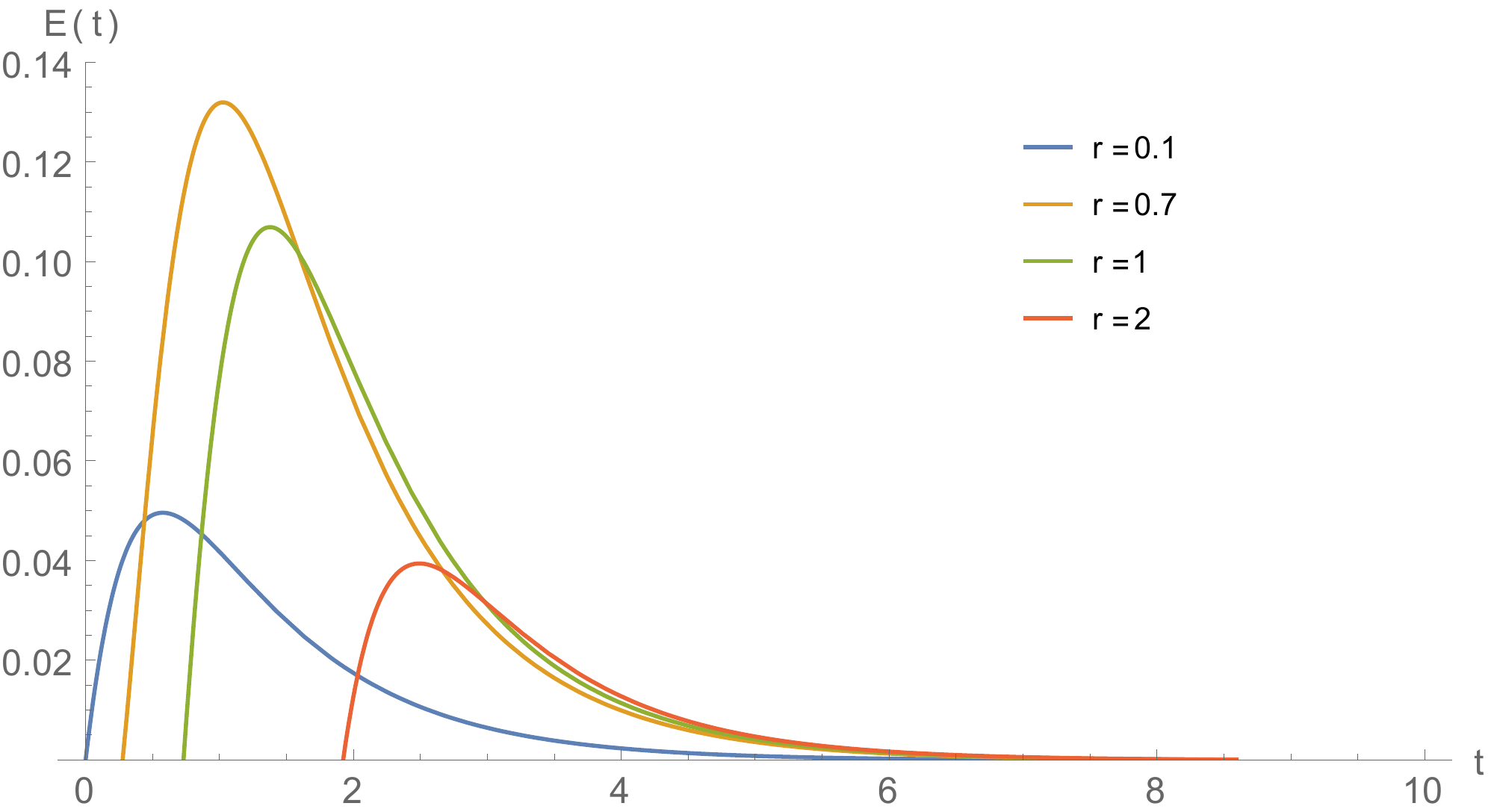}
\caption{\small Model 1: behaviour in time of the logarithmic negativity $E$ for different values of the squeezing parameter
$r=r_1=r_3$, at fixed temperature
$T=0.1$ and dissipative parameter $\gamma=1/2$.}
\label{SQUEEZE}
\end{figure}\
\begin{figure}[h!]
\center\includegraphics[scale=0.65]{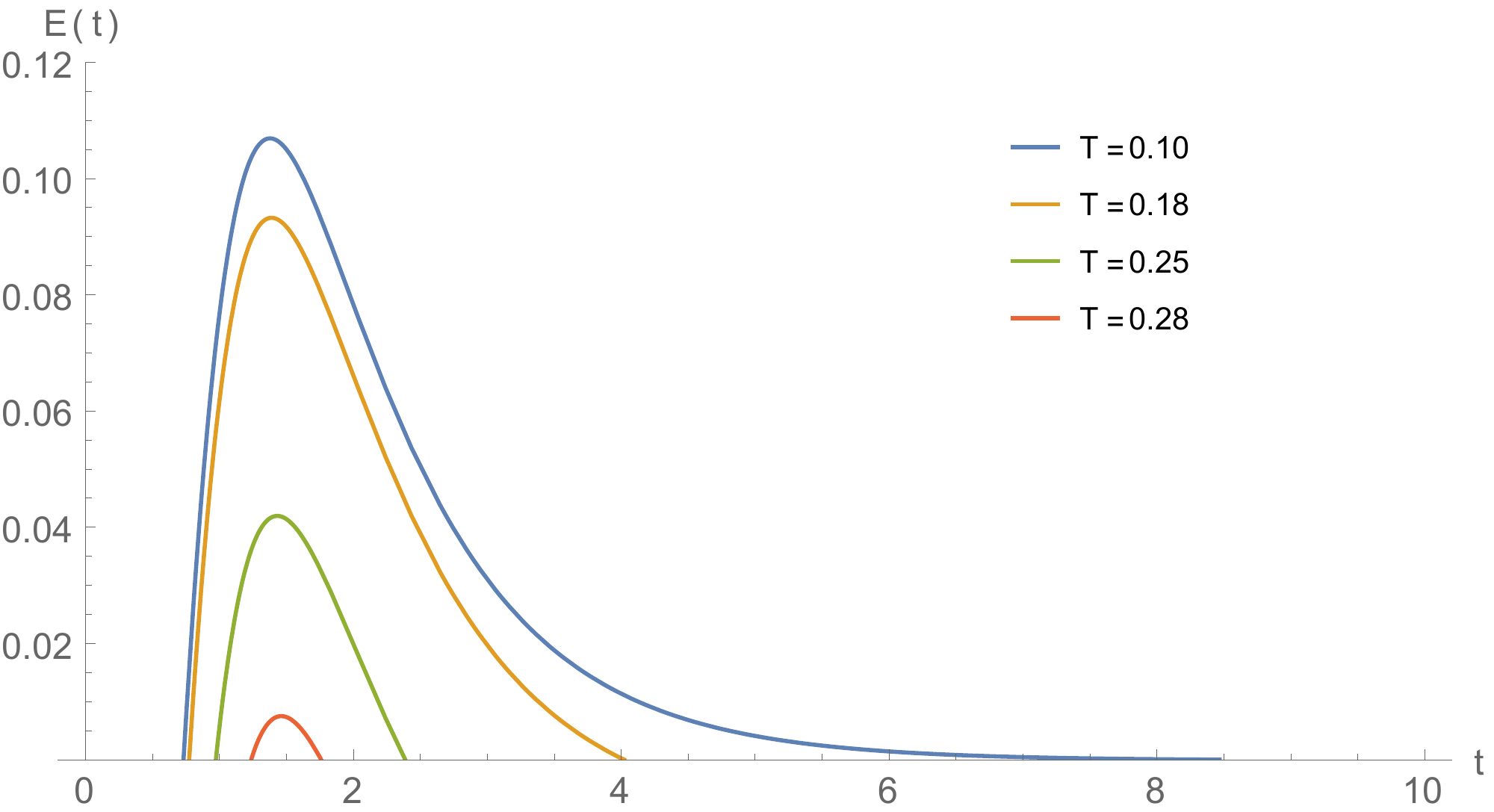}
\caption{\small Model 1: behaviour in time of the logarithmic negativity $E$ for different values of the temperature $T$, at fixed dissipative parameter $\gamma=1/2$ and squeezing $r_1=r_3=r=1$.}
\label{TEMP}
\end{figure}

Also the amount of squeezing plays an essential role; while a non-vanishing squeezing appears necessary to create quantum correlations, too much squeezing decreases the maximum value of $E$. Squeezing also influences the time at which it is
first generated. Further, for fixed $T$ and $\gamma$, there is a value of the squeezing parameter $r$ allowing for a maximal value of $E$. All this is explicitly shown in Fig.\ref{SQUEEZE}.

Finally, the effect of the temperature is displayed in Fig.\ref{TEMP}, for fixed dissipative and squeezing
parameters. One sees that
increasing the temperature, the maximum of the logarithmic negativity $E$ decreases, indicating that there exists a critical temperature $T_C$, above which no entanglement is possible.

The explanation of this result can be traced to the behaviour of the quantity $S$ appearing in the
separability criterion in (\ref{crit}). In Appendix F, this quantity has been explicitly computed both for the case
of a symmetrically squeezed initial state, see (\ref{critsas0}), and one-mode squeezed initial state, see (\ref{critsas}).
For large temperatures, the parameter $\epsilon$ becomes small, 
so that all terms but those proportional to $1/\epsilon^4$ can be neglected, obtaining in the two cases:
\begin{eqnarray*}
&&
S_{S}(t)\sim\frac{1}{16\epsilon^4}\left(1+8\sinh^2(r)\left( y_1(t)-y_1^2(t)\right)\right)\ ,\\
&&
S_A(t)\sim\frac{1}{16\epsilon^4}\left(1+4\sinh^2(r)\left(y_1(t)-y_1^2(t)\right)\right)\ ,
\end{eqnarray*}
where $y_1(t)$ is given in (\ref{y1}) of Appendix F. Notice that since $y_1(t)<1$ for $t>0$, these two
quantities are always positive; therefore, there must be a finite ``critical temperature'' $T_C$ 
beyond which entanglement is no longer present.

\begin{figure}[h!]
\center
\subfigure{\includegraphics[scale=0.35]{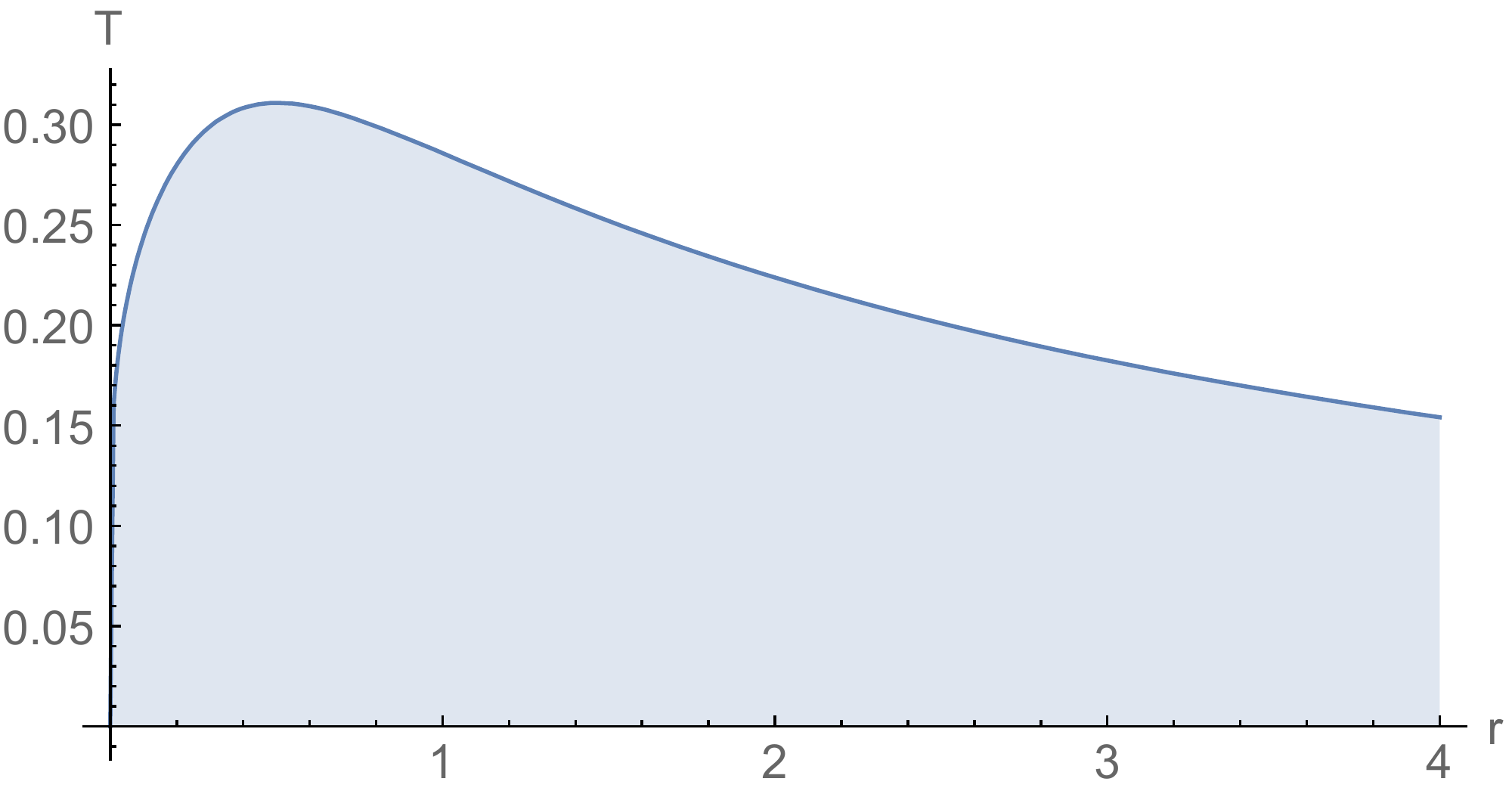}}\qquad
\subfigure{\includegraphics[scale=0.35]{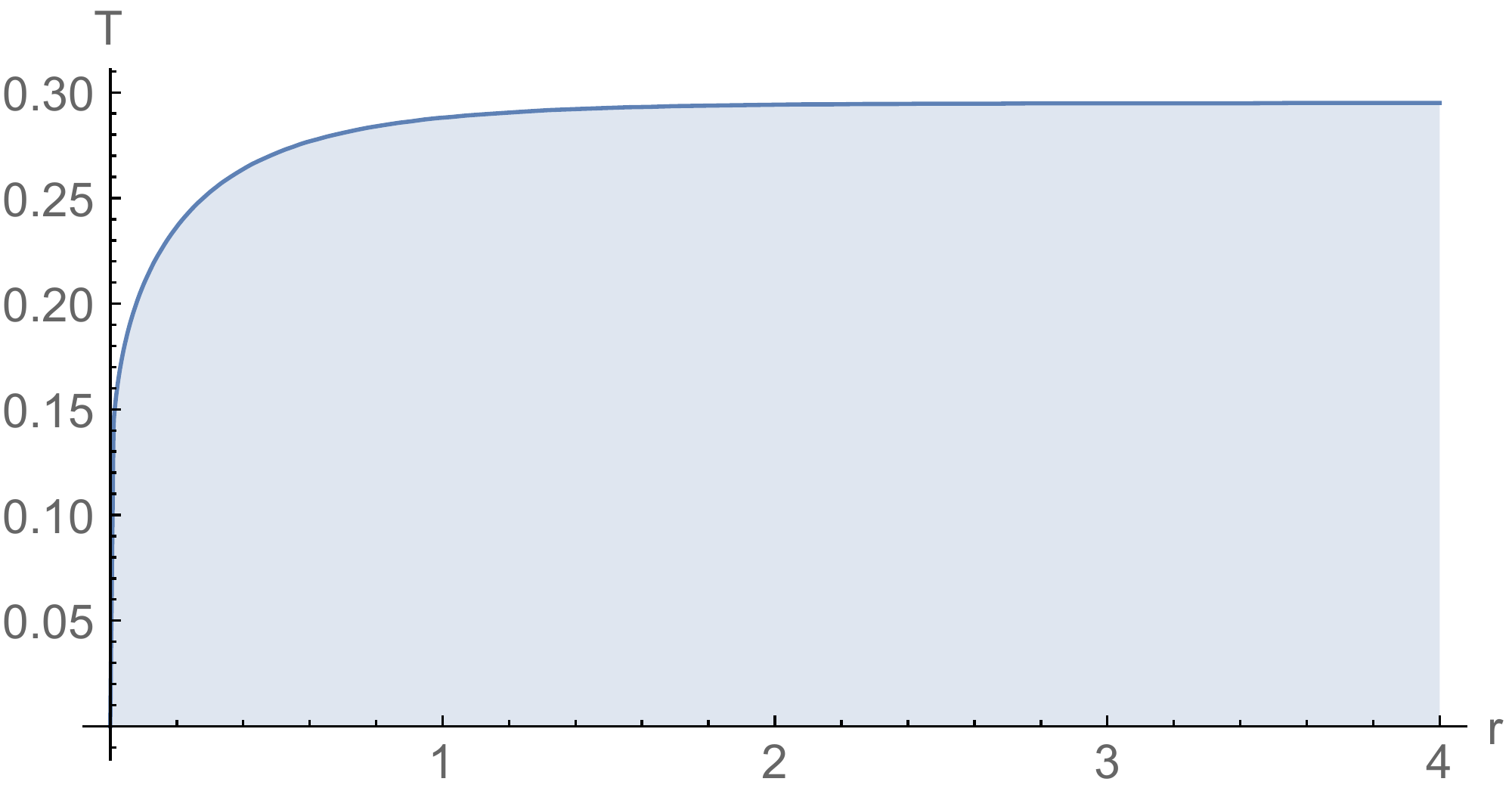}}
\caption{\small Model 1: entanglement phase diagrams for the symmetrically squeezed state $r=r_1=r_3$ (left)
and one-mode squeezed state $r=r_1$, $r_3=0$ (right), with $\gamma=1/2$;
the line separating the two regions gives the behaviour of the critical temperature $T_C$ as a function of $r$.}
\label{phase}
\end{figure}

This result is further illustrated by Fig.4, where the points in the $(r,T)$ plane 
with non-vanishing mesoscopic entanglement are highlighted. These figures show two regions,  
the dark ones associated with a non-vanishing maximal value of $E$, the brighter ones with
vanishing maximal value of $E$ and therefore no entanglement. The line separating the two regions determines the ``critical temperature'' $T_C$, above which entanglement among the two chains is not possible, 
as a function of the squeezing parameter; it is defined implicitly by the condition ${\rm max}\big(E(r,T)\big)=0$,
where the maximization is over all times.



\subsection{Entanglement sudden birth and sudden death}
\label{suddendeath}

The time behaviour of the logarithmic negativity $E$ reported in Fig.'s \ref{GAMMA},\ref{SQUEEZE},\ref{TEMP}
shows the phenomena of the
so-called ``sudden birth'' and ``sudden death'' of entanglement \cite{Eberly}, {\it i.e.} the sudden generation
of entanglement only after a finite time since the starting of the dynamics, 
and the abrupt vanishing of it at a later, finite time. These two effects can be analyzed in detail as function
of the temperature $T$ of the initial state.

Let us first consider the phenomenon of sudden death and accordingly look at the large $t$ behaviour of
evolved initial Gaussian state. As discussed before, the asymptotic state of the dynamics generated by
(\ref{LINDBOS0}) and (\ref{LINDBOS}) is thermal, with a reduced covariance in the modes $a_1$, $a_3$
given by (see Appendix F):
$$
\widetilde G_{red}^\infty\equiv\lim_{t\to\infty}\widetilde G_{red}(t)=\frac{1}{2\epsilon}\bold{1}_{4}\ .
$$
Positivity of the asymptotic state requires ({\it c.f.} (\ref{g-positivity})):

\be
\label{ineq}
\widetilde G_{red}^\infty+\frac{i}{2}\tilde\sigma\geq0\ ,\quad 
\tilde\sigma=-i\begin{pmatrix}
\sigma_3 & 0\\ 0 & \sigma_3
\end{pmatrix}\ ,
\ee
where $\tilde\sigma$ is the symplectic matrix in the reduced $a_1$, $a_3$ representation. This condition
assures also the positivity of the partially transposed state, since $\widetilde G_{red}^\infty$ is left invariant
by this transformation. In fact, the large time asymptotic limit of the 
lowest eigenvalue $\lambda_{min}(t)$ of the matrix $\displaystyle \widetilde G_{red}(t)+\frac{i}{2}\tilde\sigma$ is given by $\displaystyle\lambda_{min}^\infty=\frac{1-\epsilon}{2\epsilon}$, which is always
strictly positive, except at zero temperature ($\epsilon=1$) when it vanishes.
Therefore, when $T>0$, the bath generated entanglement must always vanish in finite times, since
$\lambda_{min}(t)$, from being negative, becomes strictly positive for $t\to\infty$.
Only at $T=0$ the created entanglement may vanish asymptotically.

In order to study the phenomenon of sudden birth of entanglement, one has to analyze the behaviour
of the logarithmic negativity $E$ in a right neighborhood of $t=0$. Let us consider first the case
of the symmetrically squeezed initial state. Using (\ref{critsas0}) in Appendix~F, one checks that
$$
\lim_{t\to{0^+}}S_{S}(t)=\frac{(1-\epsilon^2)^2}{16\epsilon^4}\ge0\ .
$$
This result already shows that only at zero temperature ($\epsilon=1$) there is the possibility of having 
generation of entanglement as soon as the dynamics starts. In fact, at $T=0$ one has:
\be
S_{S}^{T=0}(t)=\sinh^4(r)\Big(e^{-8t}-2e^{-6t}\cosh(2\gamma t)+e^{-4t}\Big)-e^{-4t}\sinh^2(2\gamma t)\sinh^2(r)\ .
\label{T0}
\ee
Since its first derivative with respect to $t$ vanishes at $t=0$, one needs to study the behaviour of its
second derivative:
$$
\frac{d^2}{dt^2}S_{S}^{T=0}(t)\Big|_{t=0}=8\big[\sinh^4(r)(1-\gamma^2)-\sinh^2(r)\gamma^2\big]\ .
$$
Since $S_{S}^{T=0}(t)=0$, there can be entanglement generation as soon as $t>0$ only if this quantity is negative,
{\it i.e.} only when  $\sinh^2(r)<\gamma^2/(1-\gamma^2)$. In the opposite case, as well as for
$T>0$, entanglement generation can occur only through the sudden creation phenomenon.

Similarly, in the case of a single mode squeezed initial state, $r_1=r$, $r_3=0$, from (\ref{critsas}) of Appendix F, 
we have:
$$
\lim_{t\to{0^+}}S_{A}(t)=\frac{(1-\epsilon^2)^2}{16\epsilon^4}\ge0\ .
$$
Therefore, also in this case, the system may become entangled as soon as $t>0$ only at zero temperature.
Indeed, one has
\be
S_{A}^{T=0}(t)=-\sinh^2(r)\frac{e^{-4t}\sinh^2(2\gamma t)}{16}\ ,
\ee
which is always negative, vanishing only at $t=0$, so that indeed entanglement is created as soon as $t>0$.
On the other hand, the phenomenon of sudden creation of entanglement always occur for $T>0$.

Concerning the behaviour of the critical temperature $T_C$ for large squeezing parameter~$r$, 
the graph on the left part of Fig.4 suggests a vanishing value for $T_C$, 
while that on the right a constant value, independent from $r$.
Indeed, in the first case, recalling the result (\ref{T0}) above, one sees that for $T=0$ and
$\gamma=1/2$, {\it i.e.} the largest admissible value for the dissipative parameter $\gamma$,
one gets for large $r$:
\be
S_{S}^{T=0}(t) \simeq e^{4(r-t)}\Big(1-e^{-3t}\Big) \Big(1-e^{-t}\Big)\ ,
\ee
which is always non negative. This means that in the limit $r\to\infty$, no entanglement is created at any time
when $T=0$. The critical temperature $T_C$ must therefore approach zero in the same limit.

Instead, in the other case one finds that for large squeezing parameter:
\be
S_{A}(t)\simeq e^{2r}\, g(t,T)\ ,
\ee
where $g(t,T)$ is the function multiplying $\sinh^2(r)$ in (\ref{critsas}).
One can show that this function takes negative values for some $t$, {\it i.e.} entanglement is generated,
only for temperatures below a certain fixed value $\bar T$, which can be computed only numerically.
As shown by the graph in the right part of Fig.4, 
the critical temperature is thus always non vanishing, reaching the asymptotic value
$\bar T$ for large squeezing.

\subsection{Entanglement Dynamics: Model 2}

While in Model 1 the microscopic dynamics is generated by a Lindblad term involving contributions 
from both chains and also different sites, the dissipative generator (\ref{modD2a}) of Model 2 contains 
only single chain Lindblad operators, and further without any statistical coupling between different sites.

\begin{figure}[hbtp]
\center\includegraphics[scale=0.57]{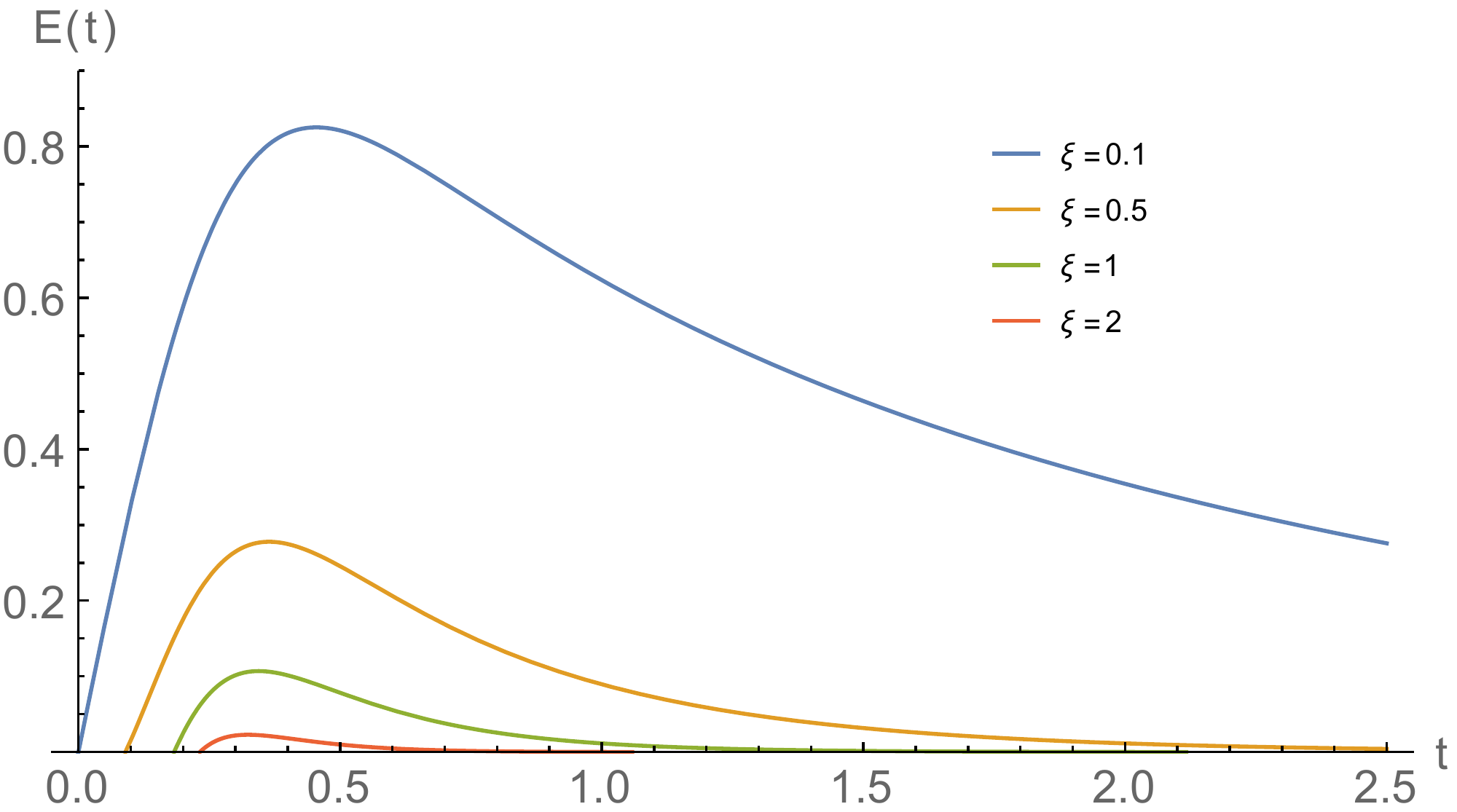}
\caption{\small Model 2: behaviour in time of the logarithmic negativity $E$ for different values of the dissipative parameter $\xi$, at fixed temperature $T=0.1$ and squeezing $r=r_1=r_3=1$.}
\label{UNDYN1}
\end{figure}

\begin{figure}[hbtp]
\center\includegraphics[scale=0.57]{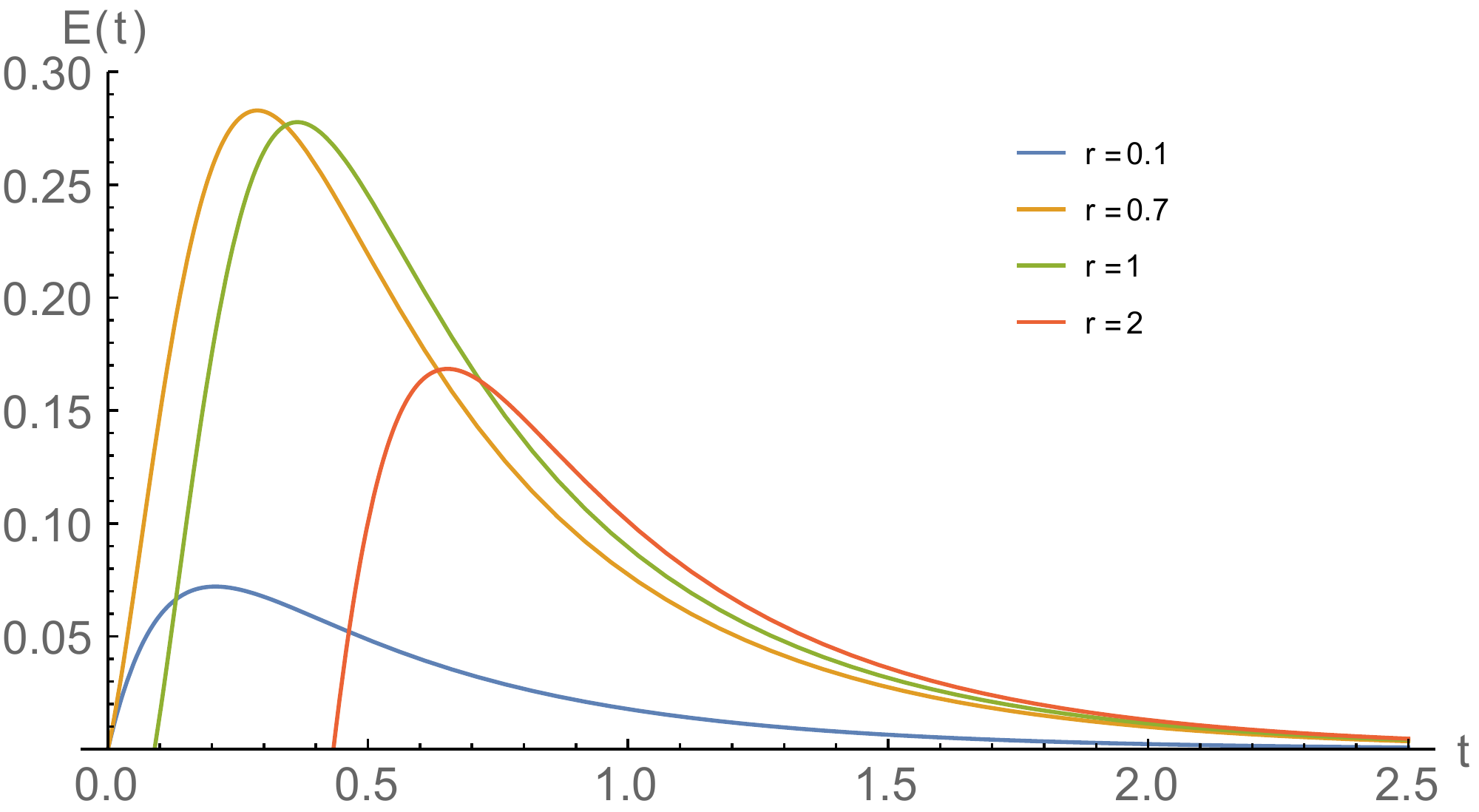}
\caption{\small Model 2: behaviour in time of the logarithmic negativity $E$ for different values of the temperature $T$, for $\xi=1/2$ and squeezing $r=r_1=r_3=1$.}
\label{UNDYN2}
\end{figure}

This model is the many-body generalization of a two-qubit system studied in \cite{Benatti3}, 
where entanglement between the two qubits was shown to occur through a purely mixing mechanism
induced by the presence of off-diagonal contributions of the form 
$(\sigma_\mu\otimes\bold{1})\,X\,(\bold{1}\otimes\sigma_\nu)$ in the dissipative generator.
In fact, the entangling power of the model depends entirely on the strength 
of the statistical coupling of the otherwise independent qubits. 

Similarly, in Model 2, mesoscopic entanglement can be dissipatively generated among the two chains 
in the large $N$ limit. Unfortunately, in this case
manageable analytic expressions for the logarithmic negativity are not available, so that the behaviour
of $E$ can be studied only numerically. For simplicity, in the following discussion
we have further set $\eta=1$, since this parameter can be reabsorbed into a redefinition
of the temperature. 

As in Model 1, some initial squeezing is necessary in order for the dynamics
to generate entanglement; further, the amount of created entanglement decreases as the
dissipative parameter $\xi$ entering the Kossakowski matrix (\ref{Kossmat2}) gets larger.
This is explicitly shown by the behaviour of the graphs in Fig.5 and Fig.6,
where the phenomena of sudden birth and sudden death of entanglement are also visible as in Model~1.
These graphs (and the ones below) refer to the choice of a symmetrically squeezed initial state; 
similar results hold also in the case of one-mode squeezed initial states.

The dependence on the initial state temperature $T$ is instead depicted in Fig.7, for fixed $\xi$
and squeezing parameter. Also in this case, one sees that
increasing the temperature, the maximum of the logarithmic negativity $E$ decreases, 
indicating that there exists a critical temperature $T_C$, above which no entanglement is possible;
the behaviour of $T_C$ as function of the squeezing parameter $r$ is given
by phase diagrams very similar to those in Fig.4.

\begin{figure}[hbtp]
\center\includegraphics[scale=0.6]{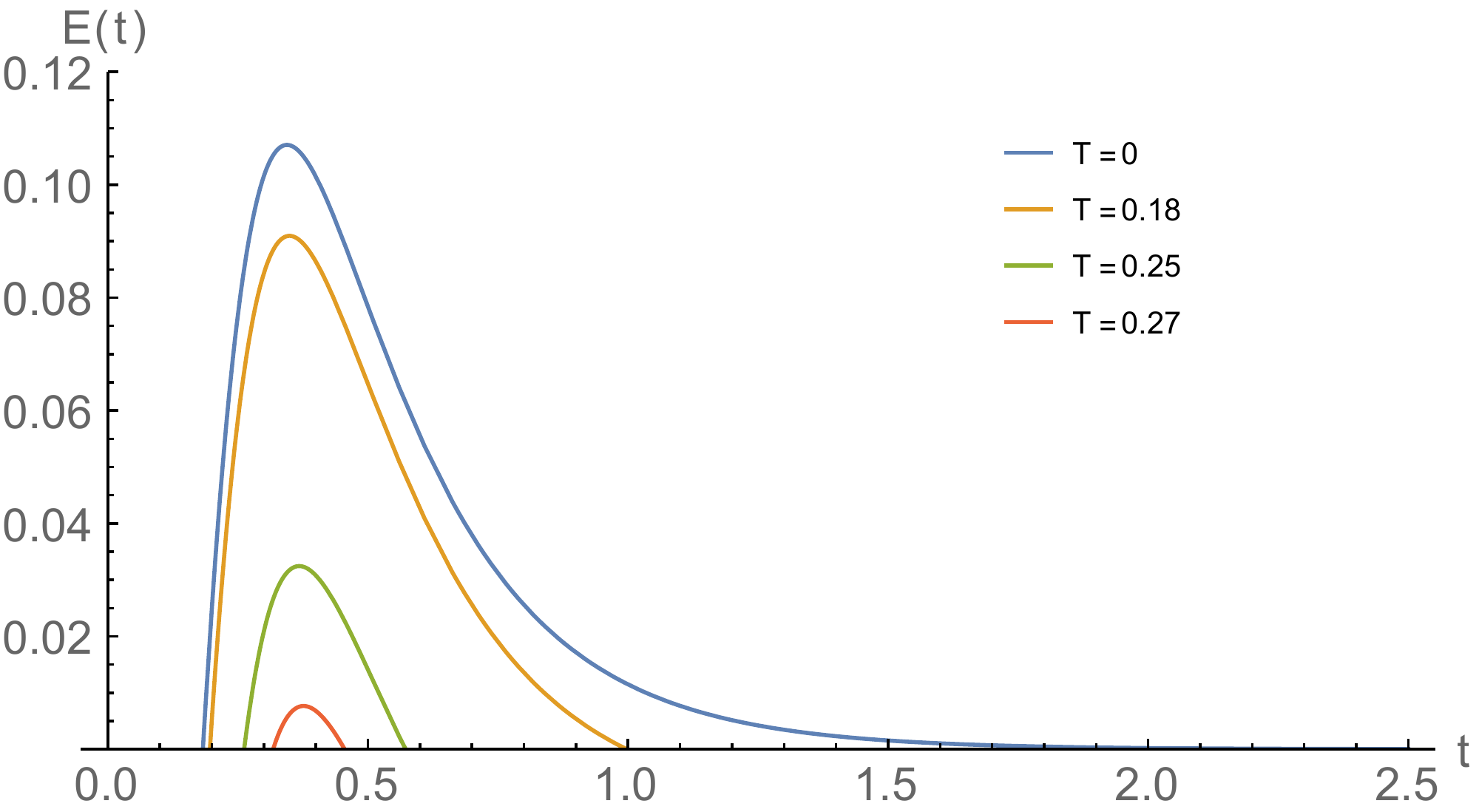}
\caption{\small Model 2: behaviour in time of the logarithmic negativity $E$ for different values of the temperature $T$, for $\xi=1$ and squeezing $r=r_1=r_3=1$.}
\end{figure}

However, unlike in Model 1, asymptotic entanglement is now possible.
Indeed, setting the parameter $\xi=0$ and decreasing the initial temperature $T$, one sees that the two chains not only get mesoscopically entangled at finite time, but remarkably, the generated mesoscopic entanglement persists for longer times.
This behaviour is clearly shown by the plots in Fig.8, where the time behaviour of the logarithmic negativity is reported
for a symmetrically squeezed initial state: in the case of zero temperature, one sees
that the generated mesoscopic entanglement persists for arbitrary long times.

\begin{figure}[hbtp]
\center
\includegraphics[scale=0.5]{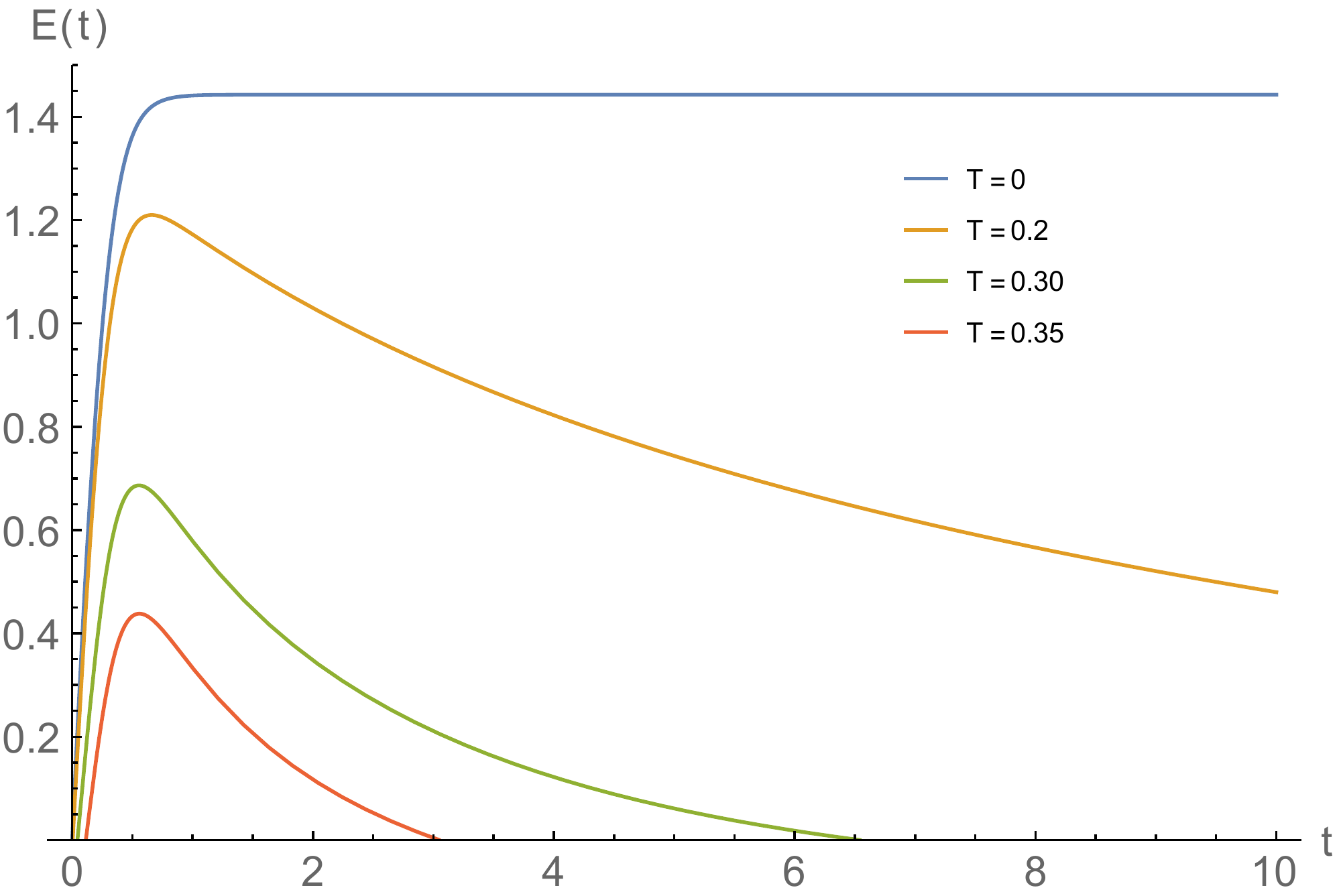}
\caption{\small Model 2: behaviour in time of the logarithmic negativity $E$ for different values of the temperature $T$, for $\xi=0$ and squeezing $r=r_1=r_3=1$.}
\end{figure}


\section{Outlook}

When dealing with many-body systems, {\it i.e.} systems with a very large number $N$
of elementary constituents, accessible observables are global, collective ones,
involving the degrees of freedom of all its parts. Typical examples are the mean-field
observables, defined as the algebraic mean of single particle observables,
as in the case of mean magnetization for spin systems. These quantities scale as $1/N$
and can be seen to behave as classical observables in the thermodynamical limit,
{\it i.e.} as the number of constituents becomes very large.

Similarly, fluctuation operators, defined in analogy with classical stochastic theory
as deviations from the mean, form another class of collective operators; however, they
scale as $1/\sqrt{N}$, and, because of this, they retain some quantum properties as $N$ increases.
Indeed, irrespective of the nature of the microscopic many-body system, 
the algebra they form turns out to be non-commutative and always of
bosonic type: they can be used to probe at the mesoscopic level 
the quantum properties of the system.

We have studied the quantum dynamics of the fluctuation operators in a many-body system
composed by two, non-interacting spin-1/2 chains, immersed in a common, weakly coupled external environment.
The system behaves as an open quantum systems, so that noise and dissipation
are expected to occur. Nevertheless, even in the thermodynamical limit,
these phenomena are not able to spoil the quantum character of suitable chosen, two-chain fluctuation operators.
Actually, despite the decohering and mixing-enhancing effects usually induced by the presence of the environment,
the two chains can get entangled by the emergent, open mesoscopic dynamics, through a purely
dissipative mechanism. 

We have studied in details the fate of the generated entanglement in the course of time 
and of its dependence on the strength of the coupling with the environment 
and on the temperature of the starting microscopic many-body state:
despite its inevitable dissipative action, the environment can nevertheless sustain non vanishing 
quantum correlations among the two chains even for very large times, provided the temperature
of the initial state is sufficiently low.

The mechanism of environment induced entanglement generation has been previously
known only for systems involving few qubits or oscillator modes;
our discussion shows that this phenomenon is at work also in the case of many-body systems
provided suitable mesoscopic observables are considered.
This result is general 
and can find direct applications in all instances
where mesoscopic, coherent quantum behaviours are expected to emerge, 
{\it e.g.} in experiments involving
spin-like and optomechanical systems, or ultra-cold gases trapped in optical lattices:
the possibility of entangling these many-body systems through a purely mixing mechanism
may reinforce their use for the actual realization of quantum information and communication protocols.

\section{Appendix A}

The relation \eqref{macro1} can be proved as follows: because of definition \eqref{macro}, it is equivalent to 
$$
\lim_{N\to\infty}\omega\bigg(a^\dag\,\big(X_N-\omega(x)\big)\big(Y_N-\omega(y)\big)\,b\bigg)=0
$$
for all $a,b\in\ca$. Set
$$
\widetilde{X}_N=\frac{1}{N}\sum_{k=0}^{N-1}\underbrace{\bigg(x^{(k)}-\omega(x)\bigg)}_{\widetilde{x}^{(k)}}\ ,\quad \widetilde{Y}_N=\frac{1}{N}\sum_{k=0}^{N-1}\underbrace{\bigg(y^{(k)}-\omega(y)\bigg)}_{\widetilde{y}^{(k)}}\ ,
$$
so that $\omega(\widetilde{x}^{(k)})=\omega(\widetilde{x})=0$, $\omega\left(\widetilde{X}_N\right)=0$ and similarly for $\widetilde{y}$, $\widetilde{Y}_N$.
Then, as shown in the main text for a single variable, the quasi-locality of $a,b$ and the clustering properties of the state yield:
$$
\lim_{N\to\infty}\omega\bigg(a^\dag\,\big(X_N-\omega(x)\big)\big(Y_N-\omega(y)\big)\,b\bigg)=\omega(a^\dag b)\lim_{N\to\infty}\omega\bigg(\widetilde{X}_N\widetilde{Y}_N\bigg)\ .
$$
Further, one can write:
$$
\omega\bigg(\widetilde{X}_N\widetilde{Y}_N\bigg)=\frac{1}{N^2}\sum_{k=0}^{N-1}
\omega\bigg(\widetilde{x}^{(k)}\widetilde{y}^{(k)}\bigg)\,+\,
\frac{1}{N^2}\sum_{k\neq\ell=0}^{N-1}
\omega\bigg(\widetilde{x}^{(k)}\widetilde{y}^{(\ell)}\bigg)\ .
$$
Since $\omega$ is translation-invariant, the first term vanishes as $\omega\big(\widetilde{x}\widetilde{y}\big)/N$ when $N\to\infty$.
Moreover, thank to the clustering property (\ref{clustates}), for any small $\epsilon>0$, 
there exists an integer $N_\epsilon$,
such that for $|k-\ell|^2>N_\epsilon$ one has:
$$
\left|\omega(\big(\widetilde{x}^{(k)}\widetilde{y}^{(\ell)}\big)-\omega(\widetilde{x})\,\omega(\widetilde{y})\right|=
\left|\omega(\big(\widetilde{x}^{(k)}\widetilde{y}^{(\ell)}\big)\right|\leq\epsilon\ .
$$
Then, using this result, one can finally write:
\beann
\left|\frac{1}{N^2}\sum_{k\neq\ell=0}^{N-1}
\omega\bigg(\widetilde{x}^{(k)}\widetilde{y}^{(\ell)}\bigg)\right|&\leq&
\frac{1}{N^2}\sum_{0<|k-\ell|\leq N_\epsilon}\left|\omega\bigg(\widetilde{x}^{(k)}\widetilde{y}^{(\ell)}\bigg)\right|\\
&+&\frac{1}{N^2}\sum_{|k-\ell|> N_\epsilon}\left|\omega\bigg(\widetilde{x}^{(k)}\widetilde{y}^{(\ell)}\bigg)\right|\\
&\leq&4\frac{2N_\epsilon+1}{N}\,\|x\|\,\|y\|\,+\,\epsilon\ ,
\eeann
so that, in the large $N$ limit, the relation \eqref{macro1} is indeed satisfied.
Notice that \eqref{macro1} entails that, in the GNS representation,
\bea
\nonumber
&&
\lim_{N\to\infty}\omega\bigg(a^\dag \big(X-\omega(x)\big)^\dag\,\big(X-\omega(x)\big)\,a\bigg)=\\
\label{macro2}
&&\hskip 2cm
=\lim_{N\to\infty}\left\|\pi_\omega\big(X-\omega(x)\big)\vert\Psi_a\rangle\right\|^2=0\ ,
\eea
for all $a\in\ca$. Namely, mean-field spin observables converge to their expectations with respect to $\omega$ in the strong operator topology on the GNS Hilbert space $\mathbb{H}_\omega$.

\section{Appendix B}

In this Appendix we collect the explicit expressions of various matrices that have been used in the main text;
these results are obtained from the corresponding multiple tensor product expressions by multiplying each matrix by the entries of the matrix which precedes it.

The correlation matrix $C^{(\beta)}$ in \eqref{modcorrmat2} then reads:
\be
C^{(\beta)}=\left(\bold{1}-\epsilon\,\sigma_1\right)\otimes\bold{1}\otimes\left(\bold{1}+
\epsilon\,\sigma_2\right)\\
=\begin{pmatrix}
\phantom{-}{\cal C}_\epsilon &-\epsilon\, {\cal C}_\epsilon\\
-\epsilon\, {\cal C}_\epsilon & \phantom{-}{\cal C}_\epsilon
\end{pmatrix}\ ,
\label{COMM2}
\ee
with
$$
{\cal C}_\epsilon=
\begin{pmatrix}
1&-i\epsilon&0&0\\
i\epsilon&1&0&0\\
0&0&1&-i\epsilon\\
0&0&i\epsilon&1
\end{pmatrix}\ .
$$
The symplectic matrix in \eqref{COMM1} and it inverse in \eqref{invCOMM1} are represented by:
\bea
&&\sigma^{(\beta)}=2i\epsilon(\epsilon\sigma_1-\bold{1})\otimes\bold{1}\otimes\sigma_2=
2\epsilon\begin{pmatrix}
\phantom{-}{\cal S} & -\epsilon\, {\cal S}\\
 -\epsilon\, {\cal S} & \phantom{-}{\cal S}
\end{pmatrix}\ ,\\
&&\big(\sigma^{(\beta)}\big)^{-1}=\frac{i}{2\epsilon\,c^2}(\bold{1}+\epsilon\sigma_1)\otimes\bold{1}\otimes\sigma_2
=-\frac{1}{2\epsilon\,c^2}
\begin{pmatrix}
{\cal S} & \epsilon\, {\cal S}\\
 \epsilon\, {\cal S} & {\cal S}
\end{pmatrix}\ ,
\eea
where
\be
{\cal S}=
\begin{pmatrix}
0&-1&0&0\\
1&0&0&0\\
0&0&0&-1\\
0&0&1&0
\end{pmatrix}\ ,
\label{S}
\ee
and $c=\sqrt{1-\epsilon^2}$, while the covariance matrix in \eqref{covmat2} is given by:
\be
\Sigma^{(\beta)}=(1-\epsilon\sigma_1)\otimes {\bf 1}\otimes{\bf 1}=
\begin{pmatrix}
\phantom{-}{\bf 1}_4 & -\epsilon\, {\bf 1}_4\\
-\epsilon\, {\bf 1}_4 & \phantom{-}{\bf 1}_4
\end{pmatrix}\ ,
\ee
with ${\bf 1}_4$ the unit matrix in four dimensions.
Furthermore, the matrix $\mathcal{M}$  in \eqref{matrix1} reads
\be
\label{matrix1.1}
\mathcal{M}=\sqrt{\epsilon}\,
\begin{pmatrix}
{\cal K} & {\cal K}^*\\
{\cal Q}^* & {\cal Q}
\end{pmatrix}\ ,
\ee
with
$$
{\cal K}=
\begin{pmatrix}
1&0&0&0\\
i&0&0&0\\
0&0&1&0\\
0&0&i&0\\
\end{pmatrix}\ ,\qquad
{\cal Q}=
\begin{pmatrix}
-\epsilon&c&0&0\\
i\epsilon&-ic&0&0\\
0&0&-\epsilon&c\\
0&0&i\epsilon&-ic\\
\end{pmatrix}\ ,
$$
while its inverse is given by
\be
\label{matrix2.1}
\mathcal{M}^{-1}=\frac{1}{2c\sqrt{\epsilon}}\,
\begin{pmatrix}
{\cal W} & {\cal Z}^*\\
{\cal W}^* & {\cal Z}
\end{pmatrix}\ ,
\ee
with
$$
{\cal W}=
\begin{pmatrix}
c&-ic&0&0\\
\epsilon&-i\epsilon&0&0\\
0&0&c&-ic\\
0&0&\epsilon&-i\epsilon\\
\end{pmatrix}\ ,\qquad
{\cal Z}=
\begin{pmatrix}
0&0&0&0\\
1&i&0&0\\
0&0&0&0\\
0&0&1&i\\
\end{pmatrix}\ .
$$
Finally, the $\mathcal{P}$ in \eqref{lemm2} is explicitly given by
\be
\label{lastmat}
\mathcal{P}=
\begin{pmatrix}
{\cal P}_{11} & {\cal P}_{12}\\
{\cal P}_{21} & {\cal P}_{22}\\
\end{pmatrix}\ ,
\ee
with
$$
{\cal P}_{11}=
\begin{pmatrix}
1&0&0&0\\
0&0&1&0\\
0&0&0&0\\
0&0&0&0\\
\end{pmatrix}\ ,\qquad
{\cal P}_{12}=
\begin{pmatrix}
0&0&0&0\\
0&0&0&0\\
1&0&0&0\\
0&0&1&0\\
\end{pmatrix}\ ,
$$
$$
{\cal P}_{21}=
\begin{pmatrix}
0&1&0&0\\
0&0&0&1\\
0&0&0&0\\
0&0&0&0\\
\end{pmatrix}\ ,\qquad
{\cal P}_{22}=
\begin{pmatrix}
0&0&0&0\\
0&0&0&0\\
0&1&0&0\\
0&0&0&1
\end{pmatrix}\ .
$$

\section{Appendix C}

We shall prove that, given a time-dependent Hermitean matrix $M_t$ and its exponential 
$\displaystyle N_t={\rm e}^{iM_t}$, then
\be
\label{appb1}
\dot{N_t}:=\frac{{\rm d}N_t}{{\rm d}t}=O_t\,N_t\ ,\quad O_t:=\sum_{k=1}^\infty\frac{i^k}{k!}\KK^{k-1}_{M_t}(\dot{M}_t)\ ,
\ee
where 
$$
\KK^n_A(B):=\Big[A\,,\,\KK^{n-1}_A(B)\Big]\ ,\quad \KK^0_A(B)=B\ .
$$

Indeed, given matrices $A$ and $B$, one has
$$
{\rm e}^{iA}\,B\,{\rm e}^{-iA}=\sum_{n=0}^\infty\frac{i^n}{n!}\underbrace{\Big[A\,\Big[A\,,\cdots\Big[}_{n\ times}B\,,\,A\Big]\cdots\Big]\Big]=\sum_{n=0}^\infty\frac{i^n}{n!}\,\KK^n_A(B)\ .
$$
Then, $[N_t\,,\,M_t]=0$ and $N_tN_t^\dag=N_t^\dag N_t=1$ imply $N_tM_tN^\dag_t=M_t$ and
$\displaystyle \dot{N}_t\,N^\dag_t=-N_t\,\dot{N}^\dag_t$. Therefore,
$$
N_t\,\dot{M}_t\,N^\dag_t\,-\,\dot{M}_t\,=-\,\dot{N}_t\,M_t\,N^\dag_t\,-\,N_t\,M_t\,\dot{N}^\dag_t=\Big[M_t\,,\,\dot{N}_t\Big]\,N_t^\dag\ .
$$
Furthermore, since, for $n\geq 1$,
$\displaystyle \mathbb{K}_A^n[B]=\Big[A\,,\,\mathbb{K}^{n-1}_A[B]\Big]$, it follows that
$$
\hskip -.5cm
N_t\,\dot{M}_t\,N^\dag_t-\dot{M}_t=\sum_{n=1}^\infty\frac{i^n}{n!}\mathbb{K}^n_{M_t}[\dot{M}_t]=\Big[M_t\,,\,O_t\Big]=\Big[M_t\,,\,\dot{N}_t\Big]\,N_t^\dag\ ,
$$
where $O_t=\sum_{k=1}^\infty\frac{i^k}{k!}\mathbb{K}_{M_t}^{k-1}[\dot{M}_t]$. 
Then, using again that $[N_t\,,\,M_t]=0$, one obtains
$$
\Big[M_t\,,\,O_t\,N_t\Big]=\Big[M_t\,,\,\dot{N}_t\Big]\ .
$$
In order to show that $\dot{N}_t=O_tN_t$, consider the orthogonal eigenvectors $\vert m_a(t)\rangle$ of $M_t$ with eigenvalues $m_a(t)$. Then, if $m_a(t)\neq m_b(t)$, the previous equality yields 
$$
\langle m_a(t)\vert O_tN_t\vert m_b(t)\rangle=\langle m_a(t)\vert\dot{N}_t\vert m_b(t)\rangle\ .
$$
On the other hand if $\vert m_a(t)\rangle$ and $\vert m_b(t)\rangle$ correspond to a same (real) eigenvalue $m(t)$, then one uses that
$$
0=\frac{{\rm d}}{{\rm d}t}\Big(\langle m_a(t)\vert m_b(t)\rangle\Big)=\langle \dot{m}_a(t)\vert m_b(t)\rangle\,+\,
\langle m_a(t)\vert\dot{m}_b(t)\rangle\ ,
$$
to deduce that also in such a case
\begin{eqnarray*}
\langle m_a(t)\vert O_t\,N_t\vert m_b(t)\rangle&=&i\,\langle m_a(t)\vert \dot{M}_t\vert m_b(t)\rangle\, {\rm e}^{im(t)}\, 
\delta_{ab}
=i\dot{m}(t)\,{\rm e}^{im(t)}\,\delta_{ab}\\ 
&=&\langle m_a(t)\vert\dot{N}_t\vert m_b(t)\rangle\ .
\end{eqnarray*}

\section{Appendix D}

In Model 1, the dynamics is generated by a Lindblad operator $\LL_N[X]=\HH_N[X]+\DD_N[X]$,
with hamiltonian part $\HH_N$ as in (\ref{mod1H}) and dissipative part $\DD_N$ given 
by~\eqref{mod1-dissip} with Kraus operators as in~\eqref{Krops}. 
When acting on the self-adjoint element $x_i^{(k)}\in\chi$ at site $k$, it reduces to:
\beann
\LL_N\left[x_i^{(k)}\right]&=&i\frac{\eta}{2}\left[\sigma_3^{(k)}\otimes\bold{1}+\bold{1}\otimes \sigma_3^{(k)}\,,\,x_i^{(k)}\right]\\
&+&\frac{J_0}{2}\sum_{\mu,\nu=1}^4 D_{\mu\nu}\left[\left[v_\mu^{(k)}\,,\,x_i^{(k)}\right]\,,\,(v^\dag_{\nu})^{(k)}\right]\ .
\eeann
One can recast the first term as:
$$
i\frac{\eta}{2}\left[\sigma_3^{(k)}\otimes\bold{1}+\bold{1}\otimes \sigma_3^{(k)}\,,\,x_i^{(k)}\right]=\sum_{j=1}^8\mathcal{H}_{ij}\,x_j^{(k)}\ ,\quad
\mathcal{H}=-i\eta\begin{pmatrix}
\sigma_2&0&0&0\cr
0&\sigma_2&0&0\cr
0&0&\sigma_2&0\cr
0&0&0&\sigma_2
\end{pmatrix}\ .
$$
Further, let $\left[v_\mu^{(k)}\,,\,x_i^{(k)}\right]=\sum_{j=1}^8\mathcal{V}_\mu^{ij}x_j^{(k)}$; then, the dissipative term reads
$$
\DD_N\left[x_i^{(k)}\right]=\sum_{k=1}^8\mathcal{D}_{ik}\,x_k^{(k)}\ ,\quad
\mathcal{D}_{ik}=J_0\sum_{\mu,\nu=1}^4\frac{D_{\mu\nu}}{2}(\mathcal{V}_\mu \mathcal{V}^*_\nu)^{ik}\ .
$$
The four $8\times 8$ matrices $\mathcal{V}_\mu$ explicitly read
\begin{eqnarray*}
\mathcal{V}_1&=&\frac{1}{2}\begin{pmatrix}
0&0&0&\bold{1}+\sigma_2\cr
0&0&\sigma_2-\bold{1}&0\cr
0&\bold{1}+\sigma_2&0&0\cr
\sigma_2-\bold{1}&0&0&0
\end{pmatrix}=
-\mathcal{V}^*_2\\
\mathcal{V}_3&=&-\begin{pmatrix}
\sigma_2&0&0&0\cr
0&0&0&0\cr
0&0&\sigma_2&0\cr
0&0&0&0
\end{pmatrix}\ ,\quad 
\mathcal{V}_4=-\begin{pmatrix}
0&0&0&0\cr
0&\sigma_2&0&0\cr
0&0&0&0\cr
0&0&0&\sigma_2
\end{pmatrix}\ .
\end{eqnarray*}
In order to make computations easier, it proves convenient to write these matrices as (sums of) $3$-fold tensor products of Pauli matrices:
\begin{eqnarray*}
\mathcal{V}_1&=&\frac{1}{2}\sigma_1\otimes\left(i\,\sigma_2\otimes\bold{1}+\sigma_1\otimes\sigma_2\right)=-\mathcal{V}^*_2\\
\mathcal{V}_3&=&-\bold{1}\otimes\left(\bold{1}+\sigma_3\right)\otimes\sigma_2\ ,\quad
\mathcal{V}_4=\bold{1}\otimes\left(\sigma_3-\bold{1}\right)\otimes\sigma_2\ .
\end{eqnarray*}
Similarly, $\mathcal{H}=-i\eta\bold{1}\otimes\bold{1}\otimes\sigma_2$, whence
$$
\mathcal{L}\equiv\mathcal{H}+\mathcal{D}=-i\eta\bold{1}\otimes\bold{1}\otimes\sigma_2-J_0\Big(\delta-\gamma\sigma_1\otimes\sigma_1\otimes\bold{1}\Big)\ .
$$
Explicitly, one has:
\be
\mathcal{H}=\eta
\begin{pmatrix}
{\cal S} & {\bf 0}_4\\
{\bf 0}_4 & {\cal S}\\
\end{pmatrix}\ ,\qquad
\mathcal{D}=
J_0\,\begin{pmatrix}
-\delta {\bf 1}_4 &\Gamma\\
\Gamma & -\delta {\bf 1}_4\\
\end{pmatrix}
\ee
where ${\cal S}$ is as in (\ref{S}) and ${\bf 0}_4$ is the null matrix in four dimensions, while
$$
\Gamma=\gamma\,
\begin{pmatrix}
0&0&1&0\\
0&0&0&1\\
1&0&0&0\\
0&1&0&0\\
\end{pmatrix}\ .
$$
The expressions of the $8\times 8$ matrices $H^{(1)}$ and $D^{(1)}$ in \eqref{Flind1a} and \eqref{Flind1b} that define the action of the mesoscopic dissipative generator in \eqref{fluctLind1}-\eqref{fluctLind2} can then be readily computed by expressing also the matrices $C^{(\beta)}$ and $(\sigma^{(\beta)})^{-1}$ as (sums of) $3$-fold tensor products of Pauli matrices, as given in \eqref{modcorrmat} and \eqref{invCOMM1}, respectively:
$$
C^{(\beta)}=\left(\bold{1}-\epsilon\,\sigma_1\right)\otimes\bold{1}\otimes\left(\bold{1}+
\epsilon\,\sigma_2\right)\ ,\quad
(\sigma^{(\beta)})^{-1}=\frac{1}{2c^2\epsilon}\,{\left(\bold{1}+\epsilon\sigma_1\right)\otimes\bold{1}\otimes\,i\sigma_2}\ ,
$$
where $\epsilon=\tanh(\eta\beta/2)$, $c^2=1-\epsilon^2$.
Then, one computes
\beann
&&\hskip -.8cm
\mathcal{L}\,C^{(\beta)}\,-\,C^{(\beta)}\,\mathcal{L}^{tr}=
-2i\eta\left(\bold{1}-\epsilon\sigma_1\right)\otimes\bold{1}\otimes\left(\epsilon+\sigma_2\right)\\
&&\hskip-.8cm
\mathcal{L}\,C^{(\beta)}\,+\,C^{(\beta)}\,\mathcal{L}^{tr}=-2J_0\Big(\delta\left(\bold{1}-\epsilon\sigma_1\right)\otimes\bold{1}\,-\,\gamma\left(\sigma_1-\epsilon\right)\otimes\sigma_1\Big)\otimes\left(\bold{1}+\epsilon\sigma_2\right)\ .
\eeann
From \eqref{Flind1a}, {\it i.e.}
$$
H^{(1)}=-i(\sigma^{(\beta)})^{-1}\left(\mathcal{L}C^{(\beta)}\,-\,C^{(\beta)}\mathcal{L}^{tr}\right)\,(\sigma^{(\beta)})^{-1}\ ,
$$
one derives that the Hamiltonian coupling among the $F(x_i)$ is given by
\be
\label{H1F}
H^{(1)}=\frac{\eta}{2c^2\epsilon^2}\left(\bold{1}+\epsilon\sigma_1\right)\otimes\bold{1}\otimes\left(\epsilon+\sigma_2\right)=
\begin{pmatrix}
{\cal E} & \epsilon\,{\cal E}\\
\epsilon\,{\cal E} & {\cal E}\\
\end{pmatrix}\ ,
\ee
with
$$
{\cal E}=
\begin{pmatrix}
\epsilon&-i&0&0\\
i&\epsilon&0&0\\
0&0&\epsilon&-i\\
0&0&i&\epsilon\\
\end{pmatrix}\ .
$$
Similarly, the hamiltonian contribution expressed in terms of creation and annihilation operators
in \eqref{Flind2a} gives rise to the matrix $H^{(2)}=\mathcal{M}^\dag\,H^{(1)}\,\mathcal{M}$, explicitly given by 
\be
\label{H1A}
H^{(2)}=
\frac{\eta}{\epsilon}
\begin{pmatrix}
(\epsilon+1) {\bf 1}_4 & {\bf 0}_4\\
{\bf 0}_4 & (\epsilon-1) {\bf 1}_4
\end{pmatrix}\ .
\ee
From \eqref{Flind1b}, {\it i.e.}
$$
D^{(1)}=(\sigma^{(\beta)})^{-1}\left(\mathcal{L}C^{(\beta)}\,+\,C^{(\beta)}\mathcal{L}^{tr}\right)(\sigma^{(\beta)})^{-1}\ ,
$$
one derives the Kossakowski matrix responsible for the dissipative action of the generator: 
\beann
D^{(1)}&=&\frac{J_0}{2c^2\epsilon^2}\Big(\delta\left(\bold{1}+\epsilon\sigma_1\right)\otimes\bold{1}-\gamma\left(\epsilon+\sigma_1\right)\otimes\sigma_1\Big)\otimes\left(\bold{1}+\epsilon\sigma_2\right)\\
&=&\frac{J_0}{2c^2\epsilon^2}
\begin{pmatrix}
D_1&\epsilon D_2&\epsilon D_1&D_2\cr
\epsilon D_2&D_1&D_2&\epsilon D_1\cr
\epsilon D_1&D_2&D_1&\epsilon D_2\cr
D_2&\epsilon D_1&\epsilon D_2&D_1\end{pmatrix}\ ,\\
D_1&=&\delta\begin{pmatrix}1&-i\epsilon\cr
i\epsilon&1\end{pmatrix}\ ,\qquad D_2=-\gamma\begin{pmatrix}1&-i\epsilon\cr
i\epsilon&1\end{pmatrix}\ .
\eeann
Instead, when the dissipative contribution is expressed in terms of creation and annihilation operators,
the corresponding Kossakowski matrix reads
\bea
\label{D1A1}
D^{(2)}&=&\mathcal{M}^\dag\,D\,\mathcal{M}=\frac{J_0}{\epsilon}
\begin{pmatrix}
D_{1+}&D_{2+}&0&0\cr
D_{2+}&D_{1+}&0&0\cr
0&0&D_{1-}&D_{2-}\cr
0&0&D_{2-}&D_{1-}
\end{pmatrix}\ ,\\
\label{D1A2}
D_{1\pm}&=&\delta(1\pm\epsilon)\begin{pmatrix}
1&0\cr0&1\end{pmatrix}\ ,\qquad
D_{2\pm}=\gamma(1\pm\epsilon)\begin{pmatrix}
\epsilon&-c\cr
-c&\epsilon\end{pmatrix}\ .
\eea

\section{Appendix E}

The Hamiltonian contribution to the Lindblad generator of the microscopic 
dynamics studied in Model 2 is the same as in Model 1, thus we concentrate on the dissipative term
$\DD_N$ of $\LL_N$. Since operators at different sites commute, the action of $\DD_N$ on an operator $x_i$ from the set $\chi$ at a given site $k$ is given by
\beann
\DD_N[x_i^{(k)}]&=&
\frac{1}{2}\Big(\left[w^{(k)}_1\,,\,\left[x^{(k)}_i\,,\,w^{(k)}_1\right]\right]+
\left[w_2^{(k)}\,,\,\left[x^{(k)}\,,\,w^{(k)}_2\right]\right]\\
&+&\gamma\left[w_3^{(k)}\,,\left[x^{(k)}\,,\,w^{(k)}_3\right]\right]\Big)\\
&-&i\frac{\epsilon}{2}\,\left\{w^{(k)}_1\,,\,\left[x^{(k)}_i\,,\,w^{(k)}_2\right]\right\}+i\frac{\epsilon}{2}\,\left\{w_2^{(k)}\,,\,\left[x^{(k)}_i\,,\,w^{(k)}_1\right]\right\}\ ,
\eeann
with $w_\mu=\sigma_\mu\otimes\bold{1}+\bold{1}\otimes\sigma_{\mu}$.
Then, by means of the Pauli algebraic relations, one explicitly computes that 
$$
\DD_N\left[x_i^{(p)}\right]=\sum_{k=1}^8\mathcal{D}_{ik}\,x_k^{(p)}\ ,
$$
where
\be
\label{mod2L}
\mathcal{D}=-2\begin{pmatrix}
1+\xi&0&0&0&0&0&-\epsilon&0\cr
0&1+\xi&0&0&0&0&0&-\epsilon\cr
0&0&1+\xi&0&-\epsilon&0&0&0\cr
0&0&0&1+\xi&0&-\epsilon&0&0\cr
2\epsilon&0&\epsilon&0&3+\xi&0&2&0\cr
0&2\epsilon&0&\epsilon&0&3+\xi&0&2\cr
\epsilon&0&2\epsilon&0&2&0&3+\xi&0\cr
0&\epsilon&0&2\epsilon&0&2&0&3+\xi\cr
\end{pmatrix}\ .
\ee
As in the previous Appendix, from \eqref{Flind1b}, with 
$$
C^{(\beta)}=\left(\bold{1}-\epsilon\,\sigma_1\right)\otimes\bold{1}\otimes\left(\bold{1}+
\epsilon\,\sigma_2\right)\ ,
$$ 
one computes
$$
D^{(1)}=(\sigma^{(\beta)})^{-1}\left(\mathcal{L}C^{(\beta)}\,+\,C^{(\beta)}\mathcal{L}^{tr}\right)(\sigma^{(\beta)})^{-1}\ .
$$
Then, by the transformation $D^{(2)}=\mathcal{M}^\dag\,D^{(1)}\,\mathcal{M}$ that maps the dissipator written in terms of the operators $F(x_i)$, $1\leq i\leq 8$, to the one expressed using annihilation and creation operators $a^\#_i$, $i=1,2,3,4$, one gets
\be
\label{mod2La1}
D^{(2)}=\frac{2}{\epsilon}\begin{pmatrix}(1+\epsilon)A&0\cr
0&(1-\epsilon)A
\end{pmatrix}\ ,
\ee
with
\be
\label{mod2La2} 
A=\begin{pmatrix}1+\xi&0&\epsilon^2&-\epsilon c\cr
0&3+\xi&-\epsilon c&-(1+c^2)\cr
\epsilon^2&-\epsilon c&1+\xi&0\cr
-\epsilon c&-(1+c^2)&0&3+\xi
\end{pmatrix}\ ,
\ee
where, as before, $\epsilon=\tanh(\eta\beta/2)$, $c^2=1-\epsilon^2$.

\section{Appendix F}

In this Appendix we derive the explicit form of the quantity 
$S$ appearing in the entanglement criterion of equation \eqref{crit} in Model 1, 
both in the case of an initial symmetrically squeezed state, $r_1=r_3=r$, 
and for a one-mode squeezed state, $r_1=r$, $r_3=0$.

The first step is to find the evolution of the reduced covariance matrix at every time $t$, 
in the language of creation and annihilation operators. From Appendix~D, {\sl Theorem~\ref{qfth}}, {\sl Lemma~\ref{lemma0}} and {\sl Lemma~\ref{lem2}}, one finds:
\begin{equation}
\Phi_t\left[D(z)\right]={\rm e}^{-\frac{1}{2}(\tilde Z, \tilde{\mathcal{Y}}_t \tilde Z)}\,D(z_t)
\end{equation}
with:
$$
\tilde {Z}_t= {\rm e}^{t\widetilde{\mathcal{L}}^{tr}} \tilde Z\ ,\qquad
{\rm e}^{t\widetilde{\mathcal{L}}^{tr}}=\mathcal{P}^T\Sigma_3\mathcal{M}^\dagger e^{t\mathcal{L}^{tr}}
(\mathcal{M}^{\dagger})^{-1}\Sigma_3\mathcal{P}\ ,\qquad
\widetilde{\mathcal{Y}}_t=\frac{1}{2\epsilon}\left({\bf 1}_{8}-\left({\rm e}^{t\widetilde{\mathcal{L}}^{tr}}\right)^\dagger {\rm e}^{t\widetilde{\mathcal{L}}^{tr}}\right)\ ,
$$
and
$$
{\rm e}^{t\widetilde{\mathcal{L}}^{tr}}=e^{-\delta J_0 t}\begin{pmatrix}
\cosh(J_0\gamma t)&0&-\epsilon\sinh(J_0\gamma t)&c\sinh(J_0\gamma t)\\
0&\cosh(J_0\gamma t)&c\sinh(J_0\gamma t)&\epsilon\sinh(J_0\gamma t)\\
-\epsilon\sinh(J_0\gamma t)&c\sinh(J_0\gamma t)&\cosh(J_0\gamma t)&0\\
c\sinh(J_0\gamma t)&\epsilon\sinh(J_0\gamma t)&0&\cosh(J_0\gamma t)
\end{pmatrix}\otimes \begin{pmatrix}
{\rm e}^{i\omega t}&0\\
0&{\rm e}^{-i\omega t}
\end{pmatrix}\ .
$$
As a result, the evolution of the covariance matrix for the four modes reads as follows:
$$
\widetilde G(t)=\left({\rm e}^{t\tilde{\mathcal{L}}^{tr}}\right)^\dagger\widetilde \Sigma^{(\beta)}_{r_1,r_3}\,{\rm e}^{t\tilde{\mathcal{L}}^{tr}}+\widetilde \Sigma^{(\beta)}_{0,0}-\left({\rm e}^{t\tilde{\mathcal{L}}^{tr}}\right)^\dagger\widetilde \Sigma^{(\beta)}_{0,0}\, {\rm e}^{t\tilde{\mathcal{L}}^{tr}}\ .
$$
In order to construct the reduced matrix for the two relevant modes under investigation, it is sufficient to look at the block structure of formula \eqref{bigcov2} and to collect the corresponding entries:
$$
\widetilde G_{red}(t)=\begin{pmatrix}
\widetilde G_{11}(t)&\widetilde G_{13}(t)\\
\widetilde G_{13}(t)& \widetilde G_{33}(t)
\end{pmatrix}\ ,
$$
where one has $\widetilde G_{13}=(\widetilde G_{13})^\dagger$. 
This allows evaluating the four quantities $I_j$ entering the definition of $S$ in
(\ref{crit}). As already mentioned, two cases have been considered for the initial state, 
a symmetrically squeezed state and a one-mode squeezed state. In the two cases, one obtains,
respectively:
\begin{eqnarray}
\nonumber
S_S(t)&=&\frac{\left(\epsilon^2-1\right)^2}{16\epsilon^4}+\sinh^2(r)\left[\left(\frac{1}{2\epsilon^2}-\frac{1}{2}\right)\left(\frac{y_\epsilon(t)}{\epsilon}-y_\epsilon^2(t)\right)-2\left(1+\frac{1}{\epsilon^2}\right)y_3^2(t)\right]+\\
\label{critsas0}
&+&\sinh^4(r)\left[\left(\frac{y_\epsilon(t)}{\epsilon}-y_\epsilon^2(t)+4y_3^2(t)\right)^2-4\frac{y_3^2(t)}{\epsilon^2}\right]\ ,
\\
\nonumber
S_A(t)&=&\frac{\left(\epsilon^2-1\right)^2}{16\epsilon^4}+\sinh^2(r)\Bigg[\left(\frac{1}{4\epsilon^2}-\frac{1}{4}\right)\left(\frac{y_1(t)-y_1^2(t)}{\epsilon^2}+y_2(t)-\epsilon^2 y_2^2(t)\right)+\\
&-&y_3^2(t)\left(\frac{1}{2}+\frac{1}{2\epsilon^2}\right)\Bigg]\ ,
\label{critsas}
\end{eqnarray} 
where
\begin{eqnarray}
\label{y1}
y_1(t)&=&\frac{e^{-2J_0\delta t}}{2}\left(\cosh(2J_0\gamma t)+1\right)\ ,\qquad
y_2(t)=\frac{e^{-2J_0\delta t}}{2}\left(\cosh(2J_0\gamma t)-1\right)\ ,\\
y_3(t)&=&\frac{e^{-2J_0\delta t}}{2}\sinh(2J_0\gamma t)\ ,\hskip 1.9cm
y_\epsilon(t)=\frac{y_1(t)}{\epsilon}+\epsilon y_2(t)\ .
\end{eqnarray}

\vfill\eject


\begin{thebibliography}{50}
\bibitem{note}
A preliminary, partial exposition of these results has been presented in \cite{BCF1},
without proofs and for a very special model. In this paper instead, all proofs are explicitly given for a general dissipative dynamics, together with a detailed discussion of the characteristics of the environment generated entanglement.
\end{thebibliography}
\end{document}